\documentclass[aps,pra,twocolumn,nofootinbib,floatfix,longbibliography,superscriptaddress]{revtex4-2}

\usepackage{amsmath,amssymb,amsthm,bbm,bm,color,dsfont,float,graphicx,hyperref,makecell,mathrsfs,mathtools,nicefrac,pgfplots,physics,tikz,times,txfonts}
\usepackage[normalem]{ulem}  

\definecolor{equationcolor}{RGB}{222,94,100}
\definecolor{alecolor}{RGB}{238,33,80}

\hypersetup{urlcolor=equationcolor, colorlinks=true,citecolor=equationcolor,linkcolor=equationcolor} 
\pgfplotsset{compat=1.18} 

\DeclareFontFamily{U}{mathb}{\hyphenchar\font45}
\DeclareFontShape{U}{mathb}{m}{n}{
	<-6> mathb5 <6-7> mathb6 <7-8> mathb7
	<8-9> mathb8 <9-10> mathb9
	<10-12> mathb10 <12-> mathb12
}{}
\DeclareSymbolFont{mathb}{U}{mathb}{m}{n}
\DeclareMathSymbol{\ggcurly}{\mathrel}{mathb}{"CF}

\def\E{ {\cal E} }
\def\T{ {\cal T} }
\def\P{ {\cal P} } 

\newcommand{\iden}{\mathbbm{1}}
\renewcommand{\v}[1]{\ensuremath{\boldsymbol #1}}

\makeatletter
\def\blfootnote{\gdef\@thefnmark{}\@footnotetext}
\makeatother

\theoremstyle{plain}
\newtheorem{thm}{Theorem}
\newtheorem{conj}{Conjecture}
\newtheorem{lem}[thm]{Lemma}
\newtheorem{cor}[thm]{Corollary}
\newtheorem{defn}{Definition}

\newlength\myindent
\begin{document}
	
	\title{Thermal recall: Memory-assisted Markovian thermal processes}
	\date{\today}
	
	\author{Jakub Czartowski$^*$}
	\affiliation{Doctoral School of Exact and Natural Sciences, Jagiellonian University, 30-348 Kraków, Poland}
	\affiliation{Faculty of Physics, Astronomy and Applied Computer Science, Jagiellonian University, 30-348 Kraków, Poland.}
	
	\author{A. de Oliveira Junior$^*$}
	\author{Kamil Korzekwa}
	\affiliation{Faculty of Physics, Astronomy and Applied Computer Science, Jagiellonian University, 30-348 Kraków, Poland.}
	
	\begin{abstract}
		
		We develop a resource-theoretic framework that allows one to bridge the gap between two approaches to quantum thermodynamics based on Markovian thermal processes (which model memoryless dynamics) and thermal operations (which model arbitrarily non-Markovian dynamics). Our approach is built on the notion of memory-assisted Markovian thermal processes, where memoryless thermodynamic processes are promoted to non-Markovianity by explicitly modelling ancillary memory systems initialised in thermal equilibrium states. Within this setting, we propose a family of protocols composed of sequences of elementary two-level thermalisations that approximate all transitions between energy-incoherent states accessible via thermal operations. We prove that, as the size of the memory increases, these approximations become arbitrarily good for all transitions in the infinite temperature limit, and for a subset of transitions in the finite temperature regime. Furthermore, we present solid numerical evidence for the convergence of our protocol to any transition at finite temperatures. We also explain how our framework can be used to quantify the role played by memory effects in thermodynamic protocols such as work extraction. Finally, our results show that elementary control over two energy levels at a given time is sufficient to generate all energy-incoherent transitions accessible via thermal operations if one allows for ancillary thermal systems.
	\end{abstract}
	
	\maketitle
	
	
	\section{Introduction}
	\blfootnote{$ ^*$\hspace{1pt}These authors contributed equally to this work.}
	
	Information has become ubiquitous in thermodynamics. It all started with Maxwell's seminal inquiry~\cite{maxwell1872theory}: \emph{what would happen if we had knowledge of a system's state?} The ramifications of this hypothesis led to potential violations of the second law of thermodynamics and a century-long puzzle~\cite{Szilard1929, brillouin1951maxwell}. Ultimately, it was found that thermodynamics imposes physical restrictions on information processing~\cite{Landauer1961, Bennett1982}, resulting in the development of frameworks devoted to incorporating information into thermodynamics~\cite{maruyama2009colloquium,seifert2012stochastic,sagawa2012thermodynamics,Goold2016,binder2018thermodynamics, Deffner2019}. A crucial concept at the intersection between these two fields is \emph{memory}, a thermodynamic resource for storing, processing, and erasing information. In particular, memory effects can bring numerous advantages, including enhanced cooling~\cite{taranto2020exponential}, generation of entanglement~\cite{mirkin2019entangling,mirkin2019information} or improved performance of heat engines and refrigerators~\cite{PhysRevA.99.052106,PhysRevA.102.012217,PhysRevE.106.014114}. However, realistic quantum mechanical systems are open and governed by non-unitary time evolution, which encompasses the irreversible phenomena such as energy dissipation, relaxation to thermal equilibrium or stationary non-equilibrium states, and the decay of correlations~\cite{breuer2002theory,rivas2012open}. Hence, assumptions like weak coupling, large bath size, and fast decaying correlations are commonly made in modelling such systems, thus neglecting memory effects. This raises the question of how memoryless processes get modified when system-bath memory effects become non-negligible, i.e., how to assess and quantify the role of memory in thermodynamic processes~\cite{rivas2014quantum}. 
	
	The resource theory of thermodynamics~\cite{Janzing2000,horodecki2013fundamental, brandao2015second,Lostaglio2019} is 
	a relatively recent framework allowing one to address foundational questions in thermodynamics. By relying on the notion of thermal operations~\cite{Janzing2000,horodecki2013fundamental}, a set of transformations that can be carried out without an external source of work or coherence, it offers a complete set of laws for characterising general state transformations under thermodynamic constraints. The downsides of this formalism are twofold. Firstly, it focuses only on snapshots of the evolution, making it hard to discuss how the processes are realised in time. Secondly, it may require precise control over the system and the bath. The first problem was addressed by developing a hybrid framework that reconciles resource theory and master equation approaches~\cite{lostaglio2021continuous,korzekwa2022}, where the concept of a Markovian thermal process was introduced. This new set of operations refines the notion of thermal operations by encoding all relevant constraints of a Markovian evolution. The second problem was partially addressed in Ref.~\cite{Lostaglio2018elementarythermal} by introducing the concept of elementary thermal operations, i.e., a subset of transformations that can be decomposed into a series of thermal operations, each acting only on two energy levels of the system. Such decompositions offer a method to bypass the need for a complete control over interactions between the system and the environment, and the approach was recently generalised to also include catalytic transformations~\cite{Jeongrak2022}. While elementary operations require only a limited control, they still rely on non-Markovian effects, and so the question of quantifying memory effects in the resource theory of thermodynamics remains open.
	
\begin{figure*}
\centering
\includegraphics{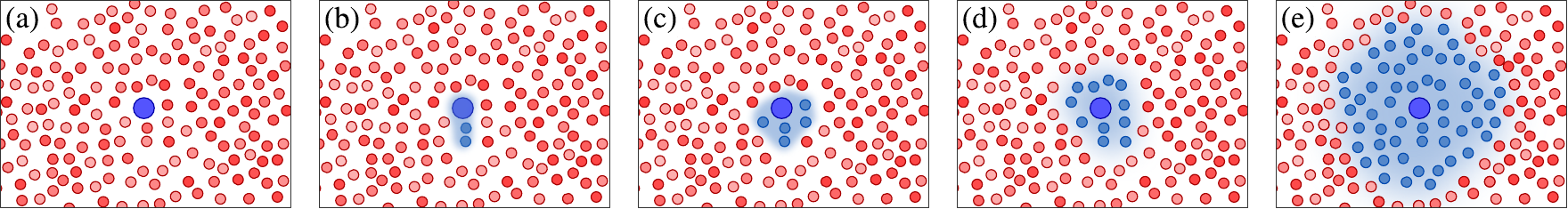}
\caption{\textbf{Memory-assisted Markovian thermal processes}. Schematic representation of the general setting. (a) Initially, the main system (large blue circle) is coupled to a heat bath at inverse temperature $\beta$ (small red circles) and their interaction is Markovian, so that the bath is in thermal equilibrium at each moment in time. (b)-(e) Then, the control is extended to parts of the environment (small blue circles with blue background) that do not instantaneously thermalise to equilibrium after interactions with the system, and can thus lead to non-Markovian dynamics of the main system.}
\label{Fig:schematic_representation.pdf}
\end{figure*}
	
	In this work, we make a step towards bridging the gap between thermal operations and Markovian thermal processes for energy-incoherent states by introducing and investigating \emph{memory-assisted Markovian thermal processes} (MeMTPs). These are defined by extending the Markovian thermal processes framework~\cite{lostaglio2021continuous,korzekwa2022} with ancillary memory systems, allowing one to interpolate between memoryless dynamics and the one with full control. More specifically, we demonstrate that energy-incoherent states achievable from a given initial state via thermal operations can be approached arbitrarily well by repeatedly interacting the main system with a memory that is initialised in a thermal equilibrium state (and therefore is thermodynamically resourceless) via an algorithmic procedure composed of Markovian thermal processes. Physically, this can be seen as a partial control over the bath degrees of freedom, where the bath can be thought of as a large, discrete, collection of smaller thermal units, and one can control the interactions of the main system with a small number of thermal subsystems~(see Fig.~\ref{Fig:schematic_representation.pdf}). 
	
	Following this idea, we introduce a family of memory-extended Markovian thermodynamic protocols that require minimal control and are valid for any temperature regime. In the infinite temperature limit, we prove that our protocol can arbitrarily well simulate any state transition that can be achieved via thermal operations. More precisely, our first main result states that when memory grows and the number of interactions goes to infinity, the full set of states achievable by thermal operations can be reached by MeMTPs. We also provide analytic expressions for the convergence rates, which scale either polynomially with the number of degrees of freedom or exponentially with the number of memory subsystems. Moreover, based on strong numerical evidence, we propose a conjecture for a more accurate approximation of arbitrary state transformation (i.e., converging faster with the growing size of the memory) through sequences of truncated versions of our protocol. Our second main result extends these considerations to the finite-temperature regime. Here, we first show analytic convergence of our MeMTP protocol to a particular subset of state transformations that can be achieved by thermal operations. These include all the so-called $\beta$-swaps, as well as $\beta$-cycles, which form a thermodynamic equivalent of cyclic permutations. Then, based on numerical simulations, we conjecture that actually an arbitrary state reachable via thermal operations can be obtained using a proper sequence of truncated protocols.
	
	With these results at hand, we then proceed to discussing their applicability. First, we explain how to assess the role played by memory in the performance of thermodynamic protocols by investigating work extraction in the intermediate regime of limited memory. We thus interpolate between the two extremes of no memory and complete control, and quantify environmental memory effects with the amount of extractable work from a given non-equilibrium state. Second, we consider the task of cooling a two-level system using a two-dimensional memory characterised by a non-trivial Hamiltonian. This example represents a minimal model requiring the manipulation of two two-level systems. Various experimental proposals are available across distinct platforms suitable for realising this specific model. Such platforms encompass quantum dots~\cite{PhysRevLett.110.256801,PhysRevB.98.045433}, superconducting circuits~\cite{PhysRevB.94.235420,chen2012quantum}, and atom-cavity systems~\cite{mitchison2016realising}. We then clarify that all transitions achievable via thermal operations can be performed using a subset of thermal operations that only affect two energy levels at the same time. This may seem to contradict the results of Refs.~\cite{Lostaglio2018elementarythermal,mazurek2018decomposability}, where it was proven that thermal operations constrained to just two energy levels of the system are not able to generate all thermodynamically allowed transitions. We resolve this apparent contradiction by noting that in our case we require the control over two levels of the joint system-memory state, and not just the system state. Finally, we discuss the behaviour of free energy during a non-Markovian evolution, explaining the role of the memory as a free energy storage. 
	
	The paper is structured as follows. First, in Sec.~\ref{sec:framework} we recall the frameworks of thermal operations and Markovian thermal processes. Next, in Sec.~\ref{sec:bridging}, we introduce the central notion of this paper, the memory-assisted Markovian thermal processes, and then describe the protocol that employs thermal memory states to approximate non-Markovian thermodynamic state transitions with Markovian thermal processes. We then explain how this approximation convergences to the full set of transitions achievable via thermal operations as the size of the memory grows. Sec.~\ref{sec:discussion} contains discussion and application of our results. Finally, in Sec.~\ref{sec:outlook}, we conclude and provide outlook for future research.
	
	
	\section{Framework}
	\label{sec:framework}
	
	In this work our aim is to investigate how memory effects affect thermodynamics of a finite-dimensional quantum system coupled to a heat bath at inverse temperature $\beta = 1/k_B T$ with $k_B$ denoting the Boltzmann constant. The investigated system of dimension $d$ is described by a Hamiltonian \mbox{$H = \sum_i E_i \ketbra{E_i}{E_i}$} and is prepared in an initial state $\rho$. The thermal environment, with a Hamiltonian $H_E$, is assumed to be in a thermal equilibrium state,  
	\begin{equation}
		\label{Eq:thermal_state}
		\gamma_E = \frac{e^{-\beta H_E}}{\tr(e^{-\beta H_E})},
	\end{equation}
	and the system together with the heat bath start initially in an uncorrelated state $\rho\otimes\gamma_E$. We will limit our considerations to \emph{energy-incoherent} states $\rho$, i.e., the ones that commute with the Hamiltonian, $[\rho,H]=0$. In that case, the state of the system can be equivalently described by a probability vector $\v p$ of eigenvalues of $\rho$, corresponding to populations in the energy eigenbasis. The thermal equilibrium state of the system, given by Eq.~\eqref{Eq:thermal_state} with $H_E$ replaced by $H$, is then represented by a vector of thermal populations $\v{\gamma}$. 
	We will denote the populations of generic input and output states by bold lowercase letters, $\v p$ and $\v q$, respectively. When referring to the components of these vectors, we use non-bold symbols with a lower index numbering the components, e.g., the $i$-th component of $\v{p}$ is denoted as $p_i$. The crucial point now is how the dynamics of the system interacting with the thermal bath is modelled. In what follows, we will review two frameworks, which can be seen as extreme cases with or without any memory effects involved.
	
	
	\emph{Thermal operations} (TOs) framework~\cite{Janzing2000,horodecki2013fundamental,Lostaglio2019} uses minimal assumptions on the joint system-bath dynamics by only assuming that the joint system is closed and thus evolves unitarily, and that this unitary evolution is energy-preserving. Formally, a set of thermal operations consists of completely positive trace-preserving maps that transform the state $\rho$ in the following way:
	\begin{equation}
		\label{Eq:thermal_operations}
		\E(\rho)=\Tr_E\left[U\left(\rho\otimes\gamma_E\right)U^{\dagger}\right],
	\end{equation}
	where $U$ is a joint unitary that commutes with the total Hamiltonian of the system and the bath
	\begin{equation}
		[U, H\otimes \iden_E+ \iden\otimes H_E] = 0,
	\end{equation}
	and the environmental Hamiltonian $H_E$ is arbitrary. Since there are no further constraints on $U$, arbitrarily strong correlations can build up between the system and the bath, and one can expect non-Markovian memory effects to come into play. At the same time, from the perspective of control theory, generating an arbitrary TO may require very complex and fine-tuned control over system-bath interactions~\cite{Lostaglio2018elementarythermal}. 
	
	\emph{Markovian thermal processes} (MTPs) framework~\cite{lostaglio2021continuous,korzekwa2022}, on the other hand, uses typical assumptions of the theory of open quantum systems (weak coupling, large bath size, quickly decaying correlations, etc.)~\cite{breuer2002theory}, to argue that the system undergoes an open dynamics described by a Lindblad master equation~\cite{kossakowski1972quantum,gorini1976completely,lindblad1976generators},
	\begin{equation}
		\label{Eq:master_equation}
		\frac{d \rho(t)}{dt}  = -i \qty[H, \rho(t)] + \mathcal{L}_t\qty(\rho(t)).
	\end{equation}
	In the above, $[\cdot,\cdot]$ denotes the commutator and $\mathcal{L}_t$ is the Lindbladian with the following general form:
	\begin{equation}
		\label{eq:lindbladian}
		\mathcal{L}_t(\rho) = \sum_{i} r_i(t) \left[ L_i(t) \rho L_i(t)^\dag - \frac{1}{2}\Bigl\{L_i(t)^\dag L_i(t), \rho\Bigl\}\, \!\right]\!,\!
	\end{equation} 
	with $\{\cdot,\cdot\}$ denoting the anticommutator, $L_i(t)$ being time-dependent jump operators, and $r_i(t)\geq 0$ being time-dependent non-negative jump rates. Moreover, the thermal state of the system is a stationary solution of the dynamics, $\mathcal L_t(\gamma) = 0$, and the Lindbladian $\mathcal L_t$ commutes with the generator of the Hamiltonian dynamics $-i[H,\cdot]$ for all times $t$. Formally, an MTP is then any quantum channel $\E$ that results from integrating Eq.~\eqref{Eq:master_equation} between 0 and $\tau\geq 0$. Since the dynamics generated by an MTP arises explicitly from a Markovian model, there are no memory effects. Also, as shown in Ref.~\cite{lostaglio2021continuous}, the universal set of controls that allows one to generate any incoherent state transformation achievable via MTPs consists only of two-level partial thermalisations. These transform the populations of two energy levels, $i$ and $j$, in the following way 
	\begin{subequations}
		\begin{align}
			\label{eq:partial1}
			p_i &\rightarrow (1-\lambda) p_i + \lambda \frac{p_i+p_j}{\gamma_i+\gamma_j} \gamma_i,\\
			\label{eq:partial2}
			p_j &\rightarrow (1-\lambda) p_j + \lambda \frac{p_i+p_j}{\gamma_i+\gamma_j} \gamma_j ,
		\end{align}
	\end{subequations}
	where $\lambda\in[0,1]$.
	
	\begin{figure}[t]
		\centering
		\includegraphics{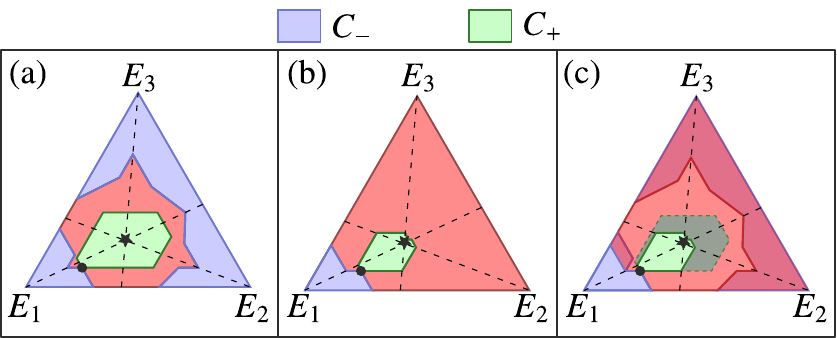}
		\caption{\textbf{Thermal operations vs Markovian thermal processes}. Sets of states that a three-level system with an equidistant energy spectrum $\v{E} = (0,1,2)$ and prepared in an energy-incoherent state $\v p = (0.7,0.2,0.1)$ (depicted by a black dot $\bullet$) can be transformed to [green region $C_+(\v p)$] or transformed from [blue region $C_-(\v p)$] by (a) thermal operations and (b) Markovian thermal processes with respect to inverse temperature $\beta = 0.3$. In (c) we show the overlap of sets of achievable states via TOs (dark green) and MTPs (light green). The thermal state of the system is depicted by a black star at the intersection of the dashed lines.}
		\label{Fig:thermalconevsmarkovianthermalcone}
	\end{figure}
	
	The set of states $C^{\textrm{TO}}_{+}(\v p)$ achievable via thermal operations from a given incoherent initial state $\v{p}$ can be fully characterised using the notion of thermomajorisation~\cite{horodecki2013fundamental} ~(see Appendix~\hyperref[app:thermomajorisation]{A-2} for more details). The so-called future thermal cone $C^{\textrm{TO}}_{+}(\v p)$~\cite{deoliveirajunior2022} is a convex set that consists of at most $d!$ extreme points, the construction of which was given in Lemma~12 of Ref.~\cite{Lostaglio2018elementarythermal} (see Fig.~\hyperref[Fig:thermalconevsmarkovianthermalcone]{\ref{Fig:thermalconevsmarkovianthermalcone}a} for an example with a three-level system). On the other hand, the set of states $C^{\textrm{MTP}}_{+}(\v p)$ achievable via Markovian thermal processes from a state $\v{p}$ was recently characterised using the notion of continuous thermomajorisation~\cite{lostaglio2021continuous}~(see Appendix~\hyperref[app:continuousmajorisation]{A-3} for an extended discussion). The future Markovian thermal cone $C^{\textrm{MTP}}_{+}(\v p)$ is not convex (as illustrated in Fig.~\hyperref[Fig:thermalconevsmarkovianthermalcone]{\ref{Fig:thermalconevsmarkovianthermalcone}b} for a three-level system), but Theorem~4 of Ref.~\cite{lostaglio2021continuous} provides a construction of its extreme points using sequences of two-level full thermalisations (i.e., transformations from Eqs.~\eqref{eq:partial1}-\eqref{eq:partial2} with $\lambda=1$). As can be seen in Fig.~\hyperref[Fig:thermalconevsmarkovianthermalcone]{\ref{Fig:thermalconevsmarkovianthermalcone}c}, $C^{\textrm{MTP}}_{+}(\v p)\subset C^{\textrm{TO}}_{+}(\v p)$ and the difference between these two sets of thermodynamically accessible states arises purely from memory effects.
	
	A recap of \emph{majorisation}, \emph{thermomajorisation}, and \emph{continuous thermomajorisation}, three notions needed to fully understand the convertibility of states under TOs and MTPs, is presented in Appendix~\ref{App:partial-order}. We will refrain from restating it here. Instead, we provide a summary in Fig.~\ref{Fig:summary_TP_and_MTP}, illustrating the existence of a given thermodynamic transformation and its conditions as expressed by these partial order relations.
	
	
	\section{Bridging the gap with memory}
	\label{sec:bridging}
	
	We begin this section by explaining the main building block of this work, namely the notion of memory-assisted Markovian thermal processes. Then, we demonstrate how energy-incoherent states achievable from a given initial state $\v{p}$ via thermal operations [i.e., any $\v{q}\in C_+^{\mathrm{TO}}(\v{p})$] can be approached arbitrarily well using memory-assisted Markovian thermal processes with large enough memory. We will start by simplifying the problem and showing that it is sufficient to only consider the achievability of the extreme points of $C_+^{\mathrm{TO}}(\v{p})$. Next, we will introduce MeMTP protocols that will serve us to approach extreme points of $C^{\textrm{TO}}_{+}(\v p)$ using MTPs acting on the system and memory state, $\v{p}\otimes \v{\gamma}_M$. Finally, we will analyse the performance of these protocols, i.e., we will show how well they approximate the desired transformations as the size of the memory $N$ grows. Due to structural differences, we will do this separately for the infinite temperature limit and the case of finite temperatures.
	
	
	\subsection{Memory-assisted Markovian thermal processes}
	
	In this work we want to interpolate between the two extreme regimes of arbitrarily strong and no memory effects described by TO and MTP frameworks. We will achieve this by focusing on the more restrictive MTP framework and extending it by explicitly modelling memory effects by bringing ancillary systems in thermal states, that will be discarded at the end. More precisely, we consider MTPs acting on a composite system consisting of the main $d$-dimensional system in a state $\rho$ and an $N$-dimensional memory system prepared in its thermal state $\gamma_M$ (i.e., given by Eq.~\eqref{Eq:thermal_state} with $H_E$ replaced by the Hamiltonian $H_M$ of the memory system, which can be arbitrary). The thermality of the ancillary system $M$ is crucial, as this way we ensure that no extra thermodynamic resources are brought in unaccounted, and the only role played by $M$ is to bring extra dimensions that can act as a memory. As already explained in the introduction, this can also be viewed as having control over the small $N$-dimensional part of the bath. Formally, we define the following set of quantum channels.
	\begin{defn}[Memory-assisted MTPs]
		A quantum channel $\E$ is called a memory-assisted Markovian thermal process (MeMTP) with memory of size $N$, if it can be written as
		\begin{equation}
			\E(\rho) = \Tr_{M}[\E_{\mathrm{MTP}}(\rho\otimes \gamma_M)],
		\end{equation}    
		where $\E_{\mathrm{MTP}}$ is a Markovian thermal process acting on the original system extended by an $N$-dimensional ancillary system $M$ prepared in a thermal state $\gamma_M$.
	\end{defn}
	
	\begin{figure}[t]
		\centering    \includegraphics{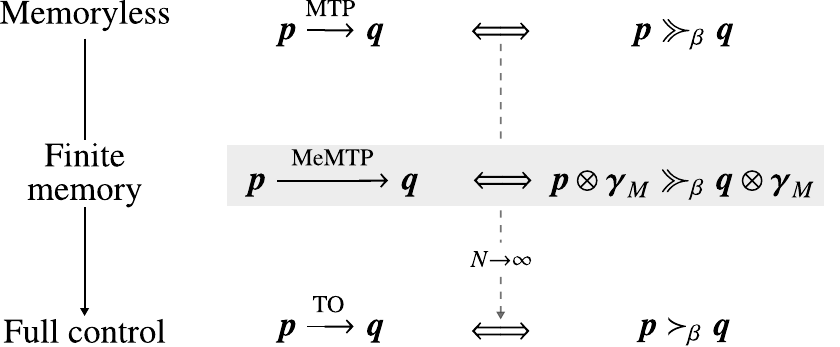}
		\caption{\textbf{Interpolating between extreme regimes}. Arrows between distributions represent the existence of specific thermodynamic transformations, whose existence is determined by partial-order relations: thermomajorisation $\succ_\beta$ for thermal operations and continuous thermomajorisation $\ggcurly_\beta$ for Markovian thermal processes. Our results demonstrate that the gap between these two frameworks can be bridged with the use memory-assisted MTPs employing ancillary memory systems of growing dimension $N$ and prepared in thermal states.}
		\label{Fig:summary_TP_and_MTP}
	\end{figure}
	
	As already mentioned, in this work we will focus on transformations between energy-incoherent states. Our aim is to show that the sets of states achievable from a given $\v{p}$ via memory-assisted MTPs interpolate between $C^{\textrm{MTP}}_{+}(\v p)$ (for $N=1$) and $C^{\textrm{TO}}_{+}(\v p)$ (for $N\to\infty$). The illustration of the growing strength of memory effects captured by our framework is presented in Fig.~\ref{Fig:summary_TP_and_MTP}. As a final note, observe that this framework can be formally related to a particular kind of catalytic transformations~\cite{Jonathan_1999,brandao2015second}. This is because the ancillary memory system can always be thermalised at the end of the process and this way be brought to the initial state. 
	
	
	\subsection{Simplification to extreme points}
	
	We start by recalling the following notion that is crucial for our analysis.
	
	\begin{defn}[$\beta$-ordering]\label{Def-beta-ordering}
		Let $\v p$ be an arbitrary energy-incoherent state of a $d$-dimensional system, and $\v \gamma$ denote the corresponding thermal Gibbs state. The $\beta$-ordering of $\v p$ is defined as the permutation $\pi_{\v p}$ that arranges the vector \mbox{$\qty(p_1/\gamma_1, ..., p_d/\gamma_d)$} in a non-increasing order, i.e., 
		\begin{equation}\label{Eq:beta-ordering}
			\v p^{\, \beta} = \qty(p_{\v \pi^{-1}_{\v p}(1)}, ...,p_{\v \pi^{-1}_{\v p}(d)}).    
		\end{equation}
		The $d$-dimensional matrix representation of $\pi_{\v p}$ will be denoted by $\Pi_{\v p}$, i.e., $\Pi_{\v p}\v{p}=\v{p}^\beta$.
	\end{defn}
	
	The future thermal cone $C_+^{\mathrm{TO}}(\v{p})$ is a polytope with at most~$d!$ extreme points, one for each possible $\beta$-order $ \pi$~\cite{Lostaglio2018elementarythermal}. We will denote them by $\v{p}^{\pi}$ (in particular it means that \mbox{$\v{p}^{\pi_{\v p}}=\v{p}$}). Now, we will use two crucial observations. First, in Ref.~\cite{Lostaglio2018elementarythermal} it was shown that
	\begin{equation}
		\label{eq:extremal_enough}
		\v{q} \in C_+^{\mathrm{TO}}(\v{p}) \Rightarrow \v{q} \in C_+^{\mathrm{TO}}\qty(\v{p}^{\pi_{\v{q}}}),
	\end{equation}
	meaning that all states with a $\beta$-order $\pi$ that can be achieved from $\v{p}$ via thermal operations can also be achieved starting from $\v{p}^\pi$. And second, it was shown in Ref.~\cite{lostaglio2021continuous} that
	\begin{equation}
		\label{eq:same_beta_order}
		\left[\v{q} \in C_+^{\mathrm{TO}}(\v{p})~~\mathrm{and}~~\pi_{\v{q}}=\pi_{\v{p}}\right] \Rightarrow \v{q} \in C_+^{\mathrm{MTP}}(\v{p}),
	\end{equation}
	meaning that within the same $\beta$-order as the initial state, the subsets of states achievable via TOs and via MTPs do coincide. As a result, if one can construct memory-assisted MTPs that reach all the extreme points of $C_+^{\mathrm{TO}}(\v{p})$, then one can also get to every state in $C_+^{\mathrm{TO}}(\v{p})$ via MeMTPs. This is done by simply first transforming $\v{p}$ to a given extreme point $\v{p}^\pi$ of $C_+^{\mathrm{TO}}(\v{p})$, and then using MTPs to get from $\v{p}^\pi$ to every state with a $\beta$-order $\pi$ in $C_+^{\mathrm{TO}}(\v{p})$.
	
	Note that the upper index on probability distributions serves a dual purpose. First, it is used to denote the $\beta$-ordered version of the state. For example, $\v p^\beta$ represents the $\beta$-ordered counterpart of $\v p$. Second, an upper index $\pi$ is used to denote particular extreme points of the state's future. More precisely, $\v p^\pi$ represents the extreme point of the future of $\v p$, with $\beta$-order given by $\pi$. Notably, $\pi_{\v p}$ represents the $\beta$-order of the state $\v p$.
	
	In order to quantify how well a given state in $C_+^{\mathrm{TO}}(\v{p})$ can be approximated, we will use the total variation distance defined by
	\begin{equation}
		\delta\qty
		(\v{p},\v{q}):=\frac{1}{2}\sum_{i=1}^d \abs{p_i-q_i}.
	\end{equation}
	From the discussion above, it should be clear that if we can construct MeMTP protocols approximating every extreme point with an error at most $\epsilon$, then
	\begin{equation}
		\forall \v{q} \in C_+^{\mathrm{TO}}\qty(\v{p}):\quad \min_{\P\in \mathrm{MeMTP}} \delta\qty(\P(\v{p}),\v{q})\leq \epsilon.
	\end{equation}
	
	A particular subset of extreme points of $C_+^{\mathrm{TO}}(\v{p})$ that we will investigate in more detail is given by those extreme states that can be achieved via sequences of $\beta$-swaps. A $\beta$-swap $\Pi_{ij}^\beta$ can be seen as a thermodynamic analogue of a population swap between levels~$i$ and~$j$~\cite{Lostaglio2018elementarythermal}:
	\begin{equation}
		\label{Eq:beta-swap}
		\Pi^{\, \beta}_{ij} := \begin{bmatrix}
			1-e^{-\beta {(E_j - E_i)}}  &1 \\ 
			e^{-\beta {(E_j - E_i)}} &0 
		\end{bmatrix}\oplus \mathbbm{1}_{{\backslash}(ij)},
	\end{equation}
	with $E_i \leq E_j$ and $\mathbbm{1}_{{\backslash}(ij)}$ denoting the $(d -2) \times (d-2)$ identity matrix on the subspace of all energy levels except $i, j$. Note that in the infinite temperature limit ($\beta=0$), the above recovers a transposition on levels $i$ and $j$, which we will simply denote by~$\Pi_{ij}$. In this limiting case, all extreme points of $C_+^{\mathrm{TO}}(\v{p})$ can be obtained by sequences of transpositions (that is because an extreme point in that case is of the form $\Pi \v{p}$ for a permutation matrix $\Pi$, and every $\Pi$ can be constructed from transpositions). 
	
	For finite temperatures, a $\beta$-swap transforms $\v{p}$ into the extreme point $\v{p}^\pi$ if the $\beta$-orders $\pi_{\v{p}}$ and $\pi$ differ only by a transposition of adjacent elements~\cite{Lostaglio2018elementarythermal,deoliveirajunior2022}. In other words, it happens for $\Pi_{ij}^\beta$ when $\pi_{\v{p}}(i)=\pi_{\v{p}}(j) \pm 1$. More generally, a sequence of $\beta$-swaps with non-overlapping supports will also produce an extreme point of $C^{\mathrm{TO}}_+(\v{p})$, and so a total number of 
	extreme points that can be achieved by sequences of $\beta$-swaps for dimension $d$ (including the starting point) 
	is given by $F(d+1)$, where $F(k)$ is the $k$-th Fibonacci number~\cite{white_1983}. Finally, we will also make use of the notion of $\beta$-cycles that we now define. For a state $\v{p}$, consider a $k$-dimensional subset of energy levels $ \qty{i_1,\hdots,i_{k}}$ neighbouring in the $\beta$-order, i.e., $\pi_{\v{p}}(i_{j+1}) = \pi_{\v{p}}(i_j) + 1$. Denote by $\pi$ a cyclic permutation on this subset, i.e., either $\pi(i_j) = i_{j+1 \text{ mod } k}$, or $\pi(i_j) = i_{j-1 \text{ mod } k}$. Then, a thermal operation mapping $\v{p}$ to its extreme point $\v{p}^{\pi'}$ is called a $\beta$-cycle, if $\Pi'=\Pi \Pi_{\v{p}}$ (here $\Pi, \Pi'$ and $\Pi_{\v{p}}$ denote matrix representations of permutations $\pi,\pi'$ and $\pi_{\v{p}}$). To emphasise that a given $\beta$-cycle acts on $k$ levels, we will sometimes refer to it as a $\beta$-$k$-cycle.
	
	To summarise the notion used, we denote a specific permutation as lowercase letter $\pi$, which acts on integers. Its matrix representation is represented by capital $\Pi$, and, in particular, the transposition of elements $i$ and $j$ has its matrix counterpart~$\Pi_{ij}$. We use the upper index $\beta$, e.g., $\Pi^\beta$, when referring to an extreme operation that permutes the $\beta$-order of the state by a permutation $\pi$.
	
	
	\subsection{Memory-assisted protocols}
	\label{Sec:Memory-assisted protocols}
	
	The basic building blocks of all our protocols are given by two-level elementary thermalisations that are formally defined as follows. 
	\begin{defn}[Two-level thermalisations\label{def:neighbour-thermalisations}] 
		Consider a system in a state $\v{p}$ with the corresponding thermal state $\v{\gamma}$. Then, an MTP transformation
		\begin{equation}\label{eq:neighbour_thermalisations}
			\{p_i, p_{j} \} \to \left\{ \frac{p_i+p_j}{\gamma_i+\gamma_j}\gamma_i,\frac{p_i+p_j}{\gamma_i+\gamma_j}\gamma_j \right\}
		\end{equation}
		is called a two-level thermalisation between levels $i$ and $j$, and the corresponding matrix acting on probability vectors will be denoted by $T_{ij}$. Moreover, if $\pi_{\v{p}}(i)=\pi_{\v{p}}(j) \pm 1$, then $T_{ij}$ is called a neighbour thermalisation.
	\end{defn}
	
	Let us note that the importance of neighbour thermalisations and the reason we employ them in our protocols stems from the fact that their sequences produce the extreme points of the Markovian thermal cone $C_+^{\mathrm{MTP}}$~\cite{lostaglio2021continuous}. Intuitively, one can expect that in order to approximate extreme states of $C_+^{\mathrm{TO}}(\v{p})$ using MeMTPs, one should get to the extreme points of $C_+^{\mathrm{MTP}}(\v{p}\otimes\v{\gamma}_M)$, and these can be achieved by neighbour thermalisations of the composite system-memory state.
	
	Before delving into the full details of our protocol for approximating $\beta$-swaps, let us start with a high-level description to provide some insight into our investigation. We begin with the simplest case of a two-level system and a two-dimensional memory, drawing an analogy between continuous (thermo)majorisation and connected vessels~(see Fig.~\ref{Fig:vessels} for a schematic representation). Considering two vessels—one filled with liquid and one empty—the most one can do when they are connected is to equalise the levels of the liquid between them. However, by adhering to the simple schematic provided in Fig. \ref{Fig:vessels}, it is possible to exceed this intuitively unbeatable limit.
		\begin{figure}[t]
			\centering
			\includegraphics{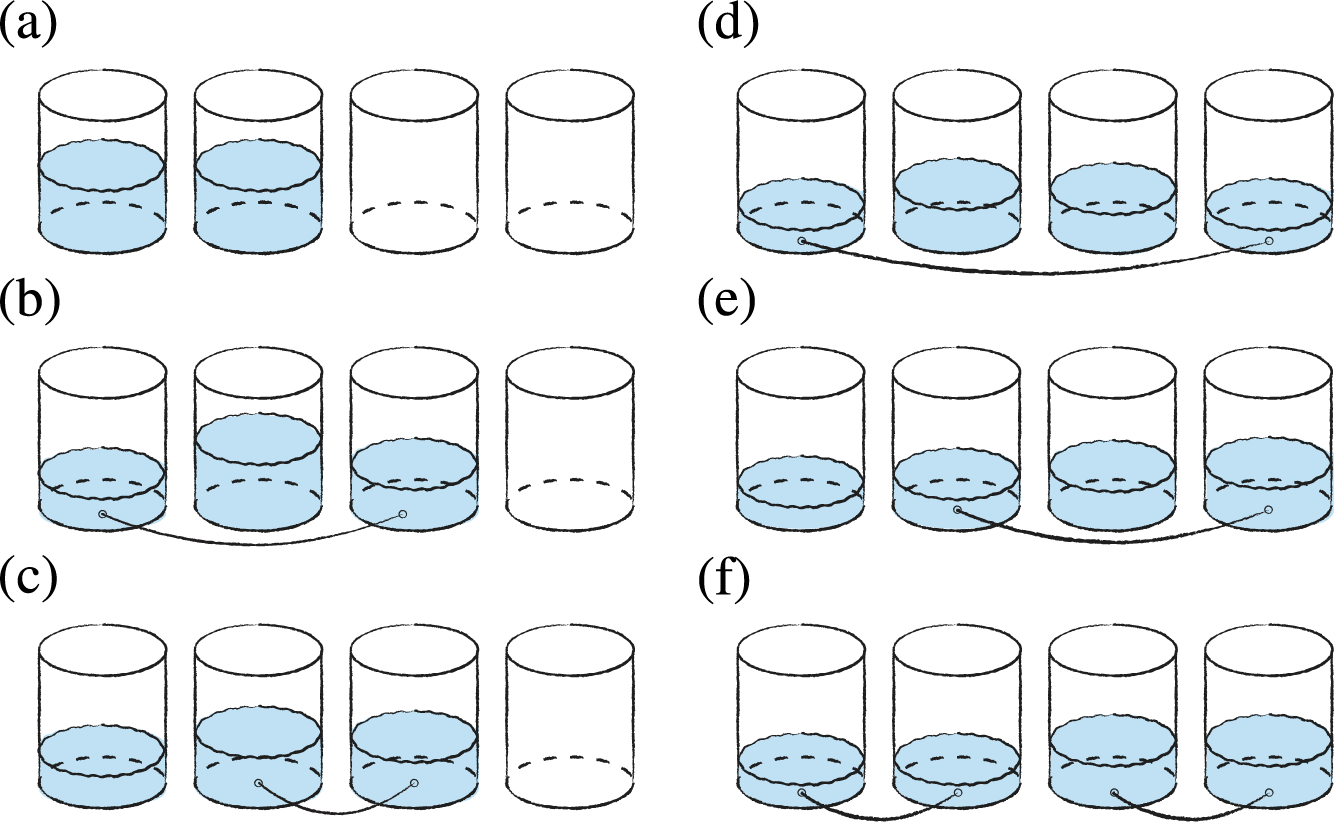}
			\caption{\textbf{Simplest example using connected vessels analogy}. Continuous (thermo)majorisation on $d$-level probability vectors is equivalent to a task of redistributing the content of $d$-ordered vessels that are connected pairwise. Adding a memory in a Gibbs state is then akin to multiplying glasses -- one empty and one full glass become $N$ pairs of full and empty glasses. The process involves five steps, read from top to bottom and left to right, which represent the simplest protocol that allows shifting more than half of the liquid from full to empty glasses. The final distribution after the protocol is applied is $(3/8, 5/8)$.}
			\label{Fig:vessels}
		\end{figure}
		\noindent 
		The sequence in Fig.~\ref{Fig:vessels} should be read from panels (a) to (f), in a top-down and left-to-right manner. We begin with four vessels: two are half-filled and two are empty. We can connect these vessels in pairs, thus equalising the fluid levels. The first two steps involve connecting the half-filled vessels sequentially to the first empty vessel. Likewise, the next two steps connect both initially half-filled vessels to the second empty one. The final step, which involves equalising the fluid levels in pairs, is analogous to thermalising the memory. An astute reader can confirm that $5/8$ of the total fluid ends up in the vessels that were initially empty, thereby surpassing the $1/2$ limit.
	
	We now describe our proposition for a MeMTP protocol approximating the $\beta$-swap $\Pi^{\, \beta}_{ij}$ between the $i$-th and $j$-th energy levels of the main system. It involves a sequence of $N^2$ two-level thermalisations of the state of the composite system (see Fig.~\ref{Fig:thermodynamic_protocol}), which includes the main system and a memory starting at thermal equilibrium, i.e., a state \mbox{$\v p \otimes \v \gamma_M$}. In particular, we focus on the populations $[\v p \, \otimes \, \v \gamma_M]_{N(i-1)+1},\hdots,[\v p \, \otimes\,  \v \gamma_M]_{Ni}$ corresponding to the $i$-th level of the main system and similarly for the $j$-th level. The protocol can be split into a sequence of $N$ rounds $\mathcal{R}_k^{(ij)}$ with $k = 1,\hdots,N$ consisting of $N$ steps each (shaded area in Fig.~\ref{Fig:thermodynamic_protocol}). In the $k$-th round, we select the entry $[\v p \otimes \v \gamma_M]_{N(i-1)+k}$ and thermalise it sequentially with all the levels corresponding to the level $j$ of the main system:
	\begin{equation}
		\mathcal{R}^{(ij)}_k(\v p \otimes \v \gamma_M) := \qty(\prod_{l=1}^N T_{(i-1)N + k,\,(j-1)N + l}) (\v p \otimes \v \gamma_M).
	\end{equation}
	Note that if $\pi_{\v{p}}(i)=\pi_{\v{p}}(j) \pm 1$ (i.e., the $\beta$-orders of $\v{p}$ and $\Pi^\beta_{ij}\v{p}$ differ by a transposition of adjacent elements), then all thermalisations performed are neighbour thermalisations. Using the above, we can now define the action of the truncated protocol $\widetilde{\mathcal{P}}^{(ij)}$:
	\begin{equation}\label{eq:trunc_protocol}
		\widetilde{\mathcal{P}}^{(ij)}(\v p \otimes \v \gamma_M) := \mathcal{R}^{(ij)}_N\circ\hdots\circ\mathcal{R}^{(ij)}_1(\v p \otimes \v \gamma_M).
	\end{equation}
	
	\begin{figure}[t]
		\centering
		\includegraphics{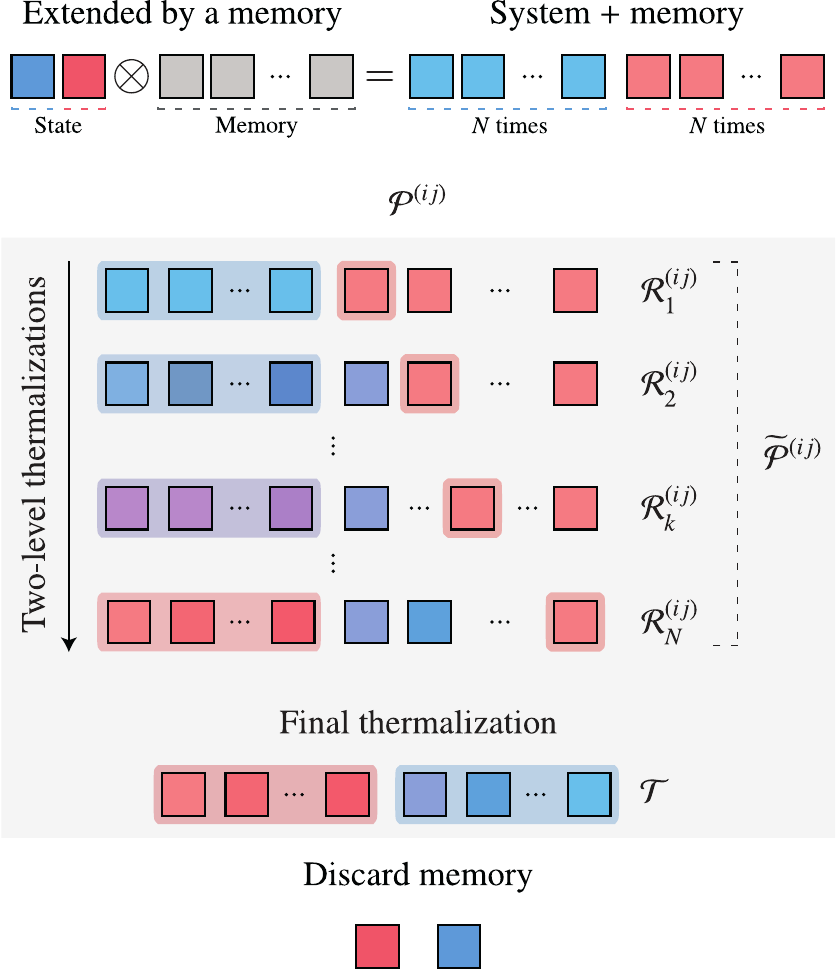}
		\caption{\textbf{$\beta$-swap protocol}. A two-level subsystem of a generically $d$-dimensional system, represented by blue and red squares, is extended by an $N$-dimensional memory represented by grey squares. The composite $2N$-dimensional system undergoes $N$ rounds of processing, where the $k$-th round involves $N$ sequential two-level thermalisations of the first $N$ entries with the $(N + k)$-th entry (represented by the shaded colour around the squares). After the final thermalisation step, the memory can be discarded.}
		\label{Fig:thermodynamic_protocol}
	\end{figure}
	
	The final step is to decouple the main system from the memory using a full thermalisation $\T$ of the memory system $M$, which acts on a general joint state $\v{Q}$ as: 
	\begin{equation}
		\mathcal{T}(\v Q) = \v q\otimes \v{\gamma}_M,\qquad q_i = \sum_{j=1}^N Q_{N(i-1)+j}.
	\end{equation}
	Thus, the full protocol approximating a $\beta$-swap $\Pi_{ij}^\beta$ is given by
	\begin{equation}
		\mathcal{P}^{(ij)}(\v p \otimes \v \gamma_M) = \mathcal{T}\circ\widetilde{\mathcal{P}}^{(ij)}(\v p \otimes \v \gamma_M).
	\end{equation}
	Putting the whole protocol in a simple words -- we take all ``full'' levels, connect them to the first ``empty'' level sequentially, and then repeat the process for all the empty levels.
	
	We will also employ more general protocols that aim at approximating the transformation of the initial state $\v{p}$ into an extreme state of $C_+^{\mathrm{TO}}(\v{p})$ given by $\v{p}^{\pi'}$. Denote matrix representations of $\beta$-orders of these states by $\Pi_{\v{p}}$ and $\Pi'=\Pi \Pi_{\v{p}}$ for some permutation matrix $\Pi$. Moreover, let us decompose $\Pi$ into neighbour transpositions with respect to $\v{p}$, i.e., we write \mbox{$\Pi = \Pi_{i_mj_m}\dots \Pi_{i_1j_1}$} with every consecutive transposition $\Pi_{i_kj_k}$ changing the $\beta$-order of the state $\Pi_{i_{k-1}j_{k-1}}\dots \Pi_{i_1j_1} \v{p}$ only by a transposition of adjacent elements. Then, we define the following two protocols to approximate $\v{p}^{\pi'}$:
	\begin{subequations}
		\begin{align}
			\label{eq:protocol}
			\mathcal{P}^{{\Pi}} &:= \mathcal{P}^{(i_mj_m)} \circ\dots\circ \mathcal{P}^{(i_1j_1)},
			\\
			\label{eq:truncated}
			\widetilde{\mathcal{P}}^{{\Pi}} &:= \T\circ \widetilde{\mathcal{P}}^{(i_mj_m)} \circ\dots\circ \widetilde{\mathcal{P}}^{(i_1j_1)}.
		\end{align}
	\end{subequations}
	Note that, by construction, all two-level thermalisations performed in the above protocols are neighbour thermalisations.
	
	Let us summarise the notation that we introduced for the protocols. Rounds of the protocol for approximating a swap between levels $i$ and $j$ of the main system are denoted by $\mathcal{R}^{(ij)}$. Similarly, a full protocol with and without final thermalization is denoted by $\mathcal{P}^{(ij)}$ and $\tilde{\mathcal{P}}^{(ij)}$, respectively. Final thermalisation is denoted by $\mathcal{T}$.
		
		As a last remark, it should be noted that the protocol discussed in this section is not unique. Alternative protocols achieving equivalent $\beta$-swap approximations can be found in Appendix~\ref{App:protocols}. Additionally, we provide a comparison of the convergence rates of different algorithms toward a specified target state.

	
	\subsection[Achieving extreme points of the future thermal cone for infinite temperature]{\texorpdfstring{Achieving extreme points of $C_+^{\mathrm{TO}}$ for $\beta= 0$}{Achieving extreme points of the future thermal cone for infinite temperature}}
	
	We are now ready to state our main results concerning the power of memory-assisted Markovian thermal processes in the infinite temperature limit. Let us recall that we focus on a $d$-level system and an $N$-dimensional memory in energy-incoherent states represented by probability distributions $\v p$ and $\v \gamma_M$. Since $\beta = 0$, the thermal state of the memory is described by a uniform distribution $\v{\eta}_M$ with every entry equal to $1/N$. We start with the following lemma, the proof of which can be found in Appendix~\ref{app:neighbour} (with the necessary background on the mathematical tools used presented in Appendix~\ref{app:beta}). 
	
	\begin{lem}[Memory-assisted transposition]
		\label{Lem:transposition} 
		In the infinite temperature limit, $\beta=0$, and for an $N$-dimensional memory, the MeMTP protocol $\mathcal{P}^{(ij)}$ acts as
		\begin{equation}
			\mathcal{P}^{(ij)} (\v{p}\otimes\v{\eta}_M)=  \v{q} \otimes \v \eta_M ,
		\end{equation}
		with
		\begin{equation}
			\v{q} = \left(\Pi_{ij} + \epsilon\qty(\mathbbm{1} - \Pi_{ij})\right) \v{p},
		\end{equation}
		and $\epsilon$ given by
		\begin{equation}
			\epsilon=(\pi N)^{-1/2} + o\qty(N^{-1/2}) \overset{N\rightarrow\infty}{\longrightarrow} 0.
		\end{equation}
	\end{lem}
	
	It is well known that any permutation of $d$ elements can be decomposed into a product of at most $d$ transpositions or $\binom{d}{2}$ neighbour transpositions. Therefore, by employing Lemma~\ref{Lem:transposition}, we can demonstrate that an arbitrary permutation can be achieved using a composition of our approximate protocols.
	
	\begin{thm}[Memory-assisted permutation]
		\label{Thm:permutation} 
		
		In the infinite temperature limit, $\beta=0$, and for an $N$-dimensional memory, $\Pi$ can be approximated by the MeMTP protocol $\mathcal{P}^{{\Pi}}$ as follows: 
		\begin{equation}
			\mathcal{P}^{{\Pi}}(\v{p}\otimes \v{\eta}_M) = \v{q}\otimes \v{\eta}_M,
		\end{equation}
		where
		\begin{equation}
			\v{q}=\left(\Pi + \epsilon \v\Delta + o\qty(N^{-1/2})\right)\v{p}
			\overset{N\rightarrow\infty}{\rightarrow} \Pi \v {p},
		\end{equation}
		with $\epsilon = (\pi N)^{-1/2}$ and the operator $\v\Delta$ defined in terms of transpositions appearing in the definition of $\P^\Pi$ in Eq.~\eqref{eq:protocol}:
		\begin{equation} \label{eq:Lambda_correction_op}
			\v\Delta = \sum_{l=1}^{m}\qty(\prod_{k=l+1}^{m} \Pi_{i_kj_k})\qty(\mathbbm{1}-\Pi_{i_lj_l})\qty(\prod_{k=1}^{l-1} \Pi_{i_kj_k}).
		\end{equation}
		
	\end{thm}
	\begin{proof}
		
		From the definition of $\mathcal{P}^\Pi$ and Lemma~\ref{Lem:transposition} we get
		\begin{equation}
			\v{q} =  \qty[\Pi_{i_mj_m} + \epsilon \qty(\mathbbm{1} - \Pi_{i_mj_m})]\dots \qty[\Pi_{i_1j_1} + \epsilon \qty(\mathbbm{1} - \Pi_{i_1j_1})]\v{p}.
		\end{equation}
		Clearly, the leading term is given by $\Pi\v{p}$, whereas the next leading term, proportional to $\epsilon$, is given by $\epsilon\Delta\v{p}$. All higher order terms scale at least as $\epsilon^2$, so are of the order $o(N^{-1/2})$.
	\end{proof}
	
	The above theorem can then be directly used to obtain the bound on how close one can get from a given $\v{p}$ to any state $\v{q}\in C_+^{\mathrm{TO}}(\v{p})$ using MeMTPs with $N$-dimensional memory. We explain how to derive such a bound for a given $\v{p}$ in Appendix~\ref{app:bound}, whereas below we present a weaker, but much simpler, bound that is independent of $\v{p}$.
	
	\begin{cor} \label{corr:general_bound}
		Consider states $\v{p}$ and $\v q\in C_+^{TO}(\v{p})$. Then, in the infinite temperature limit, $\beta=0$, and for an $N$-dimensional memory, there exists a MeMTP protocol $\mathcal{P}$ such that
		\begin{equation}
			\label{eq:cor1}
			\P(\v{p}\otimes \v{\eta}_M) = \v{q}'\otimes \v{\eta}_M,
		\end{equation}
		with
		\begin{equation}
			\label{eq:cor2}
			\delta(\v{q}',\v{q})\leq \frac{d(d-1)}{2\sqrt{\pi N}} + o\qty(N^{-1/2}).
		\end{equation}   
	\end{cor}
	\begin{proof}
		First, define $\Pi$ as a permutation that changes the $\beta$-order of $\v{p}$ to that of $\v{q}$. In other words, the $\beta$-order of $\Pi \v{p}$ is~$\pi_{\v{q}}$. Then, using Theorem~\ref{Thm:permutation}, we have that
		\begin{equation}
			\P^{\Pi}(\v{p}\otimes \v{\eta}_M)=\v{r}\otimes \v{\eta}_M
		\end{equation}
		with
		\begin{equation}
			\delta(\v{r},\Pi \v{p})\simeq\frac{1}{2\sqrt{\pi N}} \sum_{i=1}^d |(\Delta \v{p})_i|   \lesssim  \frac{m}{\sqrt{\pi N}} \lesssim  \frac{d(d-1)}{2\sqrt{\pi N}},
		\end{equation}
		where $\simeq$ and $\lesssim$ denote the equalities and inequalities up to $o(N^{-1/2})$. In the above, we have used the triangle inequality and the fact that one can always decompose $\Pi$ into at most $d(d-1)/2$ neighbour transpositions. Next, from Eq.~\eqref{eq:same_beta_order}, we know that there exists an MTP protocol $\P'$ mapping $\Pi\v{p}$ to~$\v{q}$. Using the contractiveness of the total variation distance under stochastic processing, we then have
		\begin{align}
			\delta(\P'(\v{r}),\v{q})&=\delta(\P'(\v{r}),\P'(\Pi\v{p}))\leq \delta(\v{r},\Pi\v{p})\lesssim \frac{d(d-1)}{2\sqrt{\pi N}}.
		\end{align}
		We thus conclude that by choosing $\P=(\P'\otimes \mathcal{I}_M)\circ \P^\Pi$, Eqs.~\eqref{eq:cor1}-\eqref{eq:cor2} are satisfied.
	\end{proof}
	
	Furthermore, we present the following conjecture for a better approximation of arbitrary permutations.
	
	\begin{conj}[Improved convergence]
		\label{Conj:convergence}
		In the infinite temperature limit, $\beta=0$, and for an $N$-dimensional memory, $ \widetilde{\mathcal{P}}^{{\Pi}}$ gives a better approximation of a permutation $\Pi$ than  $\mathcal{P}^{\Pi}$:
		\begin{equation}
			\delta\left( \Pi \v p,\tilde{\v q}\right) \leq  \delta\left( \Pi \v p, \v{q}\right),
		\end{equation}
		where $\tilde{\v{q}}$ and $\v{q}$ are defined via
		\begin{equation}
			\widetilde{\mathcal{P}}^{{\Pi}}(\v{p}\otimes \v{\eta}_M) = \tilde{\v{q}}\otimes \v{\eta}_M,\qquad \mathcal{P}^{{\Pi}}(\v{p}\otimes \v{\eta}_M) = \v{q}\otimes \v{\eta}_M.
		\end{equation}
	\end{conj}
	The conjecture is solidified by strong numerical evidence (see Fig.~\ref{fig:beta0_convergence} for an example considering $d = 6$). We note that the convergence is better, but the overall character of $O(N^{-1/2})$ is still preserved. More specifically, we observe that for permutations given by $\beta$-$k$-cycles with $k \leq d$, there is no advantage to removing the intermediate thermalisations (i.e., no advantage of $\widetilde{\P}^\Pi$ over ${\P}^\Pi$). The advantage already appears for a composition of $\beta$-$d$-cycle with $\beta$-$(d-1)$-cycle, leading to the $\beta$-order $(d, d-1,1,\hdots,d-2)$ (here, without loss of generality, we assumed that the initial $\beta$-order is given by $(1,2,\dots,d)$). In general, the advantage grows with the number of composed $\beta$-cycles (see Fig.~\ref{fig:beta0_convergence}, where different colours and markers correspond to different length compositions of $\beta$-cycles). In particular, we verified that for a permutation $(16)(25)(34)$, which is composed of $\beta$-cycles of length $6$ through $2$ (or $15 = \binom{6}{2}$ neighbour transpositions), both $\mathcal{P}^{\Pi}$ and $\widetilde{\mathcal{P}}^{\Pi}$ converge to the actual extreme point $\Pi \v{p}$. Surprisingly, we find that all the other possible permutations fall within the convergence advantage class of one of the aforementioned $\beta$-cycle compositions. This includes, in particular, the cases when the last $\beta$-cycle in the sequence is incomplete, i.e., it is shortened by the final subsequence of $\beta$-swaps of any length.
	
	\begin{figure}[t]
		\centering \includegraphics{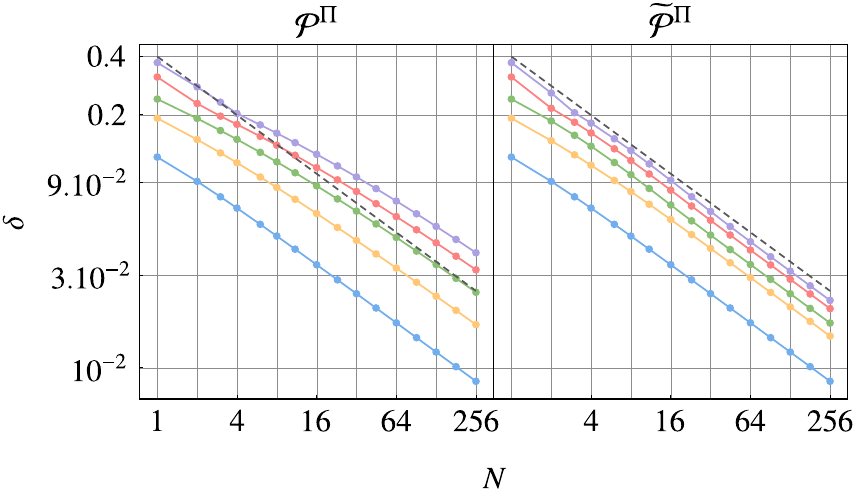}
		\caption{\textbf{Convergence rates at infinite temperature}. Log-log plot of the total variation distance $\delta$ between the extreme points \mbox{$\v{p}^\pi\in C_+^{\mathrm{(TO)}}(\v{p})$} and the states obtained from $\v{p}$ via the algorithm $\mathcal{P}^{\Pi}$ (left panel) and $\widetilde{\mathcal{P}}^{ \Pi}$ (right panel), as a function of the memory size~$N$. Here, \mbox{$\v p = \qty(0.37, 0.24, 0.16, 0.11, 0.07, 0.05)$}, $\beta = 0$, and different colours correspond to families of extreme points $\v{p}^\pi$ with different convergence rates (from bottom to top $\v{p}^\pi$ is obtained from $\v{p}$ via a $\beta$-6-cycle, a composition of a $\beta$-6-cycle with a $\beta$-5-cycle, and so on). All convergences behave as $O(N^{-1/2})$ as given in Eq.~\eqref{eq:cor2}, which can be seen by the comparison with the function $0.4/\sqrt{N}$ (dashed black line), with multiplicative advantage for $\widetilde{\mathcal{P}}^{ \Pi}$ over $\mathcal{P}^{ \Pi}$. }
		\label{fig:beta0_convergence}
	\end{figure}
	
	
	\subsection[Achieving extreme points of the future thermal cone for finite temperatures]{\texorpdfstring{Achieving extreme points of $C_+^{\mathrm{TO}}$ for $\beta\neq 0$}{Achieving extreme points of the future thermal cone for finite temperatures}}
	
	Our second main result concerns the power of memory-assisted Markovian thermal processes at finite temperatures. 
	We start with the following generalisation of Lemma~\ref{Lem:transposition}, the proof of which can be found in Appendix~\ref{app:neighbour}.  
	
	\begin{thm}[Memory-assisted $\beta$-swap]
		\label{Thm:beta-swap} 
		For a finite temperature, $\beta\neq 0$, and for an $N$-dimensional memory described by a trivial Hamiltonian (so that its thermal state is $\v{\eta}_M$), $\Pi_{ij}^\beta$ can be approximated by the MeMTP protocol $\P^{(ij)}$ as follows:    
		\begin{equation}
			\P^{(ij)} (\v{p}\otimes \v \eta_M)=\v{q}\otimes \v \eta_M,
		\end{equation}
		with
		\begin{equation}
			\delta(\v q,\Pi_{ij}^\beta \v{p}) = \frac{\qty(4\Gamma_i\Gamma_j)^N}{\qty(\Gamma_i - \Gamma_j)^2}\qty[\frac{\abs{p_i \Gamma_j - p_j\Gamma_i}}{(N+1)\sqrt{\pi N}} + o\qty(N^{-3/2})]
		\end{equation}
		where we have used $\Gamma_i = \gamma_i/(\gamma_i + \gamma_j)$ and likewise for $\Gamma_j$.
	\end{thm}

	By using the above theorem, one can approximate with arbitrary precision a total of $F(d+1)$ extreme points achievable by a composition of non-overlapping $\beta$-swaps (recall that $F(k)$ is the $k$-th Fibonacci number). However, since $F(d+1)\leq d!$ for $d\geq 3$, not all extreme points of $C_+^{\mathrm{TO}}(\v{p})$ can be obtained this way. Nevertheless, we conjecture that using MeMTP protocols $\widetilde{\P}^\Pi$ that are composed of blocks imitating $\beta$-swaps, just without intermediate thermalisations, one can reach all the extreme points of $C_+^{\mathrm{TO}}(\v{p})$.
	
	\begin{conj}[Extreme points of $C_+^{\mathrm{TO}}$]\label{Conj:beta_permutations}
		Consider a state $\v{p}$ and the extreme point of $C_+^{\mathrm{TO}}(\v{p})$ given by $\v{p}^{\pi'}$, with matrix representations of $\beta$-orders of these states satisfying \mbox{$\Pi'=\Pi\Pi_{\v{p}}$} for some permutation $\Pi$. Then, for a finite temperature, $\beta\neq 0$, and for an $N$-dimensional memory described by a trivial Hamiltonian (so that its thermal state is $\v{\eta}_M$), the MeMTP protocol $\widetilde{\P}^{\Pi}$ acts as
		\begin{equation}
			\widetilde{\P}^{\Pi}(\v{p}\otimes \v{\eta}_M)=\v{q}\otimes \v{\eta}_M,
		\end{equation}
		with
		\begin{equation} \label{eq:beta_nonzero_conv}
			\delta(\v{q},\v{p}^{\pi'})\overset{N\rightarrow\infty}{\longrightarrow} 0.
		\end{equation}
	\end{conj}
	
	The conjecture is solidified by the following two pieces of evidence, which utilize the truncated protocols $\tilde{\mathcal{P}}^{\Pi}$. First, we provide an analytical proof for convergence of the truncated protocols to a subset of extreme points beyond the ones achievable with simple $\beta$-swaps.
	
	\begin{thm}[Memory-assisted $\beta$-3-cycle]
		\label{Thm:beta-3-cycle} 
		
		Consider a state~$\v{p}$ with entries $i_1,i_2,i_3$ being neighbours in the $\beta$-order (i.e., $\pi_{\v{p}}(i_1)=\pi_{\v{p}}(i_2)+1=\pi_{\v{p}}(i_3)+2$), and the extreme point of $C_+^{\mathrm{TO}}(\v{p})$ given by $\v{p}^{\pi'}$, with matrix representations of $\beta$-orders of these states satisfying \mbox{$\Pi'=\Pi\Pi_{\v{p}}$} for $\Pi=\Pi_{i_1i_3}\Pi_{i_2i_3}$. Then, for a finite temperature, $\beta\neq 0$, and for an $N$-dimensional memory described by a trivial Hamiltonian (so that its thermal state is $\v{\eta}_M$), the MeMTP protocol $\widetilde{\P}^{\Pi}$ acts as
		\begin{equation}
			\widetilde{\P}^{\Pi}(\v{p}\otimes \v{\eta}_M)= \v{q}\otimes \v{\eta}_M
		\end{equation}
		with
		\begin{equation}
			\delta(\v{q},\v{p}^{\pi'})\overset{N\rightarrow\infty}{\rightarrow} 0.
		\end{equation}
	\end{thm}
	The proof of the above theorem can be found in Appendix~\ref{app:beta-3-cycle}, and potentially the same proving techniques can be applied to higher-order cycles. This would then provide a general method for simulating $\beta$-cycles, as well as any combinations of non-overlapping $\beta$-cycles, with arbitrary precision through MeMTPs. 
	
	\begin{figure}[t]
		\centering
		\includegraphics{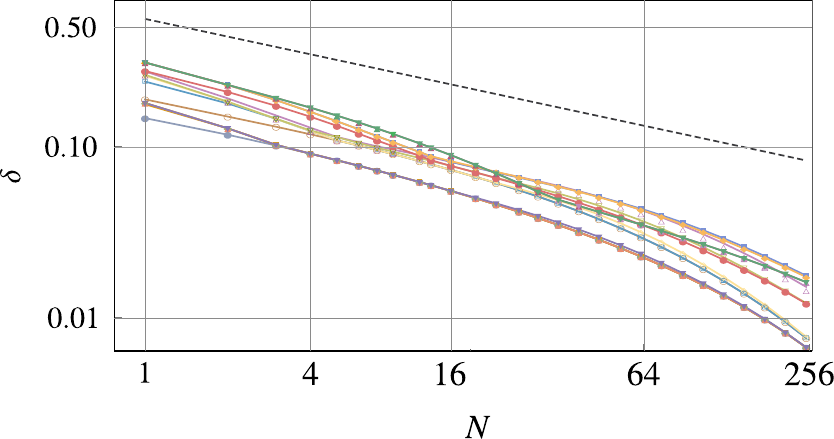}
		\caption{\textbf{Convergence rates at finite temperature}. Log-log plot of the total variation distance $\delta$ between the states obtained from $\v{p}$ via the algorithm $\widetilde{\mathcal{P}}^{ \widetilde{\Pi}_{ij}}$ and the corresponding extreme points \mbox{$\v{p}^\pi\in C_+^{\mathrm{(TO)}}(\v{p})$}, as a function of the memory size $N$. Here, \mbox{$\v p = \qty(0.37, 0.24, 0.16, 0.11, 0.07, 0.05)$}, $\beta = 0.1$, and different colours correspond to extreme points $\v{p}^\pi$ with matrix representation of the $\beta$-order $\pi$ given by $\Pi=\widetilde{\Pi}_{ij}\Pi_{\v{p}}$.
			For all curves, the convergence is better than $O(N^{-1/2})$ in agreement with Eq.~\eqref{eq:beta_perm_convergence}, as can be seen by the comparison with the limiting line $1/\sqrt{\pi N}$ for $\beta = 0$ (dot-dashed black line).}
		\label{fig:beta_nonzero_convergence}
	\end{figure}
	
	The second piece of evidence is based on extensive numerical simulations demonstrating convergence to arbitrary extreme points beyond both $\beta$-swaps and aforementioned $\beta$-cycles. Let us consider any state $\v{p}$ and a set of permutations defined by a recurrence formula
	\begin{equation}
		\widetilde{\Pi}_{i j} := \qty(\prod_{k=1}^{j} \Pi_{\pi_{\v p}(k),\, \pi_{\v p}(k+1)})\widetilde{\Pi}_{i-1,\,d+1-i}
	\end{equation}
	with $1 \leq j \leq d-i$ and assuming the starting condition $\widetilde{\Pi}_{0d} = \mathbbm{1}$. Note that $\widetilde{\Pi}_{1,d-1}$ represents a full $\beta$-$d$-cycle, $\widetilde{\Pi}_{i,d+1-i}$ a composition of $i$ $\beta$-cycles of length from $d$ to $d+1-i$, and finally $\widetilde{\Pi}_{d,1}$ is a permutation which fully reverses the $\beta$-order of $\v p$. For each such permutation, we have considered the action of the protocol $\widetilde{\mathcal{P}}^{\widetilde{\Pi}_{ij}}\qty(\v p)$ and its convergence to the respective extreme point~$\v{p}^{\pi}$ with $\Pi=\widetilde{\Pi}_{ij} \Pi_{\v{p}}$ (recall that $\Pi$ is a matrix representation of $\pi$). In each case, we have observed the convergence of the form from Eq.~\eqref{eq:beta_nonzero_conv} that is better than $N^{-1/2}$. Results for an exemplary state in dimension $d = 6$ are presented in Fig.~\ref{fig:beta_nonzero_convergence}, where a total of $15$ different curves are shown to lie below the $N^{-1/2}$ limit and diverging from it.
	
	Finally, based on Theorem~\ref{Thm:beta-swap}, the proof of Theorem~\ref{Thm:beta-3-cycle}, and numerical evidence, one can reasonably strengthen Conjecture~\ref{Conj:beta_permutations} to make the following statement on the convergence:
	\begin{equation}\label{eq:beta_perm_convergence}
		\delta(\v{q},\v{p}^{\pi'})=O\qty(\frac{e^{-A({\Pi}) N}}{N^{3/2}}),
	\end{equation}
	where $A(\Pi) = O(1)$ is a permutation-dependent exponent.

	
	\section{Discussion and applications}
	\label{sec:discussion}
	
	In Section~\ref{sec:bridging}, we demonstrated a method of achieving an arbitrary state from the future cone of TO using MeMTPs through MTP operations acting upon the system extended by memory, initiated in the thermal state $\gamma$. In the following sections, we will apply our protocol to study information-based quantum thermodynamic processes, such as work extraction and cooling. Next, we revisit the question of the sufficiency of two-level control for TOs. Finally, we provide a brief discussion of the behaviour of the free energy and correlations with the progression of our protocol. This sheds light on how non-Markovian effects arise in the memory-assisted protocol.
	
	
	\subsection{Work extraction}
	\label{subsec:workExtraction}

		Generally, definitions of work rely either on the control and manipulation of external parameters that determine the system's Hamiltonian~\cite{Alicki_1979,kosloff2013quantum} or on the assumption that work is a random variable and a controlled Hamiltonian evolution is employed to determine the work statistics~\cite{PhysRevE.90.032137,PhysRevE.92.042150,PhysRevLett.118.070601,PhysRevLett.123.230603}. In contrast, the resource-theoretic approach avoids the presence of any external agent and does not involve changes to the Hamiltonian. Furthermore, this framework differs from traditional approaches by shifting the focus from average and higher moments of the work distribution to the single-shot regime. In this regime, the question is posed as to what is the maximum amount of work that can be extracted while allowing for a probability of failure $\epsilon$.
	
	\begin{figure}[t]
		\label{}
		\centering
		\includegraphics{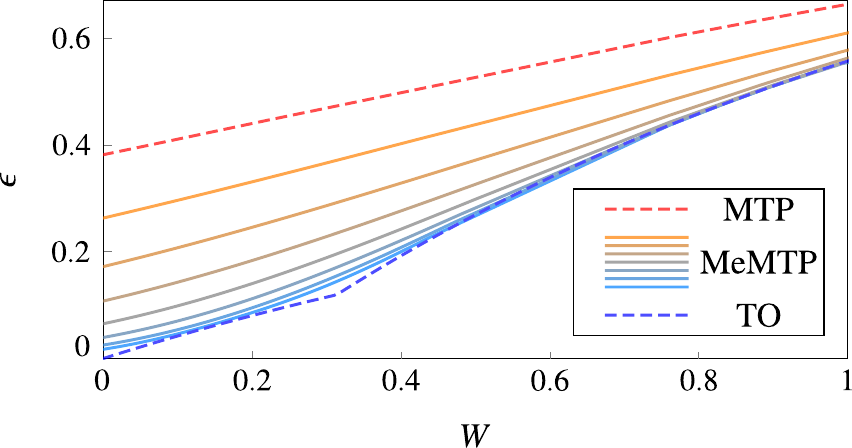}
		\caption{\textbf{$\epsilon$-deterministic work extraction with MeMTPS}. Transformation error $\epsilon$ as a function of the work $W$ extracted from a two-level system with energy splitting $\Delta$ prepared in a thermal state at temperature $1/\beta_S$ smaller than the environmental temperature $1/\beta$ with parameters 
			\mbox{$\beta_S \Delta=2$} and \mbox{$\beta \Delta=1$}. System-environment interactions are modelled by TOs (dashed black curve),  MTPs (dashed red curve) and memory-assisted Markovian thermal process with a memory of size $2$, $4$ , $8$, $16$, $32$, $64$ and $128$, respectively.}
		\label{Fig:work-extraction}
	\end{figure}
	
	The so-called $\epsilon$-deterministic work extraction, which typically involves an out-of-equilibrium system~$S$, a thermal bath at inverse temperature $\beta$, and a battery $B$ initially in an energy eigenstate $E_0$~\cite{Aberg2013,horodecki2013fundamental,Skrzypczyk2014}. The aim is to increase the energy of $B$ by an amount $W$ by exciting it from $E_0$ to $E_1=E_0+W$ with a success probability $1-\epsilon$. The optimal error $\epsilon$ for a given $W$ can be obtained via thermomajorisation condition for transformations given by thermal operations~\cite{horodecki2013fundamental} and through continuous thermomajorisation relations when transformations are given by Markovian thermal processes~\cite{korzekwa2022}. Our framework allows one to interpolate between the two extremes by including a memory system with varying dimension $N$. 
	
	Consider a two-level system $S$ and a two-level battery $B$ with energy levels $(0,\Delta)$ and $(0,W)$, respectively. Assume that the initial state of the joint system is given by $\v{p}_{SB} = \v{p}\otimes(1,0)$. One can then select the extreme point $\v{p}^{\pi'}_{SB}\in {C}^{\mathrm{TO}}_+\qty(\v{p}_{SB})$ from the future thermal cone of the composite system for which the following relation is satisfied with the minimum value of $\epsilon_{\mathrm{TO}}$:
	\begin{equation}
		\v{\gamma}\otimes(\epsilon_{\mathrm{TO}},1-\epsilon_{\mathrm{TO}}) \in 
		{C}^{\mathrm{TO}}_+\qty(\v{p}^{\pi'}_{SB}).
	\end{equation}
	In other words, $\v{p}^{\pi'}_{SB}$ is an intermediate state from which one can achieve minimal error for extracting $W$ work from $\v{p}$ via any thermal operation. We can now define $\Pi$ as a permutation that maps the matrix representation of the initial $\beta$-order of $\v{p}_{SB}$ to the final $\beta$-order of $\v{p}_{SB}^{\pi'}$. Then, by using the algorithm $\widetilde{\mathcal{P}}^\Pi$, we can transform $\v{p}_{SB}$ into a state $\v{q}_{SB}$ that approximates~$\v{p}^{\pi'}_{SB}$. Finally, due to Eq.~\eqref{eq:same_beta_order}, we can use standard thermomajorisation to find the minimal value of $\epsilon_N$ for which the state $\v{q}_{SB}$ can be transformed to $\v{\gamma}\otimes(\epsilon_N,1-\epsilon_N)$ via MTPs. Note that $\epsilon_N$ then corresponds to the probability of failure of extracting work $W$ from $\v{p}$ using a memory of size $N$. Numerical simulations of this procedure (see Fig.~\ref{Fig:work-extraction}) show that as $N$ grows, $\epsilon_N$ decreases, allowing us to conjecture that $\lim_{N\rightarrow\infty} \epsilon_N = \epsilon_{\mathrm{TO}}$. However, note that the convergence is not uniform: it is the slowest around $W = 0$ and the kink at $W=1/\beta \log (1+e^{-\beta \Delta})$. Nevertheless, Fig.~\ref{Fig:work-extraction} clearly shows that even a small size memory can significantly improve the quality of the extracted work.
	
	\subsection{Cooling a two-level system using a two-dimensional memory with nontrivial Hamiltonian}
		
		As a second application of our findings, we consider the task of cooling a two-level system with the aid of a two-dimensional memory. The setup involves a two-level system with energy gap $E_S$, extended by a memory system with energy gap $E_M$. The joint system's energy level structure is depicted in Fig.~\ref{fig-cooling}.  We assume that the difference between energy gaps is such that it allows one to selectively couple with the bath, i.e., $E_S - E_M \neq E_M$. This enables us to separately address transitions $\ket{01}\leftrightarrow\ket{10}$, $\ket{00}\leftrightarrow\ket{11}$ together with two coupled pairs of the form $\ket{0i}\leftrightarrow\ket{1i}$ and $\ket{i0}\leftrightarrow\ket{i1}$. We will refer to these operations thermalise these levels as operation 1, 2, 3 and 4, respectively.
		
		\begin{figure}[t]
			\centering
			\includegraphics{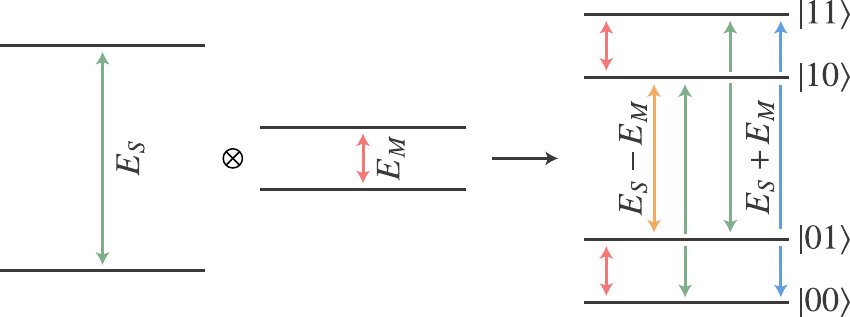}
			\caption{\textbf{Energy level structure.} Schematic diagram of a two-level system, consisting of a main system and a memory with energy gaps $E_S$ and $E_M$, respectively. The energy gap of the composite system is such that it can be selectively coupled to the thermal bath.}
			\label{fig-cooling}
		\end{figure}
		
		Let us now assume, for simplicity, that the system starts in an excited state extended by Gibbs memory, $\v{p}\otimes\v \gamma_M$ with $p_i = \delta_{i1}$. If we consider the main system alone with access only to MTPs, one can cool it down only to the ambient temperature. In this case, the system will reach thermal equilibrium, and the resulting distribution is given by
		\begin{equation}
			\v \gamma_S = \qty(\frac{1}{1+e^{-\beta E_S}}, \frac{e^{-\beta E_S}}{1 + e^{-\beta E_S}}).  
		\end{equation}
		However, by implementing our protocol, which can be realised as a sequence of operations, $1\rightarrow2\rightarrow3\rightarrow4$ and discarding (thermalising) the memory, we arrive at \mbox{$\mathcal{P}(\v{p}\otimes\v \gamma_M) = \v{q}\otimes\v \gamma_M$} with
		\begin{equation}
			\v{q} = \mqty[\frac{e^{\beta  E_M}+e^{\beta  \left(E_M+E_S\right)}+e^{\beta  \left(2 E_M+E_S\right)}+e^{\beta  \left(E_M+2 E_S\right)}+e^{\beta  E_S}}{\left(e^{\beta  E_S}+1\right) \left(e^{\beta 
					\left(E_M-E_S\right)}+1\right) \left(e^{\beta  \left(E_M+E_S\right)}+1\right)} \\
			\frac{e^{\beta  E_M}+e^{\beta  \left(2 E_M+E_S\right)}+e^{\beta  E_S}}{\left(e^{\beta  E_S}+1\right) \left(e^{\beta  E_M}+e^{\beta  E_S}\right) \left(e^{\beta 
					\left(E_M+E_S\right)}+1\right)}
			].
		\end{equation}
		The distance of this state from the Gibbs state $\v{\gamma}_S$ at ambient temperature in terms of the 1-norm is given by
		\begin{equation}
			\norm{\v{q} - \v \gamma_S}_1 = \frac{1}{\left(e^{-\beta  E_S}+1\right) \left[\cosh \left(\beta  E_M\right)+\cosh \left(\beta  E_S\right)\right]},
		\end{equation}
		which is positive for every non-zero value of $E_M$ and $E_S$. This means that despite non-triviality of the memory's spectrum, our simple memory-extended protocol achieves a cooling advantage over Markovian processes.

	
	\subsection{Two-level control is sufficient for thermal operations}
	\label{subsec:twoLevelControl}
	
	\begin{figure}[t]
		\centering    \includegraphics{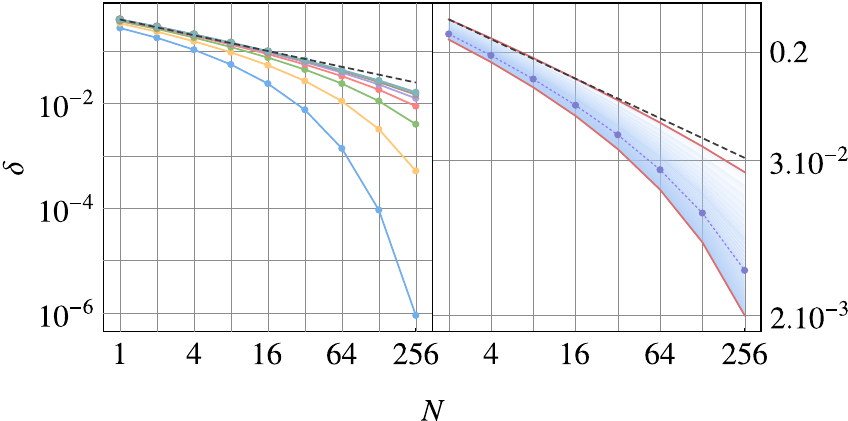}
		\caption{\textbf{Convergence to states inaccessible via elementary thermal operations.} 
			Log-log plot of the total variation distance $\delta$ between the state $\v{q}$ from Eq.~\eqref{eq:inaccessible} and the state obtained from \mbox{$\v{p}=(1,0,\dots,0)$} via the algorithm $\widetilde{\mathcal{P}}^{ \Pi_{\v{q}}}$, as a function of the memory size~$N$. Left: systems with energy spectra $E_i=i$ with $d = 3,\hdots,10$ (bottom to top) and for \mbox{$\beta = 1.1 \log(2)>\beta_{\mathrm{crit}}$}. For all presented dimensions the convergence is better than $1/(2\sqrt{N})$ (black dashed line). Right: systems with energy spectra $(0,E_1,E_2,1)$ taken from a grid with interval $\Delta = 1/64$ (translucent blue lines) for $\beta = 1.1\cdot\beta_{\text{crit}}$. The red lines represent the extreme cases of convergence, while the thick blue line is the average convergence. Note that the top red line almost agrees with $1/(2\sqrt{N})$ (black dashed line), which well approximates the expected convergence for $\beta = 0$ and agrees with the fact that it is obtained for almost completely degenerate levels.}
		\label{fig:M_point_convergence}
	\end{figure}
	
	In Ref.~\cite{Lostaglio2018elementarythermal} it has been proved that there exist thermodynamic state transformations that cannot be decomposed into the so-called \emph{elementary thermal operations}, i.e., thermal operations acting only on two levels of the system at the same time. Then, in Ref.~\cite{mazurek2018decomposability}, for any dimension $d$, an explicit final state $\v{q}\in C_+^{\mathrm{TO}}(\v{p})$ was given such that it cannot be achieved (even approximately) starting from the ground state $\v{p}=(1,0,\dots,0)$ using convex combinations of sequences of elementary thermal operations. More precisely, given the energy spectrum of the system with $E_{i+1}\geq E_i$, this final state is given by
	\begin{align}
		\label{eq:inaccessible}
		\v{q} & = \left(1 - \sum_{i=2}^{d}e^{-\beta E_i},    e^{-\beta E_2}, \dots, e^{-\beta E_{d}}\right)
	\end{align}
	with $\beta \geq \beta_{\text{crit}}$ such that $1 - \sum_{i=2}^{d}e^{-\beta_{\text{crit}} E_i} = 0$. It was then proven by the authors of Ref.~\cite{mazurek2018decomposability} that there exists $\epsilon>0$ such that any $\v{q}'$ achievable from $\v{p}$ satisfies $\delta(\v{q}, \v{q}') \geq \epsilon$. 
	
	Given the above, one might conclude that being able to selectively couple to the bath just two energy levels at once is highly restrictive and does not allow one to induce all the transitions possible via general thermal operations. This conclusion, however, would be incorrect, as the restriction only arises when one is limited to coupling only two levels of the \emph{system} at a given time. As we have seen in this paper, when one is allowed to bring an auxiliary $N$-level system in a thermal equilibrium state $\gamma_M$, then the ability to selectively couple to the bath just two energy levels of the joint system allows one to induce all transitions of the main system possible via thermal operations as $N\to\infty$. Crucially, the operation 
	\begin{equation}
		\E(\rho)=\rho\otimes \gamma_M
	\end{equation}
	is a thermal operation for every $N$. Thus, $\E$ followed by a sequence of elementary thermal operations on the joint system, followed by discarding the system $M$ at the end, can induce any energy-incoherent state transition of the system possible via general thermal operations. In other words, elementary control over two energy levels at a given time is sufficient to generate all thermodynamically possible transitions if we allow ancillary thermal systems.
	
	\begin{figure*}
		\centering
		\includegraphics{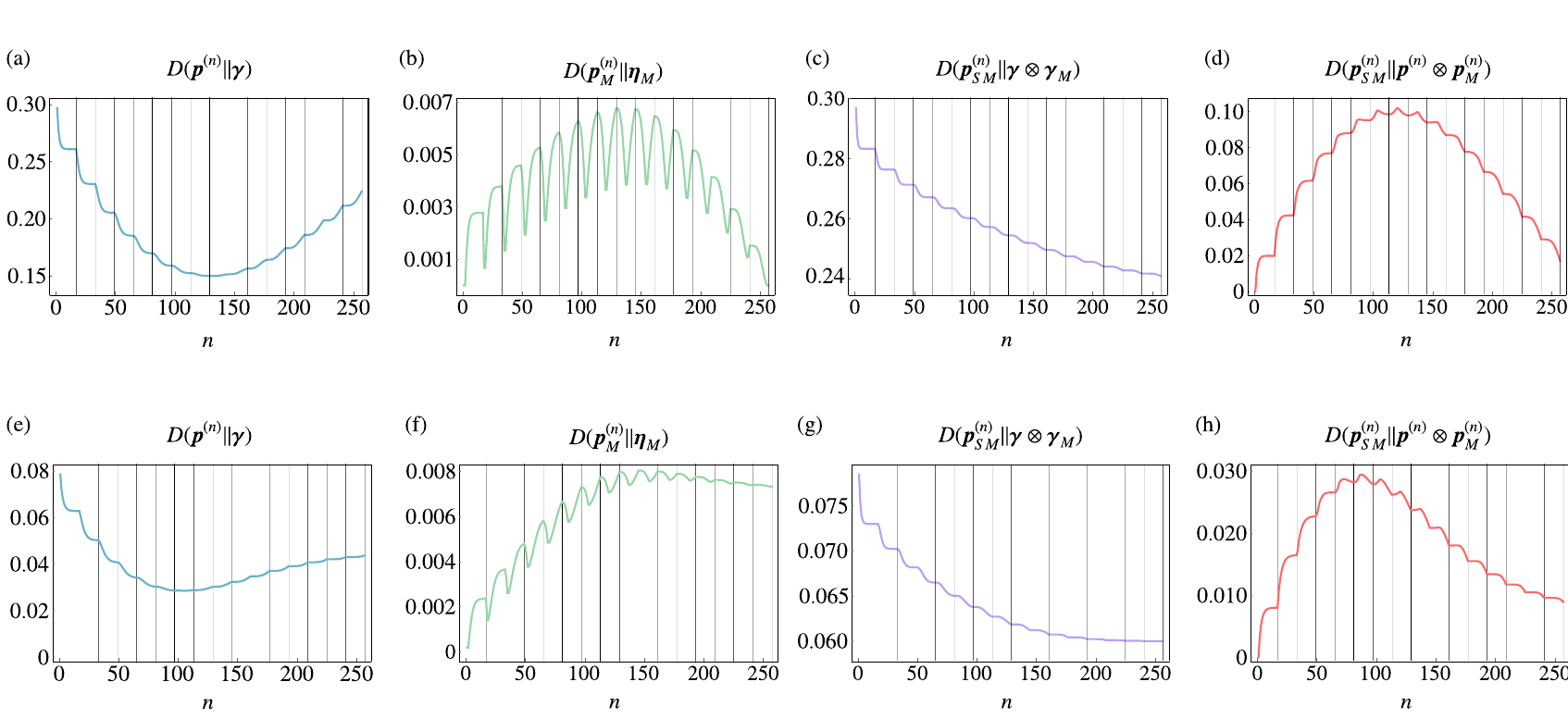}
		\caption{\textbf{Evolution of non-equilibrium free energies and correlations}. Non-equilibrium free energies of the main system [(a) and (e)], the memory [(b) and (f)], and the joint system [(c) and (g)], as well as the mutual information between the system and memory [(d) and (h)], as a function of the step number $n$ of the protocol $\P^{(ij)}$. Here, the composite system consists of a three-level system initialised in a state $\v p = (0.7,0.2,0.1)$ and a 16-dimensional degenerate memory prepared in a thermal (maximally mixed) state, and the plots are presented for two inverse temperatures, $\beta =0$ and $\beta = 0.5$.     
		}   
		\label{Fig:Free-energy}
	\end{figure*}
	
	We illustrate the above with the following numerical examples, showing that our MeMTP protocol $\widetilde{\mathcal{P}}^{\Pi_{\v q}}$ (which consists of only two-level operations) is able to transform $\v{p}$ into $\v{q}'$ that approximates $\v{q}$ arbitrarily well (i.e., $\delta(\v{q}',\v{q})\to 0$ as $N\to\infty$). In order to focus attention, we chose a constant $\beta = 1.1\cdot \beta_{\text{crit}}$. First, we considered systems of varying dimension $d$, up to $d_{\text{max}} = 80$, and fixed the energy structure to $E_i = i$, corresponding to quantum harmonic oscillator. We observed that the convergence for all these dimensions scales according to the predictions from Conjecture~\ref{Conj:beta_permutations}, which can be seen in the left panel of Fig.~\ref{fig:M_point_convergence} for $d =3,\hdots,10$ and memory sizes up to $N = 2^8$. Moreover, in order to ascertain that the convergence does not depend on the energy structure of the system, we fixed $d = 4$ and considered energy levels $(0, E_1, E_2, 1)$ with $E_1 < E_2$ taken from a grid with spacing $\Delta E = 2^{-6}$, resulting in $1953$ uniformly distributed points. For each of these points, we have considered the protocol with memory size up to $N = 2^8$. As demonstrated in the right panel of Fig.~\ref{fig:M_point_convergence}, it turns out again that the convergence, independently from the energy structure, is better than $1/\sqrt{N}$, in accordance with Conjecture~\ref{Conj:beta_permutations}. 
	
	
	\subsection{Non-equilibrium free energy evolution}
	\label{subsec:nonEquilibriumFree}
	
	To understand how non-Markovian effects arise in the memory-assisted protocol, we will now examine the evolution of the system and memory during the protocol $\mathcal{P}^{(ij)}$. More precisely, let us denote the joint state of the system and memory after the $n$-th two-level thermalisation step of the protocol by~$\v{p}_{SM}^{(n)}$. Similarly, let $\v{p}^{(n)}$ and $\v{p}_{M}^{(n)}$ denote the reduced states of the system and memory after the $n$-th step. Then, in the spirit of the analysis performed for elementary thermal operations in Ref.~\cite{Jeongrak2022}, we will examine the behaviour of the following entropic quantities. First, we will look at the relative entropy between $\v{p}^{(n)}$ and the thermal state of the system $\v \gamma$,
	\begin{equation}
		D\left(\v p^{(n)}\|\v{\gamma}\right)= \sum^{d}_{i=1} p^{(n)}_i \log \frac{p^{(n)}_i}{\gamma_i},
	\end{equation}
	which is a thermodynamic monotone, as it decreases under (Markovian) thermal operations, and is directly related to the non-equilibrium free energy~\cite{brandao2013resource}. We will also look into the behaviour of the analogous quantities for the joint system and the memory system. Moreover, to track the correlations that build up between the system and memory, we will investigate the mutual information between them, which is given by \mbox{$D(\v{p}^{(n)}_{SM} \| \v{p}^{(n)} \otimes \v{p}^{(n)}_M)$}.
	
	We use a three-level system and a $16$-dimensional memory as an illustrative example. We consider the joint system undergoing a $\beta$-swap protocol $\mathcal{P}^{(ij)}$ for $\beta=0$ and $\beta=0.5$. As shown in Figs.~\hyperref[Fig:Free-energy]{\ref{Fig:Free-energy}a} and~\hyperref[Fig:Free-energy]{\ref{Fig:Free-energy}e}, the non-equilibrium free energy of the main system initially decreases to a minimum, and then increases until it reaches a level that closely approximates the target state (a swap/$\beta$-swap). It is important to note that this observed increase is only possible because of the presence of the memory system. In contrast, note that the global non-equilibrium free energy decreases after each step, as depicted in Figs.~\hyperref[Fig:Free-energy]{\ref{Fig:Free-energy}c} and~\hyperref[Fig:Free-energy]{\ref{Fig:Free-energy}g}. However, during the process, a fraction of the main system's non-equilibrium free energy is transferred to the memory, which acts as a free energy storage. As such, it later enables the system to increase its local free energy again, hence allowing it to achieve the final state. More interestingly, the free energy of the memory, presented in Figs.~\hyperref[Fig:Free-energy]{\ref{Fig:Free-energy}b} and~\hyperref[Fig:Free-energy]{\ref{Fig:Free-energy}f}, exhibits a comb-like structure consisting of $d-1$ teeth with $(d+1)$ steps each. Specifically, within the $k$-th tooth, the first $d-k$ steps increase the free energy, while the remaining $k$ steps decrease it. Note that as long as the memory is not thermalised, its non-equilibrium free energy does not go to zero. However, for $\beta = 0$, it approaches a value very close to zero, but there are still correlations between the memory and the system, which are illustrated in Figs.~\hyperref[Fig:Free-energy]{\ref{Fig:Free-energy}d} and~\hyperref[Fig:Free-energy]{\ref{Fig:Free-energy}e}. 
	
	
	\section{Conclusions and outlook}~\label{sec:outlook}
	
	In this paper, we proposed a novel approach to investigate memory effects in thermodynamics by introducing the concept of memory-assisted Markovian thermal processes. These were defined by extending the framework of Markovian thermal processes with ancillary memory systems brought in thermal equilibrium states. Our construction allowed us to interpolate between the regime of memoryless dynamics and the one with full control over all degrees of freedom of the system and the bath. Using a family of protocols composed of Markovian thermal processes, we demonstrated that energy-incoherent states achievable from a given initial state via thermal operations could be approximated arbitrarily well via our algorithmic procedure employing memory. Furthermore, we analysed the convergence of our protocols in the infinite memory limit, finding polynomial and exponential convergence rates for infinite and finite temperatures, respectively. In the infinite temperature limit, we provided analytic convergence to the entire set of states accessible via thermal operations. For finite temperatures, we proved the convergence to a subset of accessible states and, based on extensive numerical evidence, we conjectured that a modified version of our protocol can realise arbitrary transitions achievable via thermal operations with an exponential convergence rate that grows with memory size. Our model-independent approach can be seen as a significant step forward in understanding ultimate limits of the Markovian evolution in general, which should be contrasted with the model-specific approaches to the so-called Markovian embedding \cite{siegle2010markovian,budini2013embedding,campbell2018system}. On the other hand, it may be seen as far less general than the approach taken in Ref.~\cite{ende2023finitedimensional}, where our work would correspond to a step towards simulating arbitrary evolution with Markov-Stinespring curves.

  Our framework should be contrasted with previous investigations into the embeddability of Markov chains into continuous Markov processes via hidden states~\cite{owen2019number, PhysRevResearch.3.023164, PhysRevX.11.021019}. These approaches consider a system with a set of "visible states," which are occupied and subjected to operations, along with "hidden states" that enable an embedding of operations that would otherwise be unfeasible. These hidden states could correspond to unpopulated levels within the main system—an option that remains viable for implementing our protocols. Alternatively, we could choose to extend the system by introducing a memory initialized in a specific energy state $\v r$ with populations ${r_i = \delta_{ij}}$ for some $j$. Nevertheless, this approach carries two caveats. The first involves the cost associated with bringing the non-equilibrium state $\v{r}$ into play, as opposed to a Gibbs state. The second pertains to the catalytic nature of the operation—whether operations of the form $\v{p} \otimes \v{r} \rightarrow \v{q} \otimes \v{r}$, where $\v{r}$ is retrievable, can indeed expand the space of accessible states~$\v{q}$.
	
	We also explained how our results can be employed to quantitatively assess the role of memory for the performance of thermodynamic protocols. In this context, we discussed the dependence on the memory size of the amount and quality of work extracted from a given non-equilibrium state. In addition, we introduced a minimal model designed to cool a two-level system below ambient temperature using a two-dimensional memory.
	However, the method can be used as well to investigate other thermodynamic protocols, such as information erasure or thermodynamically free encoding of information~\cite{korzekwa2022encoding}. Furthermore, we revealed that all transitions accessible via thermal operations can be accomplished using a restricted set of thermal operations that exclusively affect only two energy levels (of the system extended by a memory) at any given time. These findings carry important implications, not only for the development of efficient thermodynamic protocols, such as optimal cooling and Landauer erasure, but also for the exploration of novel avenues of research focused on characterising memory effects in thermodynamics. Finally, we also commented on the role played by the memory system as a free energy storage that enables non-Markovian effects.
	
	Our results offer many possibilities for generalisation and further research. First, one can try proving that the future thermal cone for memory-assisted Markovian thermal processes agrees with that of thermal operations in the limit of infinite memory, \mbox{$\lim_{N\rightarrow\infty} C^+_{\text{MeMTP}} = C^+_{\text{TO}}$}, as suggested by Conjecture~\ref{Conj:beta_permutations}. This can be built upon the proofs for $\beta$-swaps (Theorem \ref{Thm:beta-swap}) and $\beta$-$3$-cycles (Theorem \ref{Thm:beta-3-cycle}) presented in this work. Second, one may also attempt to show that the convergence of the proposed protocols $\mathcal{P}^{\Pi}$ and $\widetilde{\mathcal{P}}^{\Pi}$ is optimal with respect to the memory size. In other words, one could investigate the upper-bound on the power of memory-assisted Markovian thermal processes with a given size of memory $N$. Third, from a more practical point of view, it may be worthwhile to explore MeMTPs involving finite and infinite memory with non-trivial energy level structure. The practical relevance of this direction can be understood by considering the introduction of non-degenerate splitting of the levels for the full system, which would allow the level pairs to be addressed independently.
	
	In addition to the above, there are also less clear-cut goals for future efforts, such as expanding the studies beyond energy-incoherent states into the full range of quantum states. Furthermore, while our work focused on a single main system, an interesting avenue for future work could be to investigate many non-interacting subsystems. This extension could shed light on the combined consequences of finite-size and memory effects, providing valuable insights into the behaviour of larger, more complex systems. Specifically, characterising such effects could help to identify strategies for improving the efficiency of thermodynamic protocols in practical applications. Finally, one can also consider memory composed of many equivalent systems (such as a multi-qubit memory), and analyse the potential challenges arising from energy-level degeneration in such a setting.
	
	Finally, the feasibility of the introduced algorithm can be studied from a control perspective, following the approaches outlined in~\cite{wolpert2019space, PRXQuantum.4.010332}. The first approach introduces the notion of a space-time trade-off, which refers to the minimal amount of memory and time steps required to classically implement a given process. The second approach deals with control complexity, defining it as the number of levels a given operation non-trivially acts versus the time steps needed to implement that process. Our algorithm has specific time and memory requirements, namely $N$-dimensional memory and $N^2$ time steps. Furthermore, it is limited to the simplest two-level processes at any given time, meaning its control complexity is as low as possible. Nonetheless, future work might explore variations of our protocol (or any of the variants presented in Appendix~\ref{App:protocols}). Such explorations could focus on enabling parallelisation of certain steps by expanding available memory, thereby illustrating the space-time trade-off.

	\begin{acknowledgments}
		KK would like to thank Paul Skrzypczyk for asking a stimulating question at QIP2022 that started this project. The authors acknowledge financial support from the Foundation for Polish Science through the TEAM-NET project (contract no. POIR.04.04.00-00-17C1/18-00). During the process of preparing this manuscript, we were made aware of the overlap with Ref.~\cite{jeongraknelly}, and we thank Jeongrak Son and Nelly~H.~Y.~Ng for insightful discussions on this topic.
	\end{acknowledgments}
	
	\appendix
	
	
	\section{Thermodynamic evolution of energy-incoherent states}\label{App:partial-order}
		
		This work focuses on energy-incoherent states that can be represented by $d$-dimensional probability vectors of their eigenvalues, and their evolution is described by stochastic matrices. As a result, this specificity allows one to replace density operators and quantum channels with probability vectors and stochastic matrices, respectively. 
		Consequently, given two states, $\rho$ and $\sigma$, with eigenvalues $\v p$ and $\v q$, the existence of a thermal operation between $\rho$ and $\sigma$ is equivalent to the existence of a Gibbs-preserving stochastic matrix between $\v p$ and $\v q$~\cite{Janzing2000,horodecki2013fundamental}. Surprisingly, determining whether a given state $\v p$ can be transformed into a target state $\v q$ reduces to checking a finite list of conditions expressed by a partial-order relation between the initial and target state. In the infinite-temperature limit, these rules are encoded by the majorisation relation~\cite{marshall1979inequalities}, and in the finite temperature case by thermomajorisation~\cite{Rusch1978,horodecki2013fundamental}. If one further constrains the evolution to be Markovian, then the aforementioned rules are expressed by the notion of continuous thermomajorization~\cite{Lostaglio2019}. In this appendix, we review and summarise well-known results concerning these partial order relations.
		
		\subsection{Majorisation}
		
		To formulate the solution underlying the thermodynamic evolution of energy-incoherent states under thermal operation, we first need to recall the concept of majorisation~\cite{marshall1979inequalities} (see also Ref.~\cite{Lostaglio2019} for a detailed discussion).
		
		\begin{defn}[Majorisation]\label{def_Majorisation} Given two $d$-dimensional probability distributions $\v p$ and $\v q$, we say that $\v{p}$ \emph{majorises} $\v{q}$, and denote it by $\v p \succ \v q$, if and only if the following condition holds:
			\begin{equation}
				\label{eq_majorisation}
				\sum_{i=1}^k p_i^{\downarrow}\geq\sum_{i=1}^k q_i^{\downarrow} \quad \text{for all} \quad  k\in\{1\dots d\},
			\end{equation}
			where $\v{p}^{\downarrow}$ denotes the vector $\v{p}$ rearranged in a non-increasing order. 
		\end{defn}
		Equivalently, the majorisation relation between distributions can be expressed in terms of the \emph{majorisation curve}. Given a distribution $\v p$, we define a piecewise linear curve $f_{\v p}(x)$ in $\mathbb{R}^2$. This curve is obtained by joining the origin $(0,0)$ and the points $\left(\sum_{i=1}^{k} \eta_i, \sum_{i=1}^{k} p^{\downarrow}_i \right)$, for each $k$ in the set $\{1, ..., d\}$. Then, $\v{p}$ majorises $\v{q}$ if, and only if, the majorisation curve $f_{\v p}(x)$ of $\v{p}$ is always above that of $\v{q}$:
		\begin{equation}
			\v p \succ \v q \iff \forall x\in \left[0,1\right]:~f_{\v p}(x) \geq  f_{\v q}(x) \, .
		\end{equation}
		
		Majorisation can be interpreted as a formalisation of the notion of disorder with respect to the uniform distribution $\v \eta$. Note that the uniform state $\v \eta$ is majorised by any other distribution, while every distribution is majorised by the sharp state $\v s = (1, 0, ..., 0)$. Furthermore, if $\v p \succ \v q$, then all R\'{e}nyi entropies associated with $\v p$ are smaller than those of $\v q$~\cite{horodecki2013quantumness}. 
		
		\begin{figure}[t]
			\centering
			\includegraphics{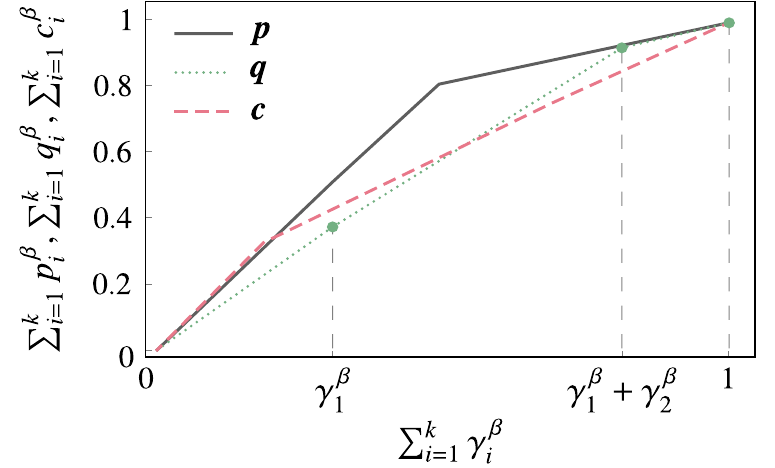}
			\caption{\label{fig-thermomajorisationcurves} \emph{Thermomajorisation curve}. For three different states $\v p$, $\v q$ and $\v c$, with a thermal Gibbs state $\v \gamma \propto (1, e^{-\beta}, e^{-2\beta})$, and $\beta >0$, we plot their thermomajorisation curves $f^{\, \beta}_{\v p}(x)$ [black curve], $f^{\, \beta}_{\v q}(x)$ [dashed red curve] and $f^{\, \beta}_{\v c}(x)$ [dot green curve], respectively. While $\v p$ thermomajorises $\v q$ [since $f^{\, \beta}_{\v p}(x)$ is never below $f^{\, \beta}_{\v q}(x)$], both states are incomparable with $\v c$, as their thermomajorisation curves cross with $f^{\, \beta}_{\v c}(x)$.}
		\end{figure}
		
		The most general transformation between two probability distributions, $\v{p}$ and $\v{q}$, is accomplished by a stochastic matrix, where $\Lambda$ fulfills $\Lambda_{ij} \geq 0$ and $\sum_{i} \Lambda_{ij} = 1$. In the infinite-temperature limit, thermodynamic transformations are accomplished by stochastic matrices that preserve the uniform distribution $\v{\eta}$. Mathematically, this means that $\Lambda$ is bistochastic, i.e., it additionally satisfies $\sum_j \Lambda_{ij} = 1$. As every bistochastic matrix can be written as a convex combination of permutation matrices, this implies that the set of $d \times d$ bistochastic matrices is a convex polytope with $d!$ vertices, one for each permutation in $\mathcal{S}_d$. Finally, since the existence of a bistochastic matrix connecting $\v{p}$ to $\v{q}$ is equivalent to $\v{p} \succ \v{q}$, we conclude that the set of states $C^{\textrm{TO}}_{+}(\v p)$ achievable via thermal operations from a given incoherent initial state $\v{p}$ is given by
		\begin{equation}\label{Eq:futurcone}
			C^{\textrm{TO}}_{+}(\v p) = \operatorname{conv}\left[\left\{\Pi \v p \, , \mathcal{S}_d \ni \pi \mapsto \Pi  \right\}\right] ,
		\end{equation}
		where $\Pi$ denotes a permutation matrix corresponding to the permutation $\v \pi$ with $d$ elements, and $\operatorname{conv[S]}$ the convex hull of the set $S$. A detailed discussion about the construction of the future $C^{\textrm{TO}}_{+}$ and its properties can be found in Sec. III of Ref.~\cite{deoliveirajunior2022}.

		\subsection{Thermomajorisation}\label{app:thermomajorisation}
		
		Thermomajorisation is a measure of disorder relative to the thermal distribution $\v \gamma$. Mathematically, given $\v \gamma$, we say that $\v p$ \emph{thermomajorise} $\v q$, and denote by $\v p \succ_{\beta}\v q$, if there exists a stochastic matrix $\Lambda^{\v \beta}$, which leaves the vector $\v \gamma$ invariant and maps $\v p$ onto $\v q$.
		
		The starting point of defining the thermodynamic equivalent of the majorisation is done by presenting the concept of $\beta$-ordering~(see Definition~\ref{Def-beta-ordering}), and then introducing the notion of thermomajorisation curves~\cite{horodecki2013fundamental}. Recall that the $\beta$-ordering of $\v p$ is defined by a permutation $\v \pi_{\v p}$ that sorts $p_i/\gamma_i$ in a non-increasing order. The $\beta$-ordered version of $\v p$, denoted as $\v p^\beta$, is obtained by arranging the elements of $\v p$ according to the permutation $\v \pi_{\v p}$
		\begin{equation}
			\v{p}^{\, \beta}=\left(p_{\v \pi_{\v{p}}^{-1}(1)},\dots ,p_{\v \pi_{\v{p}}^{-1}(d)}\right).
		\end{equation}
		
		A thermomajorisation curve is defined as a piecewise linear curve composed of linear segments connecting the point $(0,0)$ and the points defined by consecutive subsums of the $\beta$-ordered form of the probability $\v{p}^\beta$ and the Gibbs state $\v{\gamma}^\beta$,
		\begin{equation}
			\left(\sum_{i=1}^k\gamma^{\, \beta}_i,~\sum_{i=1}^k p^{\, \beta}_i\right):=\left(\sum_{i=1}^k\gamma_{\v \pi^{-1}_{\v{p}}(i)},~\sum_{i=1}^k p_{\v \pi^{-1}_{\v{p}}(i)}\right),
		\end{equation}
		for $k\in\{1,\dots,d\}$. Thus, given two $d$-dimensional probability distributions $\v p$ and $\v q$, and a fixed inverse temperature $\beta$, we say that $\v p$ \emph{thermomajorises} $\v q$ and denote it as $\v p \succ_{\beta} \v q$, if the thermomajorisation curve $f^{\, \beta}_{\v{p}}$ is above $f^{\, \beta}_{\v{q}}$ everywhere, i.e.,
		\begin{equation}
			\v p \succ_{\beta} \v q \iff \forall x\in[0,1]:~ f^{\, \beta}_{\v{p}}(x) \geq f^{\, \beta}_{\v{q}}(x) \, .
		\end{equation}
		See Fig.~\ref{fig-thermomajorisationcurves} for an example considering a three-level system. 
		
		For finite temperatures, general thermodynamic transformations between two probability distributions, $\v{p}$ and $\v{q}$, are accomplished by stochastic matrices preserving the Gibbs state, i.e., matrices $\Lambda$ such that $\Lambda \v{\gamma} = \v{\gamma}$, which are commonly referred to as Gibbs-preserving (GP) matrices in the literature. As the rules governing state transformations are no longer described by a majorisation relation, but instead by thermomajorisation, the characterisation of the future thermal cone is no longer given by Eq.~\eqref{Eq:futurcone}. Nevertheless, the set of Gibbs-preserving matrices is still a convex set, and the extreme points of the future thermal cone can be obtained by the following lemma.
		\begin{lem}[Lemma 12 of Ref.~\cite{Lostaglio2018elementarythermal}]
			\label{lem_extreme}
			Given $\v{p}$, consider the following distributions $\v{p}^{\, \v \pi}\in C^{\textrm{TO}}_{+}(\v p)$ constructed for each permutation $\v \pi\in \S_d$. For $i\in\left\{1,\dots,d\right\}$:
			\begin{enumerate}
				\item Let $x_i^{\v \pi}=\sum_{j=0}^{i} e^{-\beta E_{\v \pi^{-1}\left(j\right)}}$ and $y_i^{\v \pi}=f^{\, \beta}_{\v{p}}\left(x_i^{\v \pi}\right)$. 
				\item Define $p^{\v \pi}_i:=y^{\v \pi}_{\v \pi(i)} - y^{\v \pi}_{\v \pi(i)-1}$, with $y_{0}:=0$.
			\end{enumerate}
			Then, all extreme points of $C^{\textrm{TO}}_{+}(\v p)$ have the form $\v{p}^{\v \pi}$ for some ${\v \pi}$. In particular, this implies that $C^{\textrm{TO}}_{+}(\v p)$ has at most $d!$ extremal points.
		\end{lem}
		The above lemma allows one to characterise the future thermal cone of $\v p$ by constructing states $\v p^{\v \pi}$ for each $\v \pi \in \mathcal S_d$, and taking their convex hull.

		\subsection{Continuous thermomajorisation}\label{app:continuousmajorisation}
		If we enforce that the transformation is Markovian, the conditions are captured by the notion of continuous thermomajorisation~\cite{lostaglio2021continuous}. Formally, one asks if there exists a continuous path within the probability simplex that connects these two distributions, such that the preceding distribution is thermomajorised by the succeeding one at any two points along this path. Such a notion is defined as follows. 
		\begin{defn}[Continuous thermomajorisation] \label{def:markov_majo}
			A distribution $\v{p}$ \emph{continuously thermomajorises} $\v{q}$, denoted $\v{p} \ggcurly_\beta \v{q}$, if there exists a continuous path of probability distributions $\v{r}(t)$ for $t\in[0,t_f)$ such that
			\begin{enumerate}
				\item $\v{r}(0)=\v{p}$,
				\item $\forall~ t_1,t_2\in[0,t_f):\quad t_1 \leq t_2 \Rightarrow\v{r}(t_1)\succ_\beta \v{r}(t_2)$,
				\item $\v{r}(t_f)=\v{q}$.
			\end{enumerate}
			The path $\v{r}(t)$ is called \emph{thermomajorising trajectory} from $\v{p}$ to~$\v{q}$.
		\end{defn}
		Let us make a few comments about the above definition. Firstly, it is worth noting that when $\beta =0$ and the thermal state is replaced by the uniform fixed point, $\v{\gamma}=\v{\eta}$, the above definition corresponds to a continuous version of standard majorisation, denoted by $\ggcurly$. Secondly, determining whether a given initial state continuously thermomajorises a target state is a difficult problem, and unlike the other variants of majorisation presented so far, there is no continuous thermomajorisation curve for this type of majorisation that would facilitate a quick check. Nevertheless, the necessary and sufficient conditions are known~\cite{Lostaglio2019}. These comprise a complete set of entropy production inequalities that can be reduced to a finitely verifiable set of constraints. Lastly, the notion of continuous thermomajorisation encapsulates all constrains of memoryless thermal processes on population dynamics.
	
	\section{Regularised incomplete beta function}
	\label{app:beta}
	
	The content of this appendix is based on Refs.~\cite{artin2015gamma,NIST:DLMF}.
	
	\subsection{Definition and properties}
	
	The finite-size corrections to a $\beta$-swap and its compositions are determined by the cumulative distribution function (CDF) known as the regularised incomplete beta function. This function is closely related to the well-known beta function $B(a,b)$ and is widely used in deriving our results. To provide the necessary background and present its key properties, we first recall the definition and properties of the beta function:
	\begin{equation}\label{Eq:beta_function}
		B(a,b) = \int_{0}^{1} t^{a-1} (1-t)^{b-1} dt ,
	\end{equation}
	with $a,b \in \mathbbm{C}$. The beta function relates to the gamma function in the following way
	\begin{equation}\label{Eq:beta_and_gamma_function}
		B(a,b) = \frac{\Gamma(a)\Gamma(b)}{\Gamma(a+b)}.
	\end{equation}
	
	The incomplete beta function $B_x(a,b)$ is defined by changing the upper limit of integration in Eq.~\eqref{Eq:beta_function} to an arbitrary variable, i.e.,
	\begin{equation}\label{Eq:incomplete_beta_function}
		B_x(a,b) = \int_{0}^{x} t^{a-1} (1-t)^{b-1} dt.
	\end{equation}
	Finally, we define the regularised incomplete beta function $I_x(a,b)$ (regularised beta function for short) by normalising the incomplete beta function,
	\begin{equation}
		I_x(a,b) = \frac{B_x(a,b)}{B(a,b)}.
	\end{equation}
	We present plots of the regularised beta function for a few selected values of $x$ in Fig.~\eqref{Fig:regularised_beta_function}.
	
	Throughout this work, we assume that $a, b >0$ and \mbox{$0\leq x \leq 1$}. It is easily noted that $I_0(a,b) = 0$, $I_1(a,b) = 1$, and $I_0(a,b) \leq I_x(a,b) \leq I_1(a,b)$, thus making it a proper CDF. Furthermore, for $a,\,b\in\mathbb{Z}$, $I_x(a,b)$ can be written in terms of a binomial function
	\begin{equation}\label{Eq:incomplete_beta_function_binomial}
		I_x(a,b) = (1-x)^b\sum_{j=a}^{\infty}\binom{b+j-1}{j}x^{j}.
	\end{equation}
	From this equation, by using the geometric series and its derivatives, one can conclude that
	\begin{equation}
		\label{Eq:incomplete_beta_function_binomial_a_0}
		I_x(0,b) = 1.
	\end{equation}
	\begin{figure}[t]
		\centering
		\includegraphics{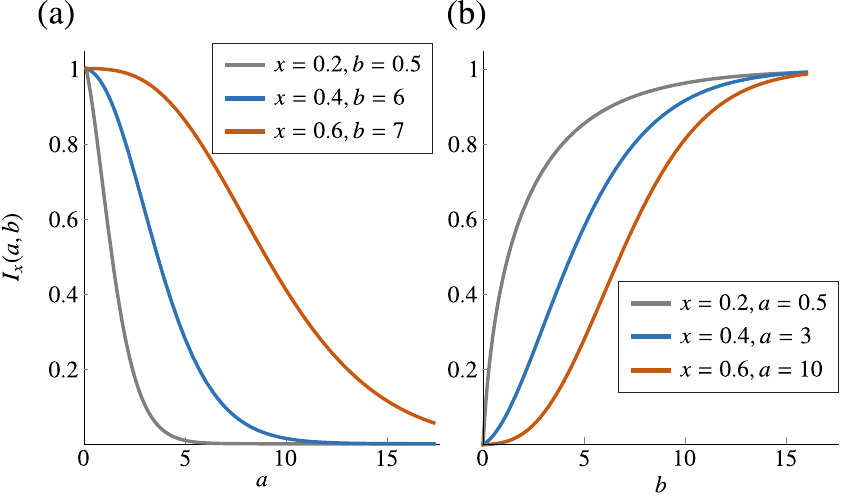}
		\caption{\textbf{Regularised beta function}. Plots of the regularised incomplete beta function as a function of (a) $a$ and (b) $b$. }
		\label{Fig:regularised_beta_function}
	\end{figure}
	Moreover,
	\begin{equation}\label{Eq:incompleta-beta-function-difference}
		(1-x)^b\sum_{j=a}^{n}\binom{b+j-1}{j}x^{j} = I_x(a,b)-I_x(n+1,b).
	\end{equation}
	If $a =1$, then $I_x(1,b) = 1 - (1-x)^b$ and Eq.~\eqref{Eq:incomplete_beta_function_binomial} simplifies to 
	\begin{equation}
		(1-x)^b\sum_{j=1}^{n}\binom{b+j-1}{j}x^{j} = 1 - (1-x)^b - I_x(n+1,b).
	\end{equation}
	The next two useful properties of $I_x(a,b)$ are the symmetry relation 
	\begin{equation}\label{Eq:incomplete_beta_function_symmetry}
		I_x(a,b) = 1- I_{1-x}(b,a),
	\end{equation}
	and the relation for an equal argument,
	\begin{equation}\label{Eq:incomplete_beta_function_recurrence}
		I_x(a,a) = \frac{1}{2}I_{4x(1-x)}\qty(a,\frac{1}{2}),
	\end{equation}
	when $0\leq x \leq 1/2$.
	Finally, there are two recurrence relations which allow one to shift either of the arguments of the function by one,
	\begin{subequations} 
		\begin{align}\label{Eq:incomplete_beta_function_recurrence2}
			I_x(a,b) &= I_x(a+1,b)+ \frac{x^a(1-x)^b}{a B(a,b)}, \\ \label{Eq:incomplete_beta_function_recurrence3}
			I_x(a,b) &= I_x(a,b+1)- \frac{x^a(1-x)^b}{a B(a,b)}, \\
			I_x(a,b) &= I_x(a+1,b-1)+\frac{x^a (1-x)^{b-1}}{aB(a,b)}  \label{Eq:incomplete_beta_function_recurrence4}, \\
			I_x(a,b) &= I_x(a-1,b+1)-\frac{x^{a-1} (1-x)^{b}}{aB(a,b)}.  \label{Eq:incomplete_beta_function_recurrence5}
		\end{align}
	\end{subequations}
	
	
	\subsection{General relations involving the gamma function}
	
	Next, let us present general relations and properties of the gamma function that are extensively used in our proofs. First, recall that for every positive integer $n$
	\begin{equation}
		\Gamma(n) = (n-1)!. 
	\end{equation}
	Using the above, one can derive a simple formula for the following expression that appears when dealing with the regularised beta function:
	\begin{equation}\label{Eq:beta-function-expansion-binomial}
		\frac{1}{n B(n,n)} = \frac{\Gamma(2n)}{n \Gamma(n)^2} = \frac{\Gamma(2n)}{n \Gamma(n)^2} \frac{2n}{2n} = \frac{1}{2}\binom{2n}{n}.
	\end{equation}
	Other important functional equation for the gamma function is the Legendre duplication formula:
	\begin{equation}\label{Eq:Gamma-function-legendre}
		\Gamma(n)\Gamma\qty(n+\frac{1}{2}) = 2^{1-2n}\sqrt{\pi}\Gamma(2n).    
	\end{equation}
	The factorial terms can be approximated using Stirling's approximation
	\begin{equation}\label{Eq:Stirlings-approximation}
		n! = \sqrt{2\pi n}\qty(\frac{n}{e})^n\left(1+O(n^{-1})\right).
	\end{equation}
	
	
	\subsection{Asymptotic analysis}
	
	We will be interested in the asymptotic behaviour of the regularised beta function. So, in this section, we introduce important relations and identities that will be useful for proving our main theorems. Let us begin by considering the regularised beta function $I_x(a,b)$, where $x$ and $b$ fixed. For $a \to \infty$, we have the following asymptotic expansion:
	\begin{align}\label{Eq:regularised-beta-asymptotic}
		I_x(a,b) = &\:\:\Gamma(a+b)x^a(1-x)^{b-1}\nonumber\\ &\:\:\times \Biggr[\sum_{k=0}^{n-1}\frac{1}{\Gamma(a+k+1)\Gamma(b-k)}\qty(\frac{x}{1-x})^k \nonumber\\ &\:+ O\qty(\frac{1}{\Gamma(a+n+1)})\Biggr].
	\end{align}
	Note that $O$-term vanishes in the limit only if $n \geq b$. Furthermore, for each $n=0,1,2, ...$ If $b = 1, 2, 3, ...$, and $n>b$,  the $O$-term can be omitted, as the result is exact. In this work, Eq.~\eqref{Eq:regularised-beta-asymptotic} will be expanded up to the second order. Specifically, for the values of $a = N, b= 1/2$, we have the following equation
	\begin{align}\label{Eq:regularised-beta-expansion-12}
		\! I_{x}\qty(N,\frac{1}{2}) &\simeq\! \frac{x^N}{\sqrt{1-x}}\qty[\frac{\Gamma\qty(N+\frac{1}{2})}{\Gamma\qty(N+1) \Gamma\qty(\frac{1}{2})}+\frac{\Gamma\qty(N+\frac{1}{2})}{\Gamma\qty(N+2) \Gamma\qty(-\frac{1}{2})}\qty(\frac{x}{1-x})] \nonumber \\
		&=\!\frac{x^N}{\sqrt{1-x}}\frac{(2N)!}{4^N (N!)^2}\qty[1-\frac{x}{2(N+1)(1-x)}],
	\end{align}
	where $\simeq$ symbol hides the terms of the order $O(x^N/N^2)$.
	
	Next, we will consider sums of the regularised beta functions over the second argument and their limit as \mbox{$a \to \infty$}. We start with
	\begin{align}
		\sum_{i=1}^{a}I_x(a,i+1) &= \sum_{i=1}^a \frac{B_x(a,i+1)\Gamma(a+1+i)}{\Gamma(a)\Gamma(i+1)} \nonumber \\
		&= \sum_{i=1}^a \frac{(a+i)!}{(a-1)! i!}\int_{0}^x dt \: t^{a-1}(1-t)^i  \nonumber \\
		&= \int_{0}^x dt \sum_{i=1}^a \binom{a+i}{i}t^{a-1}(1-t)^i,
	\end{align}
	where we have used definitions of the beta function and the gamma function for integer arguments. Using Eq.~\eqref{Eq:incompleta-beta-function-difference}, the above expression can be recast as
	\begin{align}
		\sum_{i=1}^{a}I_x(a,i+1) &= \int_{0}^x \frac{dt}{t^2}\:[I_x(1,a-1)-I_{1-x}(a-1,a-1)] \nonumber \\
		&= \int_{0}^x \frac{dt}{t^2}\:\qty[\frac{1}{2}I_{4x(1-x)}\qty(a-1,\frac{1}{2})-I_x(a-1,1)]
	\end{align}
	for $0\leq x \leq 1/2$.
	
	Finally, using the asymptotic expansion from~Eq.~\eqref{Eq:regularised-beta-asymptotic}, we get
	\begin{align}
		\int_{0}^x dt\: \Biggl\{-t^{a-3}&+\frac{1}{2x^2}\Biggl[\frac{\Gamma(a-1/2)}{\Gamma(a+1)}\frac{[4x(1-x)]^{a-1}}{\sqrt{1-4x(1-x)}} + \nonumber\\ & O\qty(\frac{(4x(1-x))^{a-1}\Gamma(a-1/2)}{\Gamma(a+1)}) \Biggl]\Biggl\}.
	\end{align}
	Thus, we see that this term vanishes in the limit of $a \to \infty$, and we find that the following sum vanishes in the limit:
	\begin{equation}\label{Eq:beta-regularised-limit-sum}
		\lim_{a\to\infty} \sum_{i=1}^{a}I_x(a,i+1) = 0. 
	\end{equation}
	
	Moreover, we will need the following sum:
	\begin{align}
		\sum_{i=1}^N &I_{x}(N+1-i,N)= \sum_{i=1}^N\frac{B_{x}(N+1-i,N)}{B(N+1-i,N)}\nonumber\\
		=&\sum_{i=1}^N\frac{N\Gamma(2N+1-i)}{N\Gamma(N+1-i)\Gamma(N)}\int_{0}^{x}dt\, t^{N-i}(1-t)^{N-1} \times \nonumber\\ 
		&\times \int_{0}^{x}dt N(1-t)^{N-1}\sum_{i=1}^N\binom{2N-i}{N-i}t^{N-i}\nonumber\\
		=& \int_{0}^{x}dt N\frac{(1-t)^{N-1}}{(1-t)^{N+1}}(1-t)^{N+1}\sum_{j=0}^{N-1}\binom{(N+1)+j-1}{j}t^{j}\nonumber\\
		=&\int_{0}^{x}dt \frac{N}{(1-t)^2}\Bigg(I_{1-x}(0,N+1)-I_{1-x}(N,N+1)\Bigg)\nonumber\\
		=& \int_{0}^{x}dt \frac{N}{(1-t)^2}\Bigg(1-I_{1-x}(N,N)-\frac{t^N(1-t)^N}{NB(N,N)}\Bigg)\nonumber\\
		=& N \int_{0}^{x}dt\frac{1}{(1-t)^2}=N \frac{x}{1-x}. \label{Eq:regularised-beta-function-sum-2}
	\end{align}
	
	
	\section{Proofs of Lemma~\ref{Lem:transposition} and Theorem~\ref{Thm:beta-swap}}~\label{app:neighbour}
	
	To prove Lemma~\ref{Lem:transposition} and Theorem~\ref{Thm:beta-swap}, we will consider a composite system consisting of the main $d$-dimensional system and an $N$-dimensional memory system. Without loss of generality, we can assume that the main system is a two-level system with $i= 1$ and $j = 2$, whose state is described by unnormalised probability vector. We begin by deriving an expression that describes how the composite system evolves under the memory-assisted protocol $\mathcal{P}^{(12)}$. Next, to gain insight into the behaviour of the joint system as the memory size grows, we will prove the asymptotic result. Finally, in the last subsection, we will show how this result implies the desired convergence rates.
	
	
	\subsection{Dynamics induced by two-level thermalisations}
	
	Consider a two-level system, described by a Hamiltonian $H = E_1 \ketbra{E_1}{E_1} + E_2 \ketbra{E_2}{E_2}$ and initially prepared in an energy-incoherent state \mbox{$\v p = (b,c)$}, together with a memory system, described by a trivial Hamiltonian $H_M = 0$ and prepared in a maximally mixed state $\v{\eta}_M = \qty(1, \hdots, 1)/N$. The joint state $\v r:= \v p \otimes \v \eta_M$ of the composite system is then given by
	\begin{equation}\label{Eq:joint-state}
		\v{r}^{(0)} \equiv\v{r} = \frac{1}{N}\Big(\underbrace{b,\hdots,b}_{N\, \text{times}}|\underbrace{c ,\hdots,c }_{N\, \text{times}}\Big),
	\end{equation}
	whereas the joint thermal distribution \mbox{$\v \Gamma := \v \gamma \otimes \v \eta_M$} is given by
	\begin{align}
		\v \Gamma &= \frac{1}{N(e^{-\beta E_1}+ e^{-\beta E_2})}\qty[e^{-\beta E_1}, \cdots , e^{-\beta E_1} |e^{-\beta E_2}, \cdots , e^{-\beta E_2}] \nonumber \\ &=\frac{1}{N}\qty(\gamma_1, ..., \gamma_1 | \gamma_2, ..., \gamma_2). 
	\end{align}
	For the sake of brevity, we introduce the rescaled Gibbs factors, which will be used extensively in the proofs:
	\begin{equation}\label{Eq:rescaled-Gibbs-factor}
		\Gamma_{ij} =\frac{\gamma_i}{\gamma_i+\gamma_j}.
	\end{equation}
	
	Under a series of two-level thermalisations, the joint state of the composite system at the $k$-th round is given by
	\begin{equation}
		\!\!\! \mathcal{R}_k\qty(\v{r}^{(k-1)})\equiv \v{r}^{(k)} \! =\! \qty(b_1^{(k)},\hdots, b_N^{(k)}\,|\,c_1^{(N)},\hdots,c^{(N)}_{k},c,\hdots,c),
	\end{equation}
	where $b^{(k)}_j$ and $c^{(N)}_j$ satisfy the following recurrence relations
	\begin{align}\label{Eq:recurrence-relation-entries_bj}
		b^{(k)}_{j} &= \Gamma^{j}_{21}\,\qty(\frac{\Gamma_{12}}{\Gamma_{21}}c+\Gamma_{12}\sum_{i=1}^j b^{(k-1)}_i \Gamma^{\, -i}_{21}), \\
		c^{(N)}_j &= \Gamma_{21}^{\, N}c+\Gamma_{21}^{\, N+1}\sum_{i=1}^N b^{(j-1)}_i \Gamma^{\, -i}_{21},
		\label{Eq:recurrence-relation-entries_cj}
	\end{align}
	with $b^{(0)}_j = b$. Equations \eqref{Eq:recurrence-relation-entries_bj} and \eqref{Eq:recurrence-relation-entries_cj} can be understood by noting that during the $k$-th round of two-level thermalisations, an additional $c$ is added to the previous $(k-1)$-th entry and the resulting state is again thermalised. By iterating this process for the first $k$ rounds, we can derive a closed-form expression for the entries $b^{(k)}_j$ and $c^{(N)}_j$:
	\begin{align}\label{Eq:bj_entries}
		\!\!\! b_j^{(k)} &=\Gamma_{12}\Gamma_{21}^{j-1}c\sum_{i=0}^{k-1}\binom{j\! +\! i\! -\! 1}{i}\Gamma_{12}^{i}+b\Gamma_{12}^{k}\sum_{i=1}^{j}\binom{j\! +\! k\! -\! 1\! -\! i}{k\! -\! 1}\Gamma_{21}^{j-i}, \\
		\!\!\! c^{(N)}_j &= c\Gamma^N_{21} \sum_{i=0}^{j-1}\binom{N\! +\! i\! -\! 1}{i}\Gamma_{12}^{i}+b\Gamma_{21}\Gamma^j_{12}\sum_{i=0}^{N-1} \binom{i\! +\!  j}{j}\Gamma_{21}^{i}.\label{Eq:cj_entries}
	\end{align}
	After $N^2$ rounds of two-level thermalisations, the composite final state $\v r^{(N)}$ is given by
	\begin{align}
		\widetilde{\mathcal{P}}^{(12)}(\v r) & \equiv\v r^{(N)} = \v q \otimes \v \eta_M \nonumber\\ &=\frac{1}{N}{\small\qty[b^{(N)}_1,\hdots,b^{(N)}_N \, \bigg| \, c^{(N)}_{1},\hdots, c^{(N)}_{N}]}.
	\end{align}
	
	
	\subsection{Infinite memory limit}
	\label{app:limit_proofs}
	
	Up till now, we have obtained a closed-form expression describing the action of the truncated protocol $\widetilde{\mathcal{P}}^{(12)}$. In order to prove Lemma~\ref{Lem:transposition}~and~Theorem~\ref{Thm:beta-swap}, the next step consists of decoupling the main system from memory using a full thermalisation $\mathcal{T}$:
	\begin{equation}
		\mathcal{P}^{(12)}(\v p \otimes \v \gamma_M) = \mathcal{T} \circ \widetilde{\mathcal{P}}^{(12)}(\v p \otimes \v \gamma_M). 
	\end{equation}
	Since we are initially focused on demonstrating the asymptotic results, our aim is to show that as $N$ approaches infinity, we achieve a $\beta$-swap:
	\begin{subequations}
		\begin{align}
			\lim_{N\to \infty} \frac{\sum_{i=1}^{N} b^{(N)}_i}{N} &= c+b(1-e^{-\beta (E_2-E_1)}),\label{Eq:limit_beta_swap1}\\ \lim_{N\to \infty} \,\frac{\sum_{i=1}^{N} c_i^{(N)}}{N} &= be^{-\beta (E_2-E_1)} .\label{Eq:limit_beta_swap2}
		\end{align}
	\end{subequations}
	To prove the limits in Eqs.~\eqref{Eq:limit_beta_swap1} and \eqref{Eq:limit_beta_swap2}, we will begin by using the conservation of probability,
	\begin{equation}\label{Eq:simple_fact}
		\lim_{N\to \infty} \sum_{i=1}^{N}\frac{ b^{(N)}_i}{N}+ \lim_{N\to \infty}\sum_{i=1}^{N}\frac{ c_{i}^{(N)}}{N} = b+c.   
	\end{equation}
	This allows us to express the first limit as a function of the second:
	\begin{equation}\label{Eq:first_limit}
		\lim_{N\to \infty} \sum_{i=1}^{N}\frac{ b^{(N)}_i}{N} = b+c - \lim_{N\to \infty}\sum_{i=1}^{N}\frac{ c^{(N)}_{i}}{N} .
	\end{equation}
	
	We will now examine the non-trivial term on the right hand side of Eq.~\eqref{Eq:first_limit}. We start by using Eq.~\eqref{Eq:incompleta-beta-function-difference} to re-write the first term of $c^{(N)}_j$ appearing in Eq.~\eqref{Eq:cj_entries} in terms of the regularised beta function
	\begin{align}
		\eval{c^{(N)}_j}_{\substack{b = 0\\c=1}} = \Gamma^N_{21} \sum_{i=0}^{j-1}\binom{N+i-1}{i}\Gamma_{12}^{i} &= [I_{\Gamma_{12}}(0,N)-I_{\Gamma_{12}}(j,N)] \nonumber \\ &= I_{\Gamma_{21}}(N,j),
	\end{align}
	where we also used Eq.~\eqref{Eq:incomplete_beta_function_symmetry} to invert the arguments of the regularised beta function. Thus, using \eqref{Eq:beta-regularised-limit-sum}, we conclude that
	\begin{equation}
		\lim_{N\to\infty}\frac{1}{N}\sum_{i=1}^{N}\eval{c^{(N)}_j}_{\substack{b = 0\\c=1}} = \lim_{N\to\infty}\frac{1}{N}\sum_{i=1}^{N}I_{\Gamma_{21}}(N,i) = 0.    
	\end{equation}
	Taking into account the second term of $c^{(N)}_j$ appearing in Eq.~\eqref{Eq:cj_entries}, one can immediately evaluate the sum
	\begin{align}
		\lim_{N\rightarrow\infty}\eval{c^{(N)}_j}_{\substack{b = 1\\c=0}}  = \Gamma_{21}\Gamma^j_{12}\sum_{i=0}^{\infty} \binom{i+k}{k}\Gamma_{21}^{i} = \frac{\Gamma_{21}}{\Gamma_{12}},
	\end{align}
	and therefore obtain that
	\begin{equation}
		\lim_{N\to \infty} \,\frac{\sum_{i=1}^{N} c^{(N)}_i}{N} = be^{-\beta (E_2-E_1)}.    
	\end{equation}
	Substituting this result to Eq.~\eqref{Eq:first_limit}, we prove that
	\begin{equation}
		\lim_{N\to \infty} \frac{\sum_{i=1}^{N} b^{(N)}_i}{N} = c+b[1-e^{-\beta (E_2 - E_1)}].
	\end{equation}
	As a result, in the limit of $N\to\infty$, the protocol $\mathcal{P}^{(12)}$ achieves a $\beta$-swap. 
	
	
	\subsection{Finite memory convergence rates}
	\label{app:convergence_proofs}
	
	To complete the proofs of Lemma~\ref{Lem:transposition} and Theorem~\ref{Thm:beta-swap}, we will now analyse what the approximation error is for a finite size of the memory $N$. This will tell us how quickly the initial state convergences to a $\beta$-swap as a function of $N$. First, we define two functions governing the convergence
	\begin{subequations}
		\begin{align}\label{Eq:error_function1}
			\!\!\! \mathbb{E}:=& \frac{1}{N} \sum_{i = 1}^N\eval{c^{(N)}_{i}}_{c=1,b=0} \!\! = \frac{\Gamma^N_{21}}{N}\sum_{i=1}^{N}\sum_{j=0}^{i-1} \binom{N\! +\! j\! -\! 1}{j}\Gamma^j_{12} , \\
			\!\!\! \mathbb{F}:=& \frac{1}{N} \sum_{j = 1}^N\eval{b_j^{(N)}}_{c=0,b=1  } \!\! = \frac{\Gamma^N_{12}}{N}\sum_{j=1}^N \sum_{i=1 }^{j}\binom{j\! +\! N\! -\! 1\! -\! i}{N\! -\! 1}\Gamma^{j-i}_{21}\label{Eq:error_function2},\!
		\end{align}
	\end{subequations}
	which allow us to write the final state of the system as
	
	\begin{equation}
		\label{eq:finalstate}
		\v{q} = b\mqty(\mathbb{F} \\ 1-\mathbb{F} ) + c\mqty(1-\mathbb{E}\\\mathbb{E}).
	\end{equation}
	
	Next, we will asymptotically expand Eq.~\eqref{Eq:error_function1}. The starting point is to reduce the double sum into one, then convert the binomial sums into regularised beta functions and use its properties given by Eqs.~\eqref{Eq:incomplete_beta_function_recurrence4},~\eqref{Eq:incomplete_beta_function_recurrence}~and~\eqref{Eq:beta-function-expansion-binomial}:
	\begin{align}
		\mathbb{E}&= \frac{\Gamma^N_{21}}{N}\sum_{i=1}^N(N-i+1) \binom{N+i-2}{i-1}\Gamma_{12}^{i-1} \nonumber \\&= \Gamma^N_{21}\sum_{i=0}^{N-1}\binom{N+i-1}{i}\Gamma_{12}^{i} - \frac{\Gamma^{N}}{N}\sum_{i=1}^{N-1} i\: \binom{N+i-1}{i}\Gamma_{12}^{i} \nonumber \\
		&= I_{\Gamma_{21}}(N,N)-\frac{\Gamma_{12}}{\Gamma_{21}}I_{\Gamma_{21}}(N+1,N-1)\nonumber \\ &=\qty(1-\frac{\Gamma_{12}}{\Gamma_{21}})I_{\Gamma_{21}}(N,N)+\frac{\Gamma_{12}}{\Gamma_{21}}\frac{\Gamma_{21}^N \Gamma_{12}^{N-1}}{NB(N,N)} \nonumber\\
		&=\frac{1}{\Gamma_{21}}\qty[(\Gamma_{21}-\Gamma_{12})\frac{1}{2}I_{4\Gamma_{21}\Gamma_{12}}\qty(N,\frac{1}{2})+\frac{(\Gamma_{21} \Gamma_{12})^{N}}{N}\frac{\Gamma(2N)}{\Gamma(N)^2}] \label{Eq:E_function_before_exp}.
	\end{align}
	
	Now we expand Eq.~\eqref{Eq:E_function_before_exp} up to second order using Eq.~\eqref{Eq:regularised-beta-expansion-12}. Recall that such an expansion is an approximation up to terms of the order $O(x^N/N^2)$ (where $x$ will be given by \mbox{$4\Gamma_{21}\Gamma_{12}$}), which will be dropped since our final approximation will be up to the order $o(x^N/N^{3/2})$ or $o(1/N^{1/2})$ for finite and infinite temperatures, respectively. Then, the gamma functions appearing in the expansions are simplified using Eq.~\eqref{Eq:Gamma-function-legendre} and~Eq.~\eqref{Eq:beta-function-expansion-binomial}. Finally, using Stirling's approximation [Eq.~\eqref{Eq:Stirlings-approximation}], we arrive at 
	\begin{align}\label{Eq:convergence_E1}
		\mathbbm{E} = & \frac{(4\Gamma_{21}\Gamma_{12})^N}{2\Gamma_{21}\sqrt{\pi N}}\Bigg[\frac{(\Gamma_{21}-\Gamma_{12})}{\sqrt{1-4\Gamma_{21}\Gamma_{12}}}-\frac{(\Gamma_{21}-\Gamma_{12})(4\Gamma_{21}\Gamma_{12})}{2(1-4\Gamma_{21}\Gamma_{12})^{3/2}(N+1)} \nonumber\\    &\quad\qquad\qquad+1+o\left(N^{-1}\right)\Bigg][1+O(N^{-1})].
	\end{align}
	As the last step, we simplify the above expression for $\Gamma_{12}>1/2$ (finite temperature case) by using the fact that \mbox{$\Gamma_{12}=1-\Gamma_{21}$}, to arrive at
	
	\begin{equation}
		\!\!\!\! \!\!\eval{\mathbb{E}}_{\Gamma_{12}>\frac{1}{2}} \!\!\!\!\!\!\! =  (4\Gamma_{12}\Gamma_{21})^{N}\left[\frac{\Gamma_{12}}{(\Gamma_{12}\! -\!\Gamma_{21})^2\sqrt{\pi N}(N\! +\! 1)}+o\left( N^{-3/2}\right)\right].
	\end{equation}
	The infinite-temperature limit, $\beta =0$, is similarly analysed by using Eq.~\eqref{Eq:E_function_before_exp} and plugging $\Gamma_{12} = \Gamma_{21} = 1/2$. In this case, the scaling is slightly different and is given by
	\begin{equation}
		\eval{\mathbb{E}}_{\Gamma_{12} = \frac{1}{2}} = \frac{1}{\sqrt{\pi N}}+o\left(N^{-1/2}\right).
	\end{equation}
	
	The other half of estimating the convergence rate stems from considering Eq.~\eqref{Eq:error_function2}. Proceeding in the same manner as before, we re-write it as
	\begin{align}
		\mathbb{F} &= \frac{\Gamma^N_{12}}{N}\sum_{k=1}^N(N-k+1)\binom{N+k-2}{N-1}\Gamma_{21}^{k-1} \nonumber\\
		&=\Gamma^N_{12}\sum_{k=0}^{N-1}\binom{N+k-1}{k}\Gamma_{21}^{k} -\frac{\Gamma^N_{12}}{N}\sum_{k=1}^{N-1} k \: \binom{N+k-1}{k}\Gamma_{21}^{k} \nonumber \\
		&= [1-I_{\Gamma_{21}}(N,N)]-\frac{\Gamma_{21}}{\Gamma_{12}}[1-I_{\Gamma_{21}}(N-1,N+1)] \nonumber\\
		&=[1-I_{\Gamma_{21}}(N,N)]-\frac{\Gamma_{21}}{\Gamma_{12}}\qty[1-I_{\Gamma_{21}}(N,N)-\frac{(\Gamma_{21}\Gamma_{12})^N}{\Gamma_{21}NB(N,N)}] \nonumber \\
		&=\frac{(\Gamma_{12}-\Gamma_{21})}{\Gamma_{12}}\qty[1-I_{4\Gamma_{12}\Gamma_{12}}\qty(N,\frac{1}{2})]+\frac{(4\Gamma_{21}\Gamma_{12})^N}{2\Gamma_{12}}\frac{1}{\sqrt{\pi N}}.\label{Eq:error_function4}
	\end{align}
	Again, we expand Eq.~\eqref{Eq:error_function4} up to second order using Eq.~\eqref{Eq:regularised-beta-expansion-12}. Then, the gamma functions appearing in the expansions are simplified using Eq.~\eqref{Eq:Gamma-function-legendre}, the remaining ones are simplified by using Eq.~\eqref{Eq:beta-function-expansion-binomial}. Finally, using Stirling's approximation [Eq.~\eqref{Eq:Stirlings-approximation}], we arrive at 
	\begin{align}\label{Eq:convergence_F}
		\mathbbm{F} = \:& \qty(\frac{\Gamma_{12}-\Gamma_{21}}{2\Gamma_{12}}) \Biggl[ 2-\frac{(4\Gamma_{21}\Gamma_{12})^N}{\sqrt{1-4\Gamma_{21}\Gamma_{12}}}\frac{1}{\sqrt{\pi N}} \nonumber \\
		&\!\!\times \qty(1-\frac{4\Gamma_{21}\Gamma_{12}}{2(N+1)(1-4\Gamma_{21}\Gamma_{12})})\Biggr] \nonumber\\
		&+\frac{(4\Gamma_{21}\Gamma_{12})^N}{2\Gamma_{12}}\left[\frac{1}{\sqrt{\pi N}}+o(N^{-3/2})\right].
	\end{align}
	For $\Gamma_{12}>1/2$ (finite temperature case), we use the fact that \mbox{$\Gamma_{12}=1-\Gamma_{21}$}, to simplify the above as
	\begin{equation}
		\mathbb{F} \simeq \underbrace{\frac{\Gamma_{12} - \Gamma_{21}}{\Gamma_{12}}}_{=1 - e^{-\beta (E_2-E_1)}} + \underbrace{\Gamma_{21}\frac{(4\Gamma_{12}\Gamma_{21})^N}{(N+1)\sqrt{\pi N}\qty(\Gamma_{12} - \Gamma_{21})^2}}_{\equiv \mathbb{G}},
	\end{equation}
	where $\simeq$ hides the $o$-terms. The infinite-temperature limit, \mbox{$\beta =0$}, is obtained by using Eq.~\eqref{Eq:error_function4} and plugging \mbox{$\Gamma_{12} = \Gamma_{21} = 1/2$}. This yields the following convergence
	\begin{equation}
		\eval{\mathbb{F}}_{\Gamma_{12} = \frac{1}{2}} = \frac{1}{\sqrt{\pi N}}+o\left(N^{-1/2}\right)
	\end{equation}
	with the expression for $\mathbb{G}$ modified to \mbox{$(\pi N)^{-1/2}$}.
	
	As a final step to prove 
	Lemma~\ref{Lem:transposition} and Theorem~\ref{Thm:beta-swap}, we calculate explicitly the state of the primary system after the protocol~$\mathcal{P}^{(12)}$. This is done by substituting the results for $\mathbb{E}$ and $\mathbb{F}$ to Eq.~\eqref{eq:finalstate}, yielding
	\begin{align}
		\v q & 
		= \mqty(1 - e^{-\beta E} & 1 \\ e^{-\beta E} & 0 )\mqty(b\\c) + \qty(b\mathbb{G} - c \mathbb{E})\mqty(1 \\ -1) \nonumber\\
		& = \Pi^\beta_{12}\v p + \qty(b\mathbb{G} - c \mathbb{E})\mqty(1 \\ -1).
	\end{align}
	Thus, the distance between $\v q$ and the target $\Pi^\beta_{12} \v p$ is given by
	\begin{align}
		\delta\qty(\Pi^\beta_{12} \v p,\,\v q) & = \abs{b\mathbb{G} - c \mathbb{E}}.
	\end{align}
	For $\beta=0$ case, the above gives
	\begin{align}
		\delta\qty(\Pi^\beta_{12} \v p,\,\v q) & = \frac{\abs{b - c }}{\sqrt{\pi N}}+o\left(N^{-1/2}\right),
	\end{align}
	whereas $\beta\neq 0$ case, it gives
	\begin{align}
		\delta\qty(\Pi^\beta_{12} \v p,\,\v q) 
		&= (4\Gamma_{12}\Gamma_{21})^{N}\Bigg[\frac{\abs{b \Gamma_{21} - c\Gamma_{12}}}{(\Gamma_{12}-\Gamma_{21})^2\sqrt{\pi N}(N+1)}\nonumber\\
		&\qquad\qquad\qquad+o(N^{-3/2})\Bigg].
	\end{align}
	These prove Lemma~\ref{Lem:transposition} and Theorem~\ref{Thm:beta-swap} after going to the notation used therein, i.e., $b \rightarrow p_1,\,c\rightarrow p_2,\,\Gamma_{12}\rightarrow\Gamma_1$ and $\Gamma_{21}\rightarrow\Gamma_2$. \qed
	
	
	\section{Strengthening Corollary \ref{corr:general_bound}}
	\label{app:bound}
	
	Corollary~\ref{corr:general_bound} deals with a very general approach to bounding the distance between the target state $\v p^{\pi}$ and its approximation obtained from $\v{p}$ via the MeMTP protocol $\mathcal{P}^{\Pi}$. However, it can be improved by taking into account the set of indices on which the permutation acts.
	
	\begin{cor} \label{corr:detailed_bound}
		Consider states $\v{p}$ and $\v q\in C_+^{TO}(\v{p})$. Then, in the infinite temperature limit, $\beta=0$, and for an $N$-dimensional memory, there exists a MeMTP protocol $\mathcal{P}$ such that
		\begin{equation}
			\P(\v{p}\otimes \v{\eta}_M) = \v{q}'\otimes \v{\eta}_M,
		\end{equation}
		with
		\begin{equation}
			\delta(\v{q}',\v{q})\leq \frac{1}{2\sqrt{\pi N}} \sum_{\substack{k,l\\ k\neq l}} \abs{p_{i_k} - p_{i_l}} + o\qty(N^{-1/2})
			=: \epsilon,
		\end{equation}
		where $\qty{i_1,\hdots,i_{d'}}\subset\qty{1,\hdots,d}$ is a subset of indices neighbouring in the $\beta$-order, $\pi_{\v p}(i_{j}) + 1 = \pi_{\v p}(i_{j+1})$, such that \mbox{$\Pi_{\v q} = \prod_{i=1} \Pi_{j_i k_i}$} with $j_i,k_i\in\qty{i_1,\hdots,i_{d'}}$.
	\end{cor}
	
	\begin{proof}
		First, we consider a target state to be an extreme point~$\v p^{\pi^*}$ such that the $\beta$-order $\Pi_{\v p^*}\equiv\Pi^*$ can be decomposed into the maximal number of $d'(d'-1)/2$ neighbour swaps on the levels $i_1$ through $i_{d'}$. Taking explicitly Eq.~\eqref{eq:Lambda_correction_op} from Theorem~\ref{Thm:permutation}, one finds that
		\begin{equation}
			\delta\left(\v{p}^{\pi^*},\v r^*\right) = \frac{1}{2\sqrt{\pi N}} \sum_{\substack{k,l\\ k\neq l}} \abs{p_{i_k} - p_{i_l}} + o\qty(N^{-1/2}),
		\end{equation}
		where for convenience we used $\mathcal{P}^{\Pi^*}(\v p\otimes\eta_M) = \v r^* \otimes \eta_M$.
		We note that the above expression in fact provides a general upper bound for any permutation $\Pi$ on the aforementioned subset of $d'$ levels -- defining $\mathcal{P}^{\Pi}(\v p\otimes\eta_M) = \v r \otimes\eta_M$ we find that
		\begin{equation}
			\delta\left(\Pi\v{p},\v r\right) \leq \epsilon.
		\end{equation}
		
		Now, there are two cases to be considered. First, take a state $\v{q}$ which is in the future of the approximation point $\v r$, $\v q \in C_+^{MTP}\qty(\v r)$, from which it follows that
		
		\begin{equation} \label{eq:exact_achiev}
			\exists\mathcal{O}\in\text{MTP}:\delta\left(\v{q},\mathcal{O}(\v r)\right) = 0.
		\end{equation}
		Otherwise, $\v q$ is not in the future cone of $\v r$. In this case, we first note that there exists a ball $B(\v r,\epsilon')\ni \Pi\v p$ with radius $\epsilon' \leq \epsilon$ with respect to $\delta(\cdot,\cdot)$. Thanks to the planarity of the boundaries $\partial C_+^{MTP}\qty(\v r)$ and $\partial C_+^{MTP}\qty(\Pi \v p)$ when restricted to a fixed $\beta$-order, we can consider the extreme case
		\begin{equation}
			\v q \in \partial C_+^{MTP}\qty(\Pi \v p) \!\Rightarrow\! \exists \v r'\!\in\! \partial C_+^{MTP}\qty(\v r):\delta\qty(\v{q},\v{r}')\!\leq\! \epsilon'\! \leq\! \epsilon
		\end{equation}
		and the same argument applies for any \mbox{$\v q \in C_+^{MTP}\qty(\Pi \v p)\backslash C_+^{MTP}\qty(\v r)$}, thus concluding the proof.
	\end{proof}
	
	The bound presented in Corollary~\ref{corr:detailed_bound} can be further improved by taking into account the possibility of dividing the set $\qty{i_1,\hdots,i_{d'}}$ into subsets that are not mixed at any step when considering the decomposition of $\Pi_{\v q}$ into neighbour transpositions. Finally, we point out that, in agreement with Eq.~\eqref{eq:exact_achiev}, there will exist such states $\v{q}$ that are attainable exactly, and moreover, their volume will increase together with the size of memory $N$.
	
	
	\section{Proof of Theorem~\ref{Thm:beta-3-cycle}}
	\label{app:beta-3-cycle}
	
	To prove Theorem~\ref{Thm:beta-3-cycle}, we will consider a composite system consisting of the main $d$-dimensional system and an $N$-dimensional memory system. Without loss of generality, we can simply assume that the main system has three levels with $i_1 = 1, i_2 = 2$ and $i_3 = 3$, and its state is described by an unnormalised probability vector \mbox{$\v p = (a, b, c)$}. The Hamiltonian is then given by \mbox{$H =\sum_{i=1}^3 E_i \ketbra{E_i}{E_i}$}, while the memory system is described by a trivial Hamiltonian $H_M = 0$ and prepared in a maximally mixed state \mbox{$\v \eta_M = (1/N, ..., 1/N)$}. The joint state of the composite system, \mbox{$\v r:= \v p \otimes \v \eta_M$}, is then given by
	\begin{equation}
		\label{Eq:entries_three_level_system}
		\v r^{(0)} \equiv \v r = \frac{1}{N}\Big(\underbrace{a, ..., a}_{N \text{  times}} | \underbrace{b, ..., b}_{N \text{  times}}| \underbrace{c, ..., c}_{N \text{  times}}\Big),
	\end{equation}
	and the joint thermal state is given by
	\begin{align}
		\v \Gamma & = \frac{1}{ZN}[e^{-\beta E_1}, ...,e^{-\beta E_1}| e^{-\beta E_2}, ..., e^{-\beta E_2}|e^{-\beta E_3}, ..., e^{-\beta E_3}] \nonumber \\
		& = \frac{1}{N}\qty(\gamma_1, ..., \gamma_1 | \gamma_2, ..., \gamma_2 | \gamma_3, ..., \gamma_3),
	\end{align}
	where $Z=\sum_{i=1}^3 e^{-\beta E_i}$.
	
	As before, the starting point consists of understanding how the joint state of the composite system changes under the action of the composite protocol $\widetilde{\mathcal{P}}^{(13)}_N \circ  \widetilde{\mathcal{P}}^{(23)}_N$, whose action is summarised in two steps:
	\begin{enumerate}
		\item Two-level thermalisation between second and third energy levels. 
		\item Two-level thermalisation between first and third energy levels.
	\end{enumerate}
	The final state $\v r^{(N)}$ is then given by
	\begin{align}
		\!\!\! \widetilde{\P}^{\Pi}(\v r) &\equiv \v r^{(N)} = \v q \otimes \v \eta_M  \nonumber \\&=\frac{1}{N}{\small\qty[a^{(N)}_1,\hdots,a^{(N)}_N \, \bigg|\, b^{(N)}_1,\hdots,b^{(N)}_N \, \bigg| \, c^{(N)}_{1},\hdots, c^{(N)}_{N}]},\label{eq:beta-3-cycle-final}
	\end{align}
	where $\Pi=\Pi_{\Gamma_{23}}\Pi_{\Gamma_{23}}$. Note that due to probability conservation, characterising the second and third entries of Eq.~\eqref{eq:beta-3-cycle-final} is sufficient.
	
	After the first protocol $\widetilde{\mathcal{P}}^{23}$, the second energy level remains ``untouched'' and, as a result, its entries are given by Eq.~\eqref{Eq:bj_entries} (with $\Gamma_{12}$ and $\Gamma_{21}$ replaced by $\Gamma_{23}$ and $\Gamma_{32}$, respectively, as defined in \eqref{Eq:rescaled-Gibbs-factor}). The other two entries are obtained in a similar way as Eqs.~\eqref{Eq:recurrence-relation-entries_bj}-\eqref{Eq:recurrence-relation-entries_cj}. The action of the protocol generates a recurrence formula that allows us to write the last entry $c^{(N)}_k$ as
	\begin{align}
		\!\!\!\! c^{(N)}_k = a\Gamma^{k-1}_{13} \sum_{i=1}^N &\binom{N+k-1-i}{k-1}\Gamma_{31}^{N+1-i}\nonumber\\&\quad \quad\quad\quad+\Gamma^N_{31}\sum_{l=0}^{N-1}\Gamma_{13}^{l}\binom{N+l-1}{l}c_{k-l},
	\end{align}
	where $c_k$ is given by
	\begin{equation}\label{ck}
		c_k = c\Gamma^N_{32} \sum_{i=0}^{k-1} \binom{N+i-1}{i}\Gamma^{i}_{23} +b\Gamma_{32}\Gamma^{k}_{23}\sum_{i=0}^{N-1} \binom{i+k}{k}\Gamma_{32}^{i}.
	\end{equation}
	
	Since, without loss of generality, we assumed that $\v p$ has $\beta$-ordering $(123)$, the proof boils down to demonstrating that $\widetilde{\P}^{\Pi}(\v r)$ sends $\v p$ to the following extreme point
	\begin{equation}
		\v p^{(321)} = \qty[a + \frac{\Gamma_{32}}{\Gamma_{23}}b - a \frac{\Gamma_{31}}{\Gamma_{13}}, c+b\left(1-\frac{\Gamma_{32}}{\Gamma_{23}}\right), \frac{\Gamma_{31}}{\Gamma_{13}}a].
	\end{equation}
	Therefore, we need to prove the following limits
	\begin{subequations}
		\begin{align}\label{Eq:first-limit-beta-cyclic-perm}
			\lim_{N \to \infty}\frac{1}{N}\sum_{i=1}^N b^{(N)}_i &= c+b \, \qty(1-\frac{\Gamma_{32}}{\Gamma_{23}}), \\ \lim_{N \to \infty}\frac{1}{N}\sum_{i=1}^N c^{(N)}_i &= a \frac{\Gamma_{31}}{\Gamma_{13}}. \label{Eq:second-limit-beta-cyclyc-perm}
		\end{align}
	\end{subequations}
	
	
	\subsection{Proof of limit (\ref{Eq:first-limit-beta-cyclic-perm}) }
	
	We start by recalling that $b^{(N)}_j$ is given by:
	\begin{equation}\label{Eq:bjn}
		b_j^{(N)}\! =c\! \frac{\Gamma_{23}}{\Gamma_{32}}\Gamma_{32}^{j}\sum_{i=0}^{N-1}\binom{j\! +\! i\! -\! 1}{i}\Gamma_{23}^{i}+b\Gamma_{23}^{N}\sum_{i=1}^{j}\binom{j\! +\! N\! -\! 1\! -\! i}{N\! -\! 1}\Gamma_{32}^{j-i}.
	\end{equation}
	Comparing Eqs. \eqref{Eq:bjn} and the right-hand side of \eqref{Eq:first-limit-beta-cyclic-perm}, we see that in order to prove Eq.~\eqref{Eq:first-limit-beta-cyclic-perm}, we need to prove the following two limits: 
	\begin{align}\label{eq:b_part_1}
		\lim_{N \to \infty}\sum_{j=1}^N\eval{\frac{b^{(N)}_j}{N}}_{\substack{b = 0\\c=1}} &=\lim_{N\rightarrow\infty}\frac{\Gamma_{23}}{\Gamma_{32}}\frac{1}{N}\sum_{j=1}^N\Gamma_{32}^j\sum_{i=0}^{N-1}\binom{j+i-1}{i}\Gamma_{23}^{i} \nonumber \\ &=1 
	\end{align}
	and
	\begin{align}\label{eq:b_part_2}
		\lim_{N \to \infty}\sum_{j=1}^N\eval{\frac{b^{(N)}_j}{N}}_{\substack{b = 1\\c=0}} &=\lim_{N\rightarrow\infty} \Gamma_{23}^{N}\frac{1}{N}\sum_{j=1}^N\Gamma_{32}^j\sum_{i=1}^j\binom{j+N-1-i}{N-1}\Gamma_{32}^{-i} \nonumber \\ &=\qty(1-\frac{\Gamma_{32}}{\Gamma_{23}}).  
	\end{align}
	
	We begin by proving Eq.~\eqref{eq:b_part_1}. First, we rewrite this expression as:
	\begin{align}
		\sum_{j=1}^N \eval{\frac{b^{(N)}_j}{N}}_{\substack{b = 1\\c=0}} &=\frac{1}{N} \frac{\Gamma_{23}}{\Gamma_{32}}\sum_{i=0}^{N-1}\Gamma_{23}^{i}\sum_{j=1}^N\Gamma_{32}^j\binom{j+i-1}{i} \nonumber\\&=  \frac{1}{N} \Gamma_{23}\sum_{i=0}^{N-1}\Gamma_{23}^{i}\sum_{j=0}^{N-1}\Gamma_{32}^j\binom{j+i}{i}.\label{eq:rewrite_b}
	\end{align} 
	We can evaluate the second sum in Eq.~\eqref{eq:rewrite_b} as follows:
	\begin{align}
		\sum_{j=0}^{N-1}\Gamma_{32}^j\binom{j+i}{i} &=  (\Gamma_{23})^{-i-1}\Bigg(1-\frac{B_{\Gamma_{32}}(N,i+1)}{B(N,i+1)}\Bigg)\nonumber\label{eq:second_sum} \nonumber \\ &=(\Gamma_{23})^{-i-1}[1-I_{\Gamma_{23}}(N,i+1)].
	\end{align}
	Thus, substituting Eq.~\eqref{eq:second_sum} into Eq.~\eqref{eq:rewrite_b}, we obtain
	\begin{align}
		\sum_{j=1}^N \eval{\frac{b^{(N)}_j}{N}}_{\substack{b = 1\\c=0}} =1-\frac{1}{N}\sum_{i=0}^{N-1}I_{\Gamma_{32}}(N,i+1).\label{eq:expr_b}
	\end{align}
	Using Eq.~\eqref{Eq:beta-regularised-limit-sum}, we conclude that the second term in Eq.~\eqref{eq:expr_b} vanishes in the limit of $N\rightarrow \infty$, and therefore
	\begin{equation}
		\lim_{n \to \infty}\sum_{j=1}^N\eval{b^{(N)}_j}_{\substack{b = 0\\c=1}} = 1,\label{eq:beta-3-cycle-first-lim}
	\end{equation}
	so that we have proved Eq.~\eqref{eq:b_part_1}. 
	
	To prove Eq.~\eqref{eq:b_part_2}, we begin by manipulating it so that we can express it in a simpler form:\! 
	\begin{align}\label{eq:simplified_b_semistep} 
		\sum_{j=1}^N \eval{\frac{b^{(N)}_j}{N}}_{\substack{b = 1\\c=0}}\!\!\!\! &= \frac{\Gamma_{23}^{N}}{N}\sum_{j=1}^N\Gamma_{32}^j\sum_{i=1}^j\binom{N+j-1-i}{N-1}\Gamma_{32}^{-i}\nonumber \\ 
		&= \frac{\Gamma_{23}^{N}}{N}\sum_{j=0}^{N-1}(N-j)\binom{N+j-1}{j}\Gamma_{32}^{j}\nonumber \\
		&=
		\Gamma_{23}^{N}\sum_{j=0}^{N-1}\binom{N\! +\! j\! -\! 1}{j}\Gamma_{32}^{j}-\frac{\Gamma_{23}^{N}}{N}\sum_{j=0}^{N-1}j\,\binom{N\! +\! j\! -\! 1}{j}\Gamma_{32}^{j}.
	\end{align} 
	Applying Eq.~\eqref{Eq:incompleta-beta-function-difference} to transform the first term of Eq.~\eqref{eq:simplified_b_semistep} into a difference of regularised beta functions, and then using its asymptotic expansion, we obtain
	\begin{equation}
		\Gamma_{23}^{N}\sum_{j=0}^{N-1}\binom{N+j-1}{j}\Gamma_{32}^{j} = I_{\Gamma_{32}}(0,N)-I_{\Gamma_{32}}(N,N)\simeq 1.
	\end{equation}
	Next, we consider the second term in Eq.~\eqref{eq:simplified_b_semistep}, which can be directly evaluated as
	\begin{align}
		-\frac{\Gamma_{23}^{N}}{N}&\sum_{j=0}^{N-1}j\binom{N+j-1}{j}\Gamma_{32}^{j} \nonumber\\
		&=-\frac{\Gamma_{32}}{\Gamma_{23}}\Gamma_{23}^{N+1}\sum_{j=-1}^{N-2}\binom{(N+1)+j-1}{j}\Gamma_{32}^j\nonumber\\
		&=-\frac{\Gamma_{32}}{\Gamma_{23}}\qty[(I_{\Gamma_{32}}(0,N+1)-I_{\Gamma_{32}}(N-1,N+1)] \nonumber\\ &\simeq -\frac{\Gamma_{32}}{\Gamma_{23}},
	\end{align}
	where in the last line we used the asymptotic expansion of $I_x(a,b)$ to approximate the difference between regularised beta functions. Collecting all the terms, we conclude that the limit is given by
	\begin{equation}
		\lim_{N \to \infty}\sum_{j=1}^N\eval{\frac{b^{(N)}_j}{N}}_{\substack{b = 1\\c=0}}  = 1-\frac{\Gamma_{32}}{\Gamma_{23}}.
	\end{equation}
	Therefore, combining the above with Eq.~\eqref{eq:beta-3-cycle-first-lim}, we get the desired limit: 
	\begin{equation} \label{eq:D_part1}
		\lim_{N \to \infty}\frac{1}{N}\sum_{j=1}^N \frac{b^{(N)}_j}{N} = c+b \, \qty(1-\frac{\Gamma_{32}}{\Gamma_{23}}) .
	\end{equation}
	
	
	\subsection{Proof of limit (\ref{Eq:second-limit-beta-cyclyc-perm})}
	
	As before, in order to prove Eq.~\eqref{Eq:second-limit-beta-cyclyc-perm}, we will also need to prove two other limits. Recall that $c^{(N)}_j$ is given by
	\begin{align}
		\!\!\!\! c^{(N)}_j = a\Gamma^{j-1}_{13} \sum_{i=1}^N &\binom{N+j-1-i}{j-1}\Gamma_{31}^{N+1-i}\nonumber\\&\quad \quad\quad\quad+\Gamma^N_{31}\sum_{l=0}^{N-1}\Gamma_{13}^{l}\binom{N+l-1}{l}c_{j-l},
	\end{align}
	with
	\begin{equation}\label{ck2}
		c_j = c\Gamma^N_{32} \sum_{i=0}^{j-1} \binom{N+i-1}{i}\Gamma^{i}_{23} +b\Gamma_{32}\Gamma^{j}_{23}\sum_{i=0}^{N-1} \binom{i+j}{j}\Gamma_{32}^{i}.
	\end{equation}
	Since $c_{j-l}$ does not depend on $a$, the problem reduces to showing that
	\begin{align}
		\lim_{N\to \infty}\frac{1}{N}\sum_{j=1}^N\eval{c^{(N)}_j}_{\substack{a = 1\\b,c=0}} &=\lim_{N\rightarrow\infty}\frac{1}{N}\sum_{j=1}^N\Gamma_{13}^{j}\sum_{i=1}^{N}\binom{N+j-i}{j}\Gamma_{31}^{N+1-i} \nonumber \\ &=\frac{\Gamma_{31}}{\Gamma_{13}},\label{eq:c_part_1}
	\end{align}
	and 
	\begin{align}
		\lim_{N\to \infty}\frac{1}{N}\sum_{j=1}^N \eval{c^{(N)}_j}_{\substack{a = 0}} &= \lim_{N\rightarrow\infty}\frac{1}{N}\sum_{j=1}^N\Gamma_{31}^{N}\sum_{l=0}^{N-1}\Gamma_{13}^{l}\binom{N+l-1}{l}c_{j-l}\nonumber\\&=0.\label{eq:ckl}
	\end{align}
	
	Let us start by proving Eq.~\eqref{eq:c_part_1}. First, we manipulate $c^{(N)}_j$ and rewrite it in terms of the incomplete beta function as follows:
	\begin{align}
		\eval{c^{(N)}_j}_{\substack{a = 1\\b,c=0}}&=\Gamma_{13}^{j-1}\sum_{i=1}^{N}\binom{N+j-1-i}{j-1}\Gamma_{31}^{N+1-i}\nonumber\\
		&=\sum_{i=1}^N\qty[(I_{\Gamma_{13}}(0,N+1-i)-I_{\Gamma_{13}}(N,N+1-i)]\nonumber\\ &= \sum_{i=1}^N I_{\Gamma_{31}}(N+1-i,N).
	\end{align}
	Using Eq.~\eqref{Eq:regularised-beta-function-sum-2}, we obtain 
	\begin{eqnarray}
		\sum_{i=1}^N I_{\Gamma_{31}}(N+1-i,N) = N\frac{\Gamma_{31}}{\Gamma_{13}}.
	\end{eqnarray}
	Thus, collecting all the terms, we get the desired limit
	\begin{equation} \label{eq:D_part2}
		\lim_{N\to \infty}\frac{1}{N}\sum_{j=1}^N\eval{c^{(N)}_j}_{\substack{a = 1\\b,c=0}}=\frac{\Gamma_{31}}{\Gamma_{13}}.
	\end{equation}
	
	The final step is to show that the remaining limit from Eq.~\eqref{eq:ckl} is zero, namely
	\begin{equation}
		\lim_{N\rightarrow\infty}\frac{1}{N}\Gamma_{31}^N\sum_{k=1}^N\sum_{l=0}^{N-1}\Gamma_{13}^lc_{k-l}=0.
	\end{equation}
	Since $c_{k-l}$ has two contributions, one needs to show that both limits go to zero. Treating each separately, we first write the first term of Eq.~\eqref{ck2} in terms of the incomplete beta function,
	\begin{equation}
		\eval{c_k}_{\substack{c = 1, b=0}} = \Gamma^N_{32} \sum_{i=0}^{k-1} \binom{N+i-1}{i}\Gamma^{i}_{23} = I_{\Gamma_{32}}(N,k) \leq 1,
	\end{equation}
	where the upper bound comes from the incomplete beta function being a CDF. Thus,
	\begin{align}
		\sum_{k=1}^N \frac{\Gamma^N_{31}}{N}\sum_{l=0}^{N-1}\Gamma_{13}^{l}&\binom{N+l-1}{l}\eval{c_{k-l}}_{\substack{c = 1, b=0}} \nonumber \\ &\leq \sum_{k=1}^N \frac{\Gamma^N_{31}}{N}\sum_{l=0}^{N-1}\Gamma_{13}^{l}\binom{N+l-1}{l} \nonumber \\ &= I_{\Gamma_{31}}(N,N), \label{eq:D_part3_1}
	\end{align}
	and this term goes to zero for $\Gamma_{31} \leq 1/2$. This can be seen from the asymptotic expansion of $I_{\Gamma_{31}}(N,N)$. 
	
	Finally, we need to show that the second term of Eq.~\eqref{ck2} is zero. To do so, we first re-write the second term as
	\begin{align}
		\!\!\!\! \sum_{k=1}^{N}\eval{c^{(N)}_k}_{\substack{b = 1\\a,c=0}}\!\! & = \frac{1}{N}\Gamma_{31}^N\sum_{l=0}^{N-1}\Gamma_{13}^l \binom{N+l-1}{l}\sum_{k=1}^N\eval{c_{k-l}}_{\substack{c = 0\\b=1}}  \nonumber \\
		& =
		\frac{1}{N}\Gamma_{31}^N\sum_{l=0}^{N-1}\Gamma_{13}^l \binom{N+l-1}{l}\sum_{k=l+1}^N\eval{c_{k-l}}_{\substack{c = 0\\b=1}} \nonumber \\
		& =
		\frac{1}{N}\Gamma_{31}^N\sum_{l=0}^{N-1}\Gamma_{13}^l \binom{N+l-1}{l}\sum_{k=1}^{N-l}\eval{c_{k}}_{\substack{c = 0\\b=1}}. 
	\end{align}  
	Notice that the above expression can be bounded by
	\begin{align}\label{Eq:bound_ck}
		\sum_{k=1}^{N}\eval{c^{(N)}_k}_{\substack{b = 1\\a,c=0}} &\leq
		\frac{1}{N}\qty(\Gamma_{31}^N\sum_{l=0}^{N-1}\Gamma_{13}^l \binom{N+l-1}{l})\qty(\sum_{k=1}^{N}\eval{c_{k}}_{\substack{c = 0\\b=1}} ), \nonumber 
	\end{align} 
	and the right-hand side of equation Eq.~\eqref{Eq:bound_ck} can be expressed in terms of the regularised beta function as follows:
	\begin{align}
		\frac{1}{N}\Gamma_{31}^N&\sum_{l=0}^{N-1}\Gamma_{13}^l \binom{N+l-1}{l}\qty(\sum_{k=1}^{N}\eval{c_{k}}_{\substack{c = 0\\b=1}} ) \nonumber\\
		& =
		\Gamma_{32}\frac{1}{N}I_{\Gamma_{31}}(N,N)\sum_{k=1}^N\Gamma^{k}_{23}\sum_{i=0}^{N-1} \binom{i+k}{k}\Gamma_{32}^{i} \nonumber\\
		& =
		\Gamma_{32}\frac{1}{N}I_{\Gamma_{31}}(N,N)\sum_{k=1}^N\Gamma^{k}_{23}\Gamma_{23}^{-1-k}\qty[1 - I_{\Gamma_{32}}(N,k+1)]\nonumber \\
		& =
		\frac{\Gamma_{32}}{\Gamma_{23}}\frac{1}{N}I_{\Gamma_{31}}(N,N)\qty(\frac{1}{N}\sum_{k=1}^N I_{\Gamma_{23}}(k+1,N)).
	\end{align}
	Now, note that the first term goes to zero when $\Gamma_{31} < 1/2$, whereas the second term is also bounded by one as
	\begin{equation}
		\frac{1}{N}\sum_{k=1}^NI_{\Gamma_{23}}(k+1,N) \leq \frac{1}{N}\sum_{k=1}^N 1 = 1.
	\end{equation}
	Since 
	\begin{equation}\label{eq:D_part3_2}
		0 \leq \lim_{N\to \infty}\sum_{k=1}^N\frac{1}{N}\eval{c^{(N)}_k}_{\substack{b = 1\\a,c=0}} \leq 0,
	\end{equation}
	we conclude that the resulting limit is zero. Therefore, combining \eqref{eq:D_part3_1} and \eqref{eq:D_part3_2}, 
	\begin{equation}\label{eq:D_part3}
		\lim_{N\rightarrow\infty}\frac{1}{N}\Gamma_{31}^N\sum_{k=1}^N\sum_{l=0}^{N-1}\Gamma_{13}^lc_{k-l}=0.
	\end{equation}
	
	Collecting together Eqs.~\eqref{eq:D_part1}, \eqref{eq:D_part2} and \eqref{eq:D_part3} we conclude that Theorem~\ref{Thm:beta-3-cycle} is proved.\qed
	
	\section{Alternative protocols realising equivalent \texorpdfstring{$\beta$}{beta}-swap approximation}\label{App:protocols}
		
		Let us revisit the truncated protocol $\tilde{\mathcal{P}}^{(ij)}$ with $d$-dimensional memory as introduced in Section~\ref{sec:bridging}, where it was defined in Eq.~\eqref{eq:trunc_protocol}. This protocol is composed of $d^2$ two-level elementary thermalisations, denoted as $T_{kl}$. Each thermalisation $T_{kl}$ can be represented as a point $(k,l)$ on a plane, and an algorithm can be represented as an arrow pointing from the previous thermalisation to the next one. For instance, we present below the diagram of $\tilde{\mathcal{P}}$ for $d = 7$:
		
		\begin{figure}[H]
			\centering
			\includegraphics{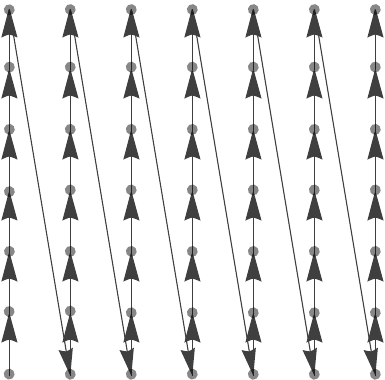}
		\end{figure}
		\noindent Visually, it is obvious that we iterate through an entire column before shifting to the next one. In other words, all the 'filled' levels are used sequentially to fill up the first 'empty' level, and the same process is repeated for all the subsequent 'empty' levels. After investigating the following two algorithms
		\begin{figure}[H]
			\centering
			\includegraphics{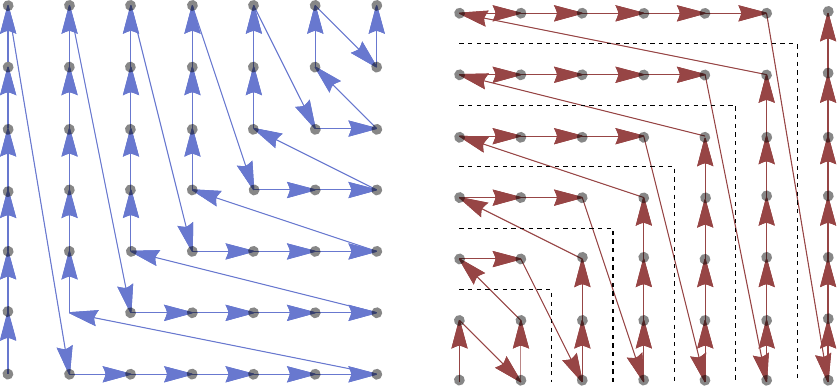}
		\end{figure}
		\noindent The generalisation of the Blue algorithm can be summarised in the following way, starting with $i = 1$:
		\begin{enumerate}
			\item Iterate through $i$-th column, starting from the first unvisited point.
			\item Iterate through $i$-th row, starting from the first unvisited point.
			\item Set $i\rightarrow i+1$ and go back to step 1 if any unvisited point remains.
		\end{enumerate}
		The Red algorithm can be most easily understood as the reverse of the Blue algorithm. Instead of decreasing the length of vertical and horizontal stretches, they are gradually increased in the Red algorithm. This allows the Red algorithm to be recursively implemented, taking into account gradually more and more levels of memory, as indicated by dashed lines. The protocol for $d$-dimensional memory is implemented by extending the $d-1$-dimensional version with an additional row and column.
		
		We furthermore investigated the Cyan family, which mimics the Blue algorithm and can be defined in the following manner:
		\begin{enumerate}
			\item Iterate through row or column, starting from the first unvisited point.
			\item Repeat step 1. until no unvisited points remain.
		\end{enumerate}
  \break
		Moreover, we considered Orange family related to the Cyan family with an analogous reversal as between Blue and Red.
		
		\begin{figure}[t]
			\centering
			\includegraphics[width=.2\textwidth]{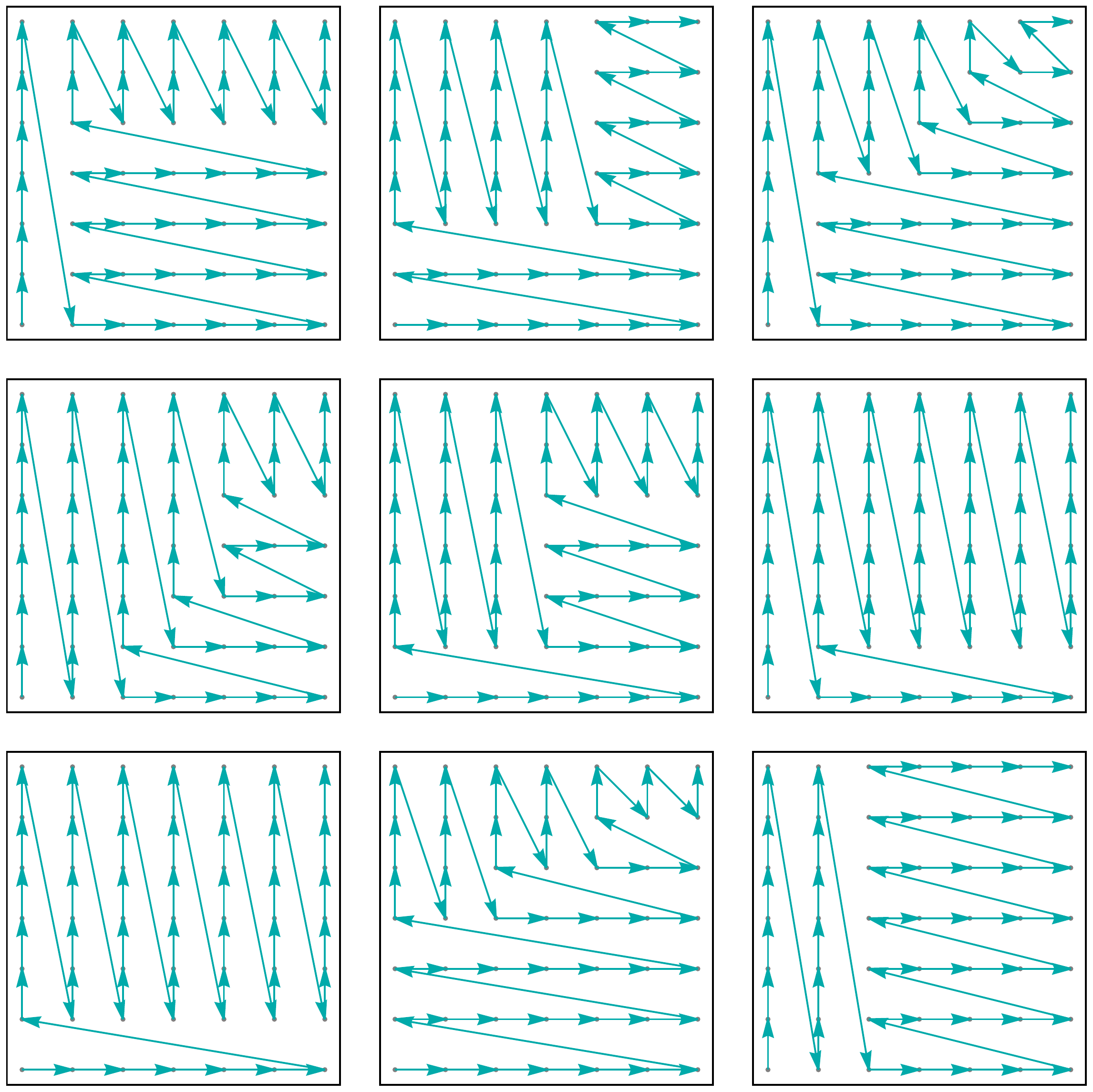}
			\hspace{.5cm}
			\includegraphics[width=.2\textwidth]{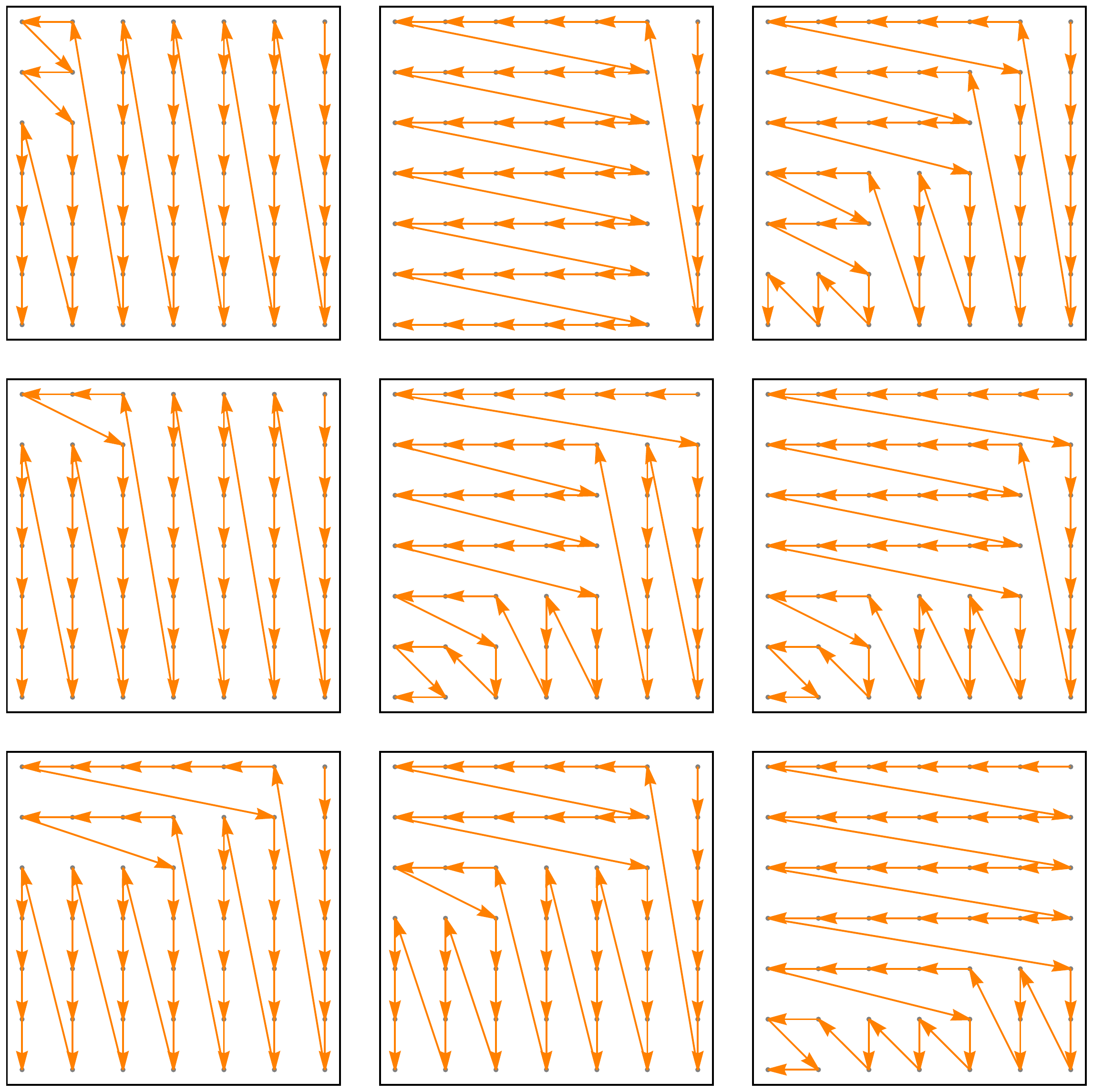}
		\end{figure}
  
		\begin{figure}[t]
			\centering
			\includegraphics{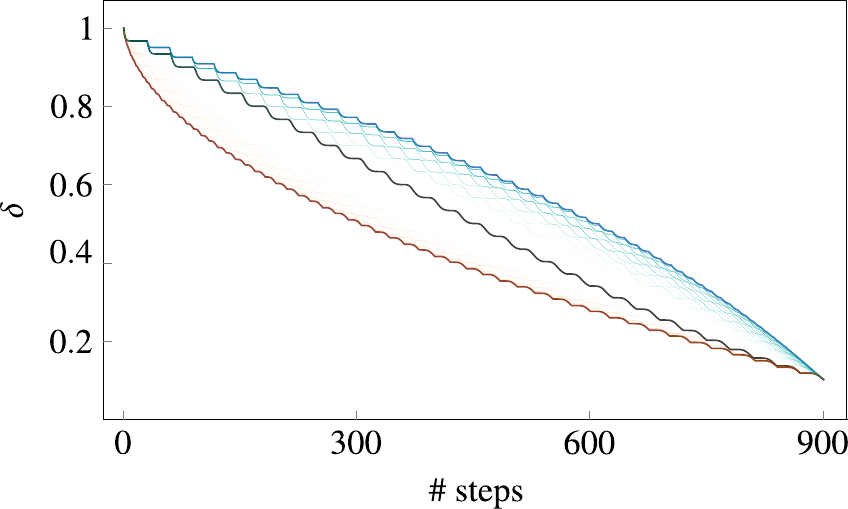}
			\caption{Convergence to the target state $\v{q} = (0,1)$ of different algorithms acting on the state $\v{p} = (1,0)$ extended by memory with dimension $d = 30$. Note that all algorithms. Note that all algorithms finish at the same value of $\norm{\v{p} - \v{q}}_1$.}
			\label{fig:enter-label}
		\end{figure}
		\noindent For these algorithms, we find the following properties we have observed from explicit implementation for a range of dimensions and inverse temperatures $\beta$, but we have not been able to prove them analytically:
		\begin{itemize}
			\item All of the aforementioned algorithms acting on an initial state $\v{p}\otimes\gamma_d$ result in the same state as $\tilde{\mathcal{P}}(\v{p}\otimes \v \gamma_d)$.
			\item Blue and Red algorithms are slowest and fastest algorithms, respective, according to the convergence to the $\beta$-swap with respect to the 1-norm.
			\item Each algorithm in the Cyan and Orange family provide slower and faster convergence than the original algorithm $\tilde{\mathcal{P}}$, respectively.
		\end{itemize}
		The statement reinforces the claims mentioned earlier and explains that the 1-norm between the intermediate states of the system and the target state (the $\beta$-swapped counterpart of $\v{p}$ with $\beta = 0$) is plotted for simplicity (see Fig.~\ref{fig:enter-label}). The memory dimension is specified as $d = 30$.
	\bibliographystyle{apsrev4-2}
	\bibliography{references}

\begin{thebibliography}{65}%
\makeatletter
\providecommand \@ifxundefined [1]{%
 \@ifx{#1\undefined}
}%
\providecommand \@ifnum [1]{%
 \ifnum #1\expandafter \@firstoftwo
 \else \expandafter \@secondoftwo
 \fi
}%
\providecommand \@ifx [1]{%
 \ifx #1\expandafter \@firstoftwo
 \else \expandafter \@secondoftwo
 \fi
}%
\providecommand \natexlab [1]{#1}%
\providecommand \enquote  [1]{``#1''}%
\providecommand \bibnamefont  [1]{#1}%
\providecommand \bibfnamefont [1]{#1}%
\providecommand \citenamefont [1]{#1}%
\providecommand \href@noop [0]{\@secondoftwo}%
\providecommand \href [0]{\begingroup \@sanitize@url \@href}%
\providecommand \@href[1]{\@@startlink{#1}\@@href}%
\providecommand \@@href[1]{\endgroup#1\@@endlink}%
\providecommand \@sanitize@url [0]{\catcode `\\12\catcode `\$12\catcode `\&12\catcode `\#12\catcode `\^12\catcode `\_12\catcode `\%12\relax}%
\providecommand \@@startlink[1]{}%
\providecommand \@@endlink[0]{}%
\providecommand \url  [0]{\begingroup\@sanitize@url \@url }%
\providecommand \@url [1]{\endgroup\@href {#1}{\urlprefix }}%
\providecommand \urlprefix  [0]{URL }%
\providecommand \Eprint [0]{\href }%
\providecommand \doibase [0]{https://doi.org/}%
\providecommand \selectlanguage [0]{\@gobble}%
\providecommand \bibinfo  [0]{\@secondoftwo}%
\providecommand \bibfield  [0]{\@secondoftwo}%
\providecommand \translation [1]{[#1]}%
\providecommand \BibitemOpen [0]{}%
\providecommand \bibitemStop [0]{}%
\providecommand \bibitemNoStop [0]{.\EOS\space}%
\providecommand \EOS [0]{\spacefactor3000\relax}%
\providecommand \BibitemShut  [1]{\csname bibitem#1\endcsname}%
\let\auto@bib@innerbib\@empty
\bibitem [{\citenamefont {Maxwell}(1872)}]{maxwell1872theory}%
  \BibitemOpen
  \bibfield  {author} {\bibinfo {author} {\bibfnamefont {J.}~\bibnamefont {Maxwell}},\ }\href {https://books.google.pl/books?id=5u84AAAAMAAJ} {\emph {\bibinfo {title} {Theory of Heat}}},\ Text-books of science\ (\bibinfo  {publisher} {Longmans, Green, and Company},\ \bibinfo {year} {1872})\BibitemShut {NoStop}%
\bibitem [{\citenamefont {Szilard}(1929)}]{Szilard1929}%
  \BibitemOpen
  \bibfield  {author} {\bibinfo {author} {\bibfnamefont {L.}~\bibnamefont {Szilard}},\ }\href {https://doi.org/10.1007/BF01341281} {\bibfield  {journal} {\bibinfo  {journal} {Zeitschrift f{\"u}r Physik}\ }\textbf {\bibinfo {volume} {53}},\ \bibinfo {pages} {840} (\bibinfo {year} {1929})}\BibitemShut {NoStop}%
\bibitem [{\citenamefont {Brillouin}(1951)}]{brillouin1951maxwell}%
  \BibitemOpen
  \bibfield  {author} {\bibinfo {author} {\bibfnamefont {L.}~\bibnamefont {Brillouin}},\ }\href {https://aip.scitation.org/doi/10.1063/1.1699951} {\bibfield  {journal} {\bibinfo  {journal} {J. Appl. Phys.}\ }\textbf {\bibinfo {volume} {22}},\ \bibinfo {pages} {334} (\bibinfo {year} {1951})}\BibitemShut {NoStop}%
\bibitem [{\citenamefont {{Landauer}}(1961)}]{Landauer1961}%
  \BibitemOpen
  \bibfield  {author} {\bibinfo {author} {\bibfnamefont {R.}~\bibnamefont {{Landauer}}},\ }\href {https://doi.org/10.1147/rd.53.0183} {\bibfield  {journal} {\bibinfo  {journal} {IBM J. Res. Dev.}\ }\textbf {\bibinfo {volume} {5}},\ \bibinfo {pages} {183} (\bibinfo {year} {1961})}\BibitemShut {NoStop}%
\bibitem [{\citenamefont {Bennett}(1982)}]{Bennett1982}%
  \BibitemOpen
  \bibfield  {author} {\bibinfo {author} {\bibfnamefont {C.~H.}\ \bibnamefont {Bennett}},\ }\href {https://doi.org/10.1007/BF02084158} {\bibfield  {journal} {\bibinfo  {journal} {Int. J. Theor. Phys.}\ }\textbf {\bibinfo {volume} {21}},\ \bibinfo {pages} {905} (\bibinfo {year} {1982})}\BibitemShut {NoStop}%
\bibitem [{\citenamefont {Maruyama}\ \emph {et~al.}(2009)\citenamefont {Maruyama}, \citenamefont {Nori},\ and\ \citenamefont {Vedral}}]{maruyama2009colloquium}%
  \BibitemOpen
  \bibfield  {author} {\bibinfo {author} {\bibfnamefont {K.}~\bibnamefont {Maruyama}}, \bibinfo {author} {\bibfnamefont {F.}~\bibnamefont {Nori}},\ and\ \bibinfo {author} {\bibfnamefont {V.}~\bibnamefont {Vedral}},\ }\href {https://doi.org/10.1103/RevModPhys.81.1} {\bibfield  {journal} {\bibinfo  {journal} {Rev. Mod. Phys.}\ }\textbf {\bibinfo {volume} {81}},\ \bibinfo {pages} {1} (\bibinfo {year} {2009})}\BibitemShut {NoStop}%
\bibitem [{\citenamefont {Seifert}(2012)}]{seifert2012stochastic}%
  \BibitemOpen
  \bibfield  {author} {\bibinfo {author} {\bibfnamefont {U.}~\bibnamefont {Seifert}},\ }\href {https://doi.org/10.1088/0034-4885/75/12/126001} {\bibfield  {journal} {\bibinfo  {journal} {Rep. Prog. Phys.}\ }\textbf {\bibinfo {volume} {75}},\ \bibinfo {pages} {126001} (\bibinfo {year} {2012})}\BibitemShut {NoStop}%
\bibitem [{\citenamefont {Sagawa}(2012)}]{sagawa2012thermodynamics}%
  \BibitemOpen
  \bibfield  {author} {\bibinfo {author} {\bibfnamefont {T.}~\bibnamefont {Sagawa}},\ }\href {https://doi.org/10.1007/978-4-431-54168-4} {\bibfield  {journal} {\bibinfo  {journal} {Prog. Theor. Exp. Phys.}\ }\textbf {\bibinfo {volume} {127}},\ \bibinfo {pages} {1} (\bibinfo {year} {2012})}\BibitemShut {NoStop}%
\bibitem [{\citenamefont {Goold}\ \emph {et~al.}(2016)\citenamefont {Goold}, \citenamefont {Huber}, \citenamefont {Riera}, \citenamefont {del Rio},\ and\ \citenamefont {Skrzypczyk}}]{Goold2016}%
  \BibitemOpen
  \bibfield  {author} {\bibinfo {author} {\bibfnamefont {J.}~\bibnamefont {Goold}}, \bibinfo {author} {\bibfnamefont {M.}~\bibnamefont {Huber}}, \bibinfo {author} {\bibfnamefont {A.}~\bibnamefont {Riera}}, \bibinfo {author} {\bibfnamefont {L.}~\bibnamefont {del Rio}},\ and\ \bibinfo {author} {\bibfnamefont {P.}~\bibnamefont {Skrzypczyk}},\ }\href {https://doi.org/10.1088/1751-8113/49/14/143001} {\bibfield  {journal} {\bibinfo  {journal} {J. Phys. A Math. Theor.}\ }\textbf {\bibinfo {volume} {49}},\ \bibinfo {pages} {143001} (\bibinfo {year} {2016})}\BibitemShut {NoStop}%
\bibitem [{\citenamefont {Binder}\ \emph {et~al.}(2018)\citenamefont {Binder}, \citenamefont {Correa}, \citenamefont {Gogolin}, \citenamefont {Anders},\ and\ \citenamefont {Adesso}}]{binder2018thermodynamics}%
  \BibitemOpen
  \bibfield  {author} {\bibinfo {author} {\bibfnamefont {F.}~\bibnamefont {Binder}}, \bibinfo {author} {\bibfnamefont {L.~A.}\ \bibnamefont {Correa}}, \bibinfo {author} {\bibfnamefont {C.}~\bibnamefont {Gogolin}}, \bibinfo {author} {\bibfnamefont {J.}~\bibnamefont {Anders}},\ and\ \bibinfo {author} {\bibfnamefont {G.}~\bibnamefont {Adesso}},\ }\href {https://doi.org/10.1007/978-3-319-99046-0} {\bibfield  {journal} {\bibinfo  {journal} {Fundam. Theor. Phys.}\ }\textbf {\bibinfo {volume} {195}},\ \bibinfo {pages} {1} (\bibinfo {year} {2018})}\BibitemShut {NoStop}%
\bibitem [{\citenamefont {Deffner}\ and\ \citenamefont {Campbell}(2019)}]{Deffner2019}%
  \BibitemOpen
  \bibfield  {author} {\bibinfo {author} {\bibfnamefont {S.}~\bibnamefont {Deffner}}\ and\ \bibinfo {author} {\bibfnamefont {S.}~\bibnamefont {Campbell}},\ }\href {https://doi.org/10.1088/2053-2571/ab21c6} {\emph {\bibinfo {title} {Quantum Thermodynamics. An introduction to the thermodynamics of quantum information}}}\ (\bibinfo  {publisher} {Morgan Claypool Publishers},\ \bibinfo {year} {2019})\BibitemShut {NoStop}%
\bibitem [{\citenamefont {Taranto}\ \emph {et~al.}(2020)\citenamefont {Taranto}, \citenamefont {Bakhshinezhad}, \citenamefont {Sch{\"u}ttelkopf}, \citenamefont {Clivaz},\ and\ \citenamefont {Huber}}]{taranto2020exponential}%
  \BibitemOpen
  \bibfield  {author} {\bibinfo {author} {\bibfnamefont {P.}~\bibnamefont {Taranto}}, \bibinfo {author} {\bibfnamefont {F.}~\bibnamefont {Bakhshinezhad}}, \bibinfo {author} {\bibfnamefont {P.}~\bibnamefont {Sch{\"u}ttelkopf}}, \bibinfo {author} {\bibfnamefont {F.}~\bibnamefont {Clivaz}},\ and\ \bibinfo {author} {\bibfnamefont {M.}~\bibnamefont {Huber}},\ }\href {https://doi.org/10.1103/PhysRevApplied.14.054005} {\bibfield  {journal} {\bibinfo  {journal} {Phys. Rev. Appl.}\ }\textbf {\bibinfo {volume} {14}},\ \bibinfo {pages} {054005} (\bibinfo {year} {2020})}\BibitemShut {NoStop}%
\bibitem [{\citenamefont {Mirkin}\ \emph {et~al.}(2019{\natexlab{a}})\citenamefont {Mirkin}, \citenamefont {Poggi},\ and\ \citenamefont {Wisniacki}}]{mirkin2019entangling}%
  \BibitemOpen
  \bibfield  {author} {\bibinfo {author} {\bibfnamefont {N.}~\bibnamefont {Mirkin}}, \bibinfo {author} {\bibfnamefont {P.}~\bibnamefont {Poggi}},\ and\ \bibinfo {author} {\bibfnamefont {D.}~\bibnamefont {Wisniacki}},\ }\href {https://doi.org/10.1103/PhysRevA.99.020301} {\bibfield  {journal} {\bibinfo  {journal} {Phys. Rev. A}\ }\textbf {\bibinfo {volume} {99}},\ \bibinfo {pages} {020301} (\bibinfo {year} {2019}{\natexlab{a}})}\BibitemShut {NoStop}%
\bibitem [{\citenamefont {Mirkin}\ \emph {et~al.}(2019{\natexlab{b}})\citenamefont {Mirkin}, \citenamefont {Poggi},\ and\ \citenamefont {Wisniacki}}]{mirkin2019information}%
  \BibitemOpen
  \bibfield  {author} {\bibinfo {author} {\bibfnamefont {N.}~\bibnamefont {Mirkin}}, \bibinfo {author} {\bibfnamefont {P.}~\bibnamefont {Poggi}},\ and\ \bibinfo {author} {\bibfnamefont {D.}~\bibnamefont {Wisniacki}},\ }\href {https://doi.org/10.1103/PhysRevA.99.062327} {\bibfield  {journal} {\bibinfo  {journal} {Phys. Rev. A}\ }\textbf {\bibinfo {volume} {99}},\ \bibinfo {pages} {062327} (\bibinfo {year} {2019}{\natexlab{b}})}\BibitemShut {NoStop}%
\bibitem [{\citenamefont {Abiuso}\ and\ \citenamefont {Giovannetti}(2019)}]{PhysRevA.99.052106}%
  \BibitemOpen
  \bibfield  {author} {\bibinfo {author} {\bibfnamefont {P.}~\bibnamefont {Abiuso}}\ and\ \bibinfo {author} {\bibfnamefont {V.}~\bibnamefont {Giovannetti}},\ }\href {https://doi.org/10.1103/PhysRevA.99.052106} {\bibfield  {journal} {\bibinfo  {journal} {Phys. Rev. A}\ }\textbf {\bibinfo {volume} {99}},\ \bibinfo {pages} {052106} (\bibinfo {year} {2019})}\BibitemShut {NoStop}%
\bibitem [{\citenamefont {Camati}\ \emph {et~al.}(2020)\citenamefont {Camati}, \citenamefont {Santos},\ and\ \citenamefont {Serra}}]{PhysRevA.102.012217}%
  \BibitemOpen
  \bibfield  {author} {\bibinfo {author} {\bibfnamefont {P.~A.}\ \bibnamefont {Camati}}, \bibinfo {author} {\bibfnamefont {J.~F.~G.}\ \bibnamefont {Santos}},\ and\ \bibinfo {author} {\bibfnamefont {R.~M.}\ \bibnamefont {Serra}},\ }\href {https://doi.org/10.1103/PhysRevA.102.012217} {\bibfield  {journal} {\bibinfo  {journal} {Phys. Rev. A}\ }\textbf {\bibinfo {volume} {102}},\ \bibinfo {pages} {012217} (\bibinfo {year} {2020})}\BibitemShut {NoStop}%
\bibitem [{\citenamefont {Ptaszy\ifmmode~\acute{n}\else \'{n}\fi{}ski}(2022)}]{PhysRevE.106.014114}%
  \BibitemOpen
  \bibfield  {author} {\bibinfo {author} {\bibfnamefont {K.}~\bibnamefont {Ptaszy\ifmmode~\acute{n}\else \'{n}\fi{}ski}},\ }\href {https://doi.org/10.1103/PhysRevE.106.014114} {\bibfield  {journal} {\bibinfo  {journal} {Phys. Rev. E}\ }\textbf {\bibinfo {volume} {106}},\ \bibinfo {pages} {014114} (\bibinfo {year} {2022})}\BibitemShut {NoStop}%
\bibitem [{\citenamefont {Breuer}\ \emph {et~al.}(2002)\citenamefont {Breuer}, \citenamefont {Petruccione} \emph {et~al.}}]{breuer2002theory}%
  \BibitemOpen
  \bibfield  {author} {\bibinfo {author} {\bibfnamefont {H.-P.}\ \bibnamefont {Breuer}}, \bibinfo {author} {\bibfnamefont {F.}~\bibnamefont {Petruccione}}, \emph {et~al.},\ }\href {https://doi.org/10.1093/acprof:oso/9780199213900.001.0001} {\emph {\bibinfo {title} {The theory of open quantum systems}}}\ (\bibinfo  {publisher} {Oxford University Press on Demand},\ \bibinfo {year} {2002})\BibitemShut {NoStop}%
\bibitem [{\citenamefont {Rivas}\ and\ \citenamefont {Huelga}(2012)}]{rivas2012open}%
  \BibitemOpen
  \bibfield  {author} {\bibinfo {author} {\bibfnamefont {A.}~\bibnamefont {Rivas}}\ and\ \bibinfo {author} {\bibfnamefont {S.~F.}\ \bibnamefont {Huelga}},\ }\href {https://doi.org/10.1007/978-3-642-23354-8} {\emph {\bibinfo {title} {Open quantum systems}}}\ (\bibinfo  {publisher} {Springer},\ \bibinfo {year} {2012})\BibitemShut {NoStop}%
\bibitem [{\citenamefont {Rivas}\ \emph {et~al.}(2014)\citenamefont {Rivas}, \citenamefont {Huelga},\ and\ \citenamefont {Plenio}}]{rivas2014quantum}%
  \BibitemOpen
  \bibfield  {author} {\bibinfo {author} {\bibfnamefont {{\'A}.}~\bibnamefont {Rivas}}, \bibinfo {author} {\bibfnamefont {S.~F.}\ \bibnamefont {Huelga}},\ and\ \bibinfo {author} {\bibfnamefont {M.~B.}\ \bibnamefont {Plenio}},\ }\href {https://doi.org/10.1088/0034-4885/77/9/094001} {\bibfield  {journal} {\bibinfo  {journal} {Rep. Prog. Phys.}\ }\textbf {\bibinfo {volume} {77}},\ \bibinfo {pages} {094001} (\bibinfo {year} {2014})}\BibitemShut {NoStop}%
\bibitem [{\citenamefont {Janzing}\ \emph {et~al.}(2000)\citenamefont {Janzing}, \citenamefont {Wocjan}, \citenamefont {Zeier}, \citenamefont {Geiss},\ and\ \citenamefont {Beth}}]{Janzing2000}%
  \BibitemOpen
  \bibfield  {author} {\bibinfo {author} {\bibfnamefont {D.}~\bibnamefont {Janzing}}, \bibinfo {author} {\bibfnamefont {P.}~\bibnamefont {Wocjan}}, \bibinfo {author} {\bibfnamefont {R.}~\bibnamefont {Zeier}}, \bibinfo {author} {\bibfnamefont {R.}~\bibnamefont {Geiss}},\ and\ \bibinfo {author} {\bibfnamefont {T.}~\bibnamefont {Beth}},\ }\href {https://doi.org/10.1023/A:1026422630734} {\bibfield  {journal} {\bibinfo  {journal} {Int. J. Theor. Phys.}\ }\textbf {\bibinfo {volume} {39}},\ \bibinfo {pages} {2717} (\bibinfo {year} {2000})}\BibitemShut {NoStop}%
\bibitem [{\citenamefont {{Horodecki}}\ and\ \citenamefont {{Oppenheim}}(2013)}]{horodecki2013fundamental}%
  \BibitemOpen
  \bibfield  {author} {\bibinfo {author} {\bibfnamefont {M.}~\bibnamefont {{Horodecki}}}\ and\ \bibinfo {author} {\bibfnamefont {J.}~\bibnamefont {{Oppenheim}}},\ }\href {https://www.nature.com/articles/ncomms3059} {\bibfield  {journal} {\bibinfo  {journal} {Nat. Commun.}\ }\textbf {\bibinfo {volume} {4}},\ \bibinfo {eid} {2059} (\bibinfo {year} {2013})}\BibitemShut {NoStop}%
\bibitem [{\citenamefont {{Brand\~ao}}\ \emph {et~al.}(2015)\citenamefont {{Brand\~ao}}, \citenamefont {{Horodecki}}, \citenamefont {{Ng}}, \citenamefont {{Oppenheim}},\ and\ \citenamefont {{Wehner}}}]{brandao2015second}%
  \BibitemOpen
  \bibfield  {author} {\bibinfo {author} {\bibfnamefont {F.~G.~S.~L.}\ \bibnamefont {{Brand\~ao}}}, \bibinfo {author} {\bibfnamefont {M.}~\bibnamefont {{Horodecki}}}, \bibinfo {author} {\bibfnamefont {N.~H.~Y.}\ \bibnamefont {{Ng}}}, \bibinfo {author} {\bibfnamefont {J.}~\bibnamefont {{Oppenheim}}},\ and\ \bibinfo {author} {\bibfnamefont {S.}~\bibnamefont {{Wehner}}},\ }\href {https://doi.org/10.1073/pnas.1411728112} {\bibfield  {journal} {\bibinfo  {journal} {Proc. Natl. Acad. Sci. U.S.A.}\ }\textbf {\bibinfo {volume} {112}},\ \bibinfo {pages} {3275} (\bibinfo {year} {2015})}\BibitemShut {NoStop}%
\bibitem [{\citenamefont {Lostaglio}(2019)}]{Lostaglio2019}%
  \BibitemOpen
  \bibfield  {author} {\bibinfo {author} {\bibfnamefont {M.}~\bibnamefont {Lostaglio}},\ }\href {https://doi.org/10.1088/1361-6633/ab46e5} {\bibfield  {journal} {\bibinfo  {journal} {Rep. Prog. Phys.}\ }\textbf {\bibinfo {volume} {82}},\ \bibinfo {pages} {114001} (\bibinfo {year} {2019})}\BibitemShut {NoStop}%
\bibitem [{\citenamefont {Lostaglio}\ and\ \citenamefont {Korzekwa}(2022)}]{lostaglio2021continuous}%
  \BibitemOpen
  \bibfield  {author} {\bibinfo {author} {\bibfnamefont {M.}~\bibnamefont {Lostaglio}}\ and\ \bibinfo {author} {\bibfnamefont {K.}~\bibnamefont {Korzekwa}},\ }\href {https://doi.org/10.1103/PhysRevA.106.012426} {\bibfield  {journal} {\bibinfo  {journal} {Phys. Rev. A}\ }\textbf {\bibinfo {volume} {106}},\ \bibinfo {pages} {012426} (\bibinfo {year} {2022})}\BibitemShut {NoStop}%
\bibitem [{\citenamefont {Korzekwa}\ and\ \citenamefont {Lostaglio}(2022)}]{korzekwa2022}%
  \BibitemOpen
  \bibfield  {author} {\bibinfo {author} {\bibfnamefont {K.}~\bibnamefont {Korzekwa}}\ and\ \bibinfo {author} {\bibfnamefont {M.}~\bibnamefont {Lostaglio}},\ }\href {https://doi.org/10.1103/PhysRevLett.129.040602} {\bibfield  {journal} {\bibinfo  {journal} {Phys. Rev. Lett.}\ }\textbf {\bibinfo {volume} {129}},\ \bibinfo {pages} {040602} (\bibinfo {year} {2022})}\BibitemShut {NoStop}%
\bibitem [{\citenamefont {Lostaglio}\ \emph {et~al.}(2018)\citenamefont {Lostaglio}, \citenamefont {Alhambra},\ and\ \citenamefont {Perry}}]{Lostaglio2018elementarythermal}%
  \BibitemOpen
  \bibfield  {author} {\bibinfo {author} {\bibfnamefont {M.}~\bibnamefont {Lostaglio}}, \bibinfo {author} {\bibfnamefont {{\'{A}}.~M.}\ \bibnamefont {Alhambra}},\ and\ \bibinfo {author} {\bibfnamefont {C.}~\bibnamefont {Perry}},\ }\href {https://doi.org/10.22331/q-2018-02-08-52} {\bibfield  {journal} {\bibinfo  {journal} {{Quantum}}\ }\textbf {\bibinfo {volume} {2}},\ \bibinfo {pages} {52} (\bibinfo {year} {2018})}\BibitemShut {NoStop}%
\bibitem [{\citenamefont {Son}\ and\ \citenamefont {Ng}(2022)}]{Jeongrak2022}%
  \BibitemOpen
  \bibfield  {author} {\bibinfo {author} {\bibfnamefont {J.}~\bibnamefont {Son}}\ and\ \bibinfo {author} {\bibfnamefont {N.~H.~Y.}\ \bibnamefont {Ng}},\ }\href {https://arxiv.org/abs/2209.15213} {\bibfield  {journal} {\bibinfo  {journal} {arXiv:2209.15213}\ } (\bibinfo {year} {2022})}\BibitemShut {NoStop}%
\bibitem [{\citenamefont {Venturelli}\ \emph {et~al.}(2013)\citenamefont {Venturelli}, \citenamefont {Fazio},\ and\ \citenamefont {Giovannetti}}]{PhysRevLett.110.256801}%
  \BibitemOpen
  \bibfield  {author} {\bibinfo {author} {\bibfnamefont {D.}~\bibnamefont {Venturelli}}, \bibinfo {author} {\bibfnamefont {R.}~\bibnamefont {Fazio}},\ and\ \bibinfo {author} {\bibfnamefont {V.}~\bibnamefont {Giovannetti}},\ }\href {https://doi.org/10.1103/PhysRevLett.110.256801} {\bibfield  {journal} {\bibinfo  {journal} {Phys. Rev. Lett.}\ }\textbf {\bibinfo {volume} {110}},\ \bibinfo {pages} {256801} (\bibinfo {year} {2013})}\BibitemShut {NoStop}%
\bibitem [{\citenamefont {Erdman}\ \emph {et~al.}(2018)\citenamefont {Erdman}, \citenamefont {Bhandari}, \citenamefont {Fazio}, \citenamefont {Pekola},\ and\ \citenamefont {Taddei}}]{PhysRevB.98.045433}%
  \BibitemOpen
  \bibfield  {author} {\bibinfo {author} {\bibfnamefont {P.~A.}\ \bibnamefont {Erdman}}, \bibinfo {author} {\bibfnamefont {B.}~\bibnamefont {Bhandari}}, \bibinfo {author} {\bibfnamefont {R.}~\bibnamefont {Fazio}}, \bibinfo {author} {\bibfnamefont {J.~P.}\ \bibnamefont {Pekola}},\ and\ \bibinfo {author} {\bibfnamefont {F.}~\bibnamefont {Taddei}},\ }\href {https://doi.org/10.1103/PhysRevB.98.045433} {\bibfield  {journal} {\bibinfo  {journal} {Phys. Rev. B}\ }\textbf {\bibinfo {volume} {98}},\ \bibinfo {pages} {045433} (\bibinfo {year} {2018})}\BibitemShut {NoStop}%
\bibitem [{\citenamefont {Hofer}\ \emph {et~al.}(2016)\citenamefont {Hofer}, \citenamefont {Perarnau-Llobet}, \citenamefont {Brask}, \citenamefont {Silva}, \citenamefont {Huber},\ and\ \citenamefont {Brunner}}]{PhysRevB.94.235420}%
  \BibitemOpen
  \bibfield  {author} {\bibinfo {author} {\bibfnamefont {P.~P.}\ \bibnamefont {Hofer}}, \bibinfo {author} {\bibfnamefont {M.}~\bibnamefont {Perarnau-Llobet}}, \bibinfo {author} {\bibfnamefont {J.~B.}\ \bibnamefont {Brask}}, \bibinfo {author} {\bibfnamefont {R.}~\bibnamefont {Silva}}, \bibinfo {author} {\bibfnamefont {M.}~\bibnamefont {Huber}},\ and\ \bibinfo {author} {\bibfnamefont {N.}~\bibnamefont {Brunner}},\ }\href {https://doi.org/10.1103/PhysRevB.94.235420} {\bibfield  {journal} {\bibinfo  {journal} {Phys. Rev. B}\ }\textbf {\bibinfo {volume} {94}},\ \bibinfo {pages} {235420} (\bibinfo {year} {2016})}\BibitemShut {NoStop}%
\bibitem [{\citenamefont {Chen}\ and\ \citenamefont {Li}(2012)}]{chen2012quantum}%
  \BibitemOpen
  \bibfield  {author} {\bibinfo {author} {\bibfnamefont {Y.-X.}\ \bibnamefont {Chen}}\ and\ \bibinfo {author} {\bibfnamefont {S.-W.}\ \bibnamefont {Li}},\ }\href {https://doi.org/10.1209/0295-5075/97/40003} {\bibfield  {journal} {\bibinfo  {journal} {EPL}\ }\textbf {\bibinfo {volume} {97}},\ \bibinfo {pages} {40003} (\bibinfo {year} {2012})}\BibitemShut {NoStop}%
\bibitem [{\citenamefont {Mitchison}\ \emph {et~al.}(2016)\citenamefont {Mitchison}, \citenamefont {Huber}, \citenamefont {Prior}, \citenamefont {Woods},\ and\ \citenamefont {Plenio}}]{mitchison2016realising}%
  \BibitemOpen
  \bibfield  {author} {\bibinfo {author} {\bibfnamefont {M.~T.}\ \bibnamefont {Mitchison}}, \bibinfo {author} {\bibfnamefont {M.}~\bibnamefont {Huber}}, \bibinfo {author} {\bibfnamefont {J.}~\bibnamefont {Prior}}, \bibinfo {author} {\bibfnamefont {M.~P.}\ \bibnamefont {Woods}},\ and\ \bibinfo {author} {\bibfnamefont {M.~B.}\ \bibnamefont {Plenio}},\ }\href {https://doi.org/10.1088/2058-9565/1/1/015001} {\bibfield  {journal} {\bibinfo  {journal} {Quantum Sci. Technol.}\ }\textbf {\bibinfo {volume} {1}},\ \bibinfo {pages} {015001} (\bibinfo {year} {2016})}\BibitemShut {NoStop}%
\bibitem [{\citenamefont {Mazurek}\ and\ \citenamefont {Horodecki}(2018)}]{mazurek2018decomposability}%
  \BibitemOpen
  \bibfield  {author} {\bibinfo {author} {\bibfnamefont {P.}~\bibnamefont {Mazurek}}\ and\ \bibinfo {author} {\bibfnamefont {M.}~\bibnamefont {Horodecki}},\ }\href {https://iopscience.iop.org/article/10.1088/1367-2630/aac057} {\bibfield  {journal} {\bibinfo  {journal} {New J. Phys.}\ }\textbf {\bibinfo {volume} {20}},\ \bibinfo {pages} {053040} (\bibinfo {year} {2018})}\BibitemShut {NoStop}%
\bibitem [{\citenamefont {Kossakowski}(1972)}]{kossakowski1972quantum}%
  \BibitemOpen
  \bibfield  {author} {\bibinfo {author} {\bibfnamefont {A.}~\bibnamefont {Kossakowski}},\ }\href {https://doi.org/10.1016/0034-4877(72)90010-9} {\bibfield  {journal} {\bibinfo  {journal} {Rep. Math. Phys.}\ }\textbf {\bibinfo {volume} {3}},\ \bibinfo {pages} {247} (\bibinfo {year} {1972})}\BibitemShut {NoStop}%
\bibitem [{\citenamefont {Gorini}\ \emph {et~al.}(1976)\citenamefont {Gorini}, \citenamefont {Kossakowski},\ and\ \citenamefont {Sudarshan}}]{gorini1976completely}%
  \BibitemOpen
  \bibfield  {author} {\bibinfo {author} {\bibfnamefont {V.}~\bibnamefont {Gorini}}, \bibinfo {author} {\bibfnamefont {A.}~\bibnamefont {Kossakowski}},\ and\ \bibinfo {author} {\bibfnamefont {E.~C.~G.}\ \bibnamefont {Sudarshan}},\ }\href {https://doi.org/10.1063/1.522979} {\bibfield  {journal} {\bibinfo  {journal} {J. Math. Phys.}\ }\textbf {\bibinfo {volume} {17}},\ \bibinfo {pages} {821} (\bibinfo {year} {1976})}\BibitemShut {NoStop}%
\bibitem [{\citenamefont {Lindblad}(1976)}]{lindblad1976generators}%
  \BibitemOpen
  \bibfield  {author} {\bibinfo {author} {\bibfnamefont {G.}~\bibnamefont {Lindblad}},\ }\href {https://doi.org/10.1007/BF01608499} {\bibfield  {journal} {\bibinfo  {journal} {Commun. Math. Phys.}\ }\textbf {\bibinfo {volume} {48}},\ \bibinfo {pages} {119} (\bibinfo {year} {1976})}\BibitemShut {NoStop}%
\bibitem [{\citenamefont {de~Oliveira~Junior}\ \emph {et~al.}(2022)\citenamefont {de~Oliveira~Junior}, \citenamefont {Czartowski}, \citenamefont {{\.Z}yczkowski},\ and\ \citenamefont {Korzekwa}}]{deoliveirajunior2022}%
  \BibitemOpen
  \bibfield  {author} {\bibinfo {author} {\bibfnamefont {A.}~\bibnamefont {de~Oliveira~Junior}}, \bibinfo {author} {\bibfnamefont {J.}~\bibnamefont {Czartowski}}, \bibinfo {author} {\bibfnamefont {K.}~\bibnamefont {{\.Z}yczkowski}},\ and\ \bibinfo {author} {\bibfnamefont {K.}~\bibnamefont {Korzekwa}},\ }\href {https://doi.org/10.1103/PhysRevE.106.064109} {\bibfield  {journal} {\bibinfo  {journal} {Phys. Rev. E}\ }\textbf {\bibinfo {volume} {106}},\ \bibinfo {pages} {064109} (\bibinfo {year} {2022})}\BibitemShut {NoStop}%
\bibitem [{\citenamefont {Jonathan}\ and\ \citenamefont {Plenio}(1999)}]{Jonathan_1999}%
  \BibitemOpen
  \bibfield  {author} {\bibinfo {author} {\bibfnamefont {D.}~\bibnamefont {Jonathan}}\ and\ \bibinfo {author} {\bibfnamefont {M.~B.}\ \bibnamefont {Plenio}},\ }\href {https://doi.org/10.1103/physrevlett.83.3566} {\bibfield  {journal} {\bibinfo  {journal} {Phys. Rev. Lett.}\ }\textbf {\bibinfo {volume} {83}},\ \bibinfo {pages} {3566} (\bibinfo {year} {1999})}\BibitemShut {NoStop}%
\bibitem [{\citenamefont {White}(1983)}]{white_1983}%
  \BibitemOpen
  \bibfield  {author} {\bibinfo {author} {\bibfnamefont {A.~T.}\ \bibnamefont {White}},\ }\href {https://doi.org/10.1017/S0305004100061053} {\bibfield  {journal} {\bibinfo  {journal} {Math. Proc. Camb. Phil. Soc.}\ }\textbf {\bibinfo {volume} {94}},\ \bibinfo {pages} {203–215} (\bibinfo {year} {1983})}\BibitemShut {NoStop}%
\bibitem [{\citenamefont {Alicki}(1979)}]{Alicki_1979}%
  \BibitemOpen
  \bibfield  {author} {\bibinfo {author} {\bibfnamefont {R.}~\bibnamefont {Alicki}},\ }\href {https://doi.org/10.1088/0305-4470/12/5/007} {\bibfield  {journal} {\bibinfo  {journal} {J. Phys. A: Math. Gen.}\ }\textbf {\bibinfo {volume} {12}},\ \bibinfo {pages} {L103} (\bibinfo {year} {1979})}\BibitemShut {NoStop}%
\bibitem [{\citenamefont {Kosloff}(2013)}]{kosloff2013quantum}%
  \BibitemOpen
  \bibfield  {author} {\bibinfo {author} {\bibfnamefont {R.}~\bibnamefont {Kosloff}},\ }\href@noop {} {\bibfield  {journal} {\bibinfo  {journal} {Entropy}\ }\textbf {\bibinfo {volume} {15}},\ \bibinfo {pages} {2100} (\bibinfo {year} {2013})}\BibitemShut {NoStop}%
\bibitem [{\citenamefont {Allahverdyan}(2014)}]{PhysRevE.90.032137}%
  \BibitemOpen
  \bibfield  {author} {\bibinfo {author} {\bibfnamefont {A.~E.}\ \bibnamefont {Allahverdyan}},\ }\href {https://doi.org/10.1103/PhysRevE.90.032137} {\bibfield  {journal} {\bibinfo  {journal} {Phys. Rev. E}\ }\textbf {\bibinfo {volume} {90}},\ \bibinfo {pages} {032137} (\bibinfo {year} {2014})}\BibitemShut {NoStop}%
\bibitem [{\citenamefont {Solinas}\ and\ \citenamefont {Gasparinetti}(2015)}]{PhysRevE.92.042150}%
  \BibitemOpen
  \bibfield  {author} {\bibinfo {author} {\bibfnamefont {P.}~\bibnamefont {Solinas}}\ and\ \bibinfo {author} {\bibfnamefont {S.}~\bibnamefont {Gasparinetti}},\ }\href {https://doi.org/10.1103/PhysRevE.92.042150} {\bibfield  {journal} {\bibinfo  {journal} {Phys. Rev. E}\ }\textbf {\bibinfo {volume} {92}},\ \bibinfo {pages} {042150} (\bibinfo {year} {2015})}\BibitemShut {NoStop}%
\bibitem [{\citenamefont {Perarnau-Llobet}\ \emph {et~al.}(2017)\citenamefont {Perarnau-Llobet}, \citenamefont {B\"aumer}, \citenamefont {Hovhannisyan}, \citenamefont {Huber},\ and\ \citenamefont {Acin}}]{PhysRevLett.118.070601}%
  \BibitemOpen
  \bibfield  {author} {\bibinfo {author} {\bibfnamefont {M.}~\bibnamefont {Perarnau-Llobet}}, \bibinfo {author} {\bibfnamefont {E.}~\bibnamefont {B\"aumer}}, \bibinfo {author} {\bibfnamefont {K.~V.}\ \bibnamefont {Hovhannisyan}}, \bibinfo {author} {\bibfnamefont {M.}~\bibnamefont {Huber}},\ and\ \bibinfo {author} {\bibfnamefont {A.}~\bibnamefont {Acin}},\ }\href {https://doi.org/10.1103/PhysRevLett.118.070601} {\bibfield  {journal} {\bibinfo  {journal} {Phys. Rev. Lett.}\ }\textbf {\bibinfo {volume} {118}},\ \bibinfo {pages} {070601} (\bibinfo {year} {2017})}\BibitemShut {NoStop}%
\bibitem [{\citenamefont {Miller}\ \emph {et~al.}(2019)\citenamefont {Miller}, \citenamefont {Scandi}, \citenamefont {Anders},\ and\ \citenamefont {Perarnau-Llobet}}]{PhysRevLett.123.230603}%
  \BibitemOpen
  \bibfield  {author} {\bibinfo {author} {\bibfnamefont {H.~J.~D.}\ \bibnamefont {Miller}}, \bibinfo {author} {\bibfnamefont {M.}~\bibnamefont {Scandi}}, \bibinfo {author} {\bibfnamefont {J.}~\bibnamefont {Anders}},\ and\ \bibinfo {author} {\bibfnamefont {M.}~\bibnamefont {Perarnau-Llobet}},\ }\href {https://doi.org/10.1103/PhysRevLett.123.230603} {\bibfield  {journal} {\bibinfo  {journal} {Phys. Rev. Lett.}\ }\textbf {\bibinfo {volume} {123}},\ \bibinfo {pages} {230603} (\bibinfo {year} {2019})}\BibitemShut {NoStop}%
\bibitem [{\citenamefont {{\AA}berg}(2013)}]{Aberg2013}%
  \BibitemOpen
  \bibfield  {author} {\bibinfo {author} {\bibfnamefont {J.}~\bibnamefont {{\AA}berg}},\ }\href {https://doi.org/10.1038/ncomms2712} {\bibfield  {journal} {\bibinfo  {journal} {Nat. Commun.}\ }\textbf {\bibinfo {volume} {4}},\ \bibinfo {pages} {1925} (\bibinfo {year} {2013})}\BibitemShut {NoStop}%
\bibitem [{\citenamefont {Skrzypczyk}\ \emph {et~al.}(2014)\citenamefont {Skrzypczyk}, \citenamefont {Short},\ and\ \citenamefont {Popescu}}]{Skrzypczyk2014}%
  \BibitemOpen
  \bibfield  {author} {\bibinfo {author} {\bibfnamefont {P.}~\bibnamefont {Skrzypczyk}}, \bibinfo {author} {\bibfnamefont {A.~J.}\ \bibnamefont {Short}},\ and\ \bibinfo {author} {\bibfnamefont {S.}~\bibnamefont {Popescu}},\ }\href {https://doi.org/10.1038/ncomms5185} {\bibfield  {journal} {\bibinfo  {journal} {Nat. Commun.}\ }\textbf {\bibinfo {volume} {5}},\ \bibinfo {pages} {4185} (\bibinfo {year} {2014})}\BibitemShut {NoStop}%
\bibitem [{\citenamefont {Brand\~ao}\ \emph {et~al.}(2013)\citenamefont {Brand\~ao}, \citenamefont {Horodecki}, \citenamefont {Oppenheim}, \citenamefont {Renes},\ and\ \citenamefont {Spekkens}}]{brandao2013resource}%
  \BibitemOpen
  \bibfield  {author} {\bibinfo {author} {\bibfnamefont {F.~G. S.~L.}\ \bibnamefont {Brand\~ao}}, \bibinfo {author} {\bibfnamefont {M.}~\bibnamefont {Horodecki}}, \bibinfo {author} {\bibfnamefont {J.}~\bibnamefont {Oppenheim}}, \bibinfo {author} {\bibfnamefont {J.~M.}\ \bibnamefont {Renes}},\ and\ \bibinfo {author} {\bibfnamefont {R.~W.}\ \bibnamefont {Spekkens}},\ }\href {https://doi.org/10.1103/PhysRevLett.111.250404} {\bibfield  {journal} {\bibinfo  {journal} {Phys. Rev. Lett.}\ }\textbf {\bibinfo {volume} {111}},\ \bibinfo {pages} {250404} (\bibinfo {year} {2013})}\BibitemShut {NoStop}%
\bibitem [{\citenamefont {Siegle}\ \emph {et~al.}(2010)\citenamefont {Siegle}, \citenamefont {Goychuk}, \citenamefont {Talkner},\ and\ \citenamefont {H{\"a}nggi}}]{siegle2010markovian}%
  \BibitemOpen
  \bibfield  {author} {\bibinfo {author} {\bibfnamefont {P.}~\bibnamefont {Siegle}}, \bibinfo {author} {\bibfnamefont {I.}~\bibnamefont {Goychuk}}, \bibinfo {author} {\bibfnamefont {P.}~\bibnamefont {Talkner}},\ and\ \bibinfo {author} {\bibfnamefont {P.}~\bibnamefont {H{\"a}nggi}},\ }\href {https://doi.org/10.1103/PhysRevE.81.011136} {\bibfield  {journal} {\bibinfo  {journal} {Phys. Rev. E}\ }\textbf {\bibinfo {volume} {81}},\ \bibinfo {pages} {011136} (\bibinfo {year} {2010})}\BibitemShut {NoStop}%
\bibitem [{\citenamefont {Budini}(2013)}]{budini2013embedding}%
  \BibitemOpen
  \bibfield  {author} {\bibinfo {author} {\bibfnamefont {A.~A.}\ \bibnamefont {Budini}},\ }\href {https://doi.org/10.1103/PhysRevA.88.032115} {\bibfield  {journal} {\bibinfo  {journal} {Phys. Rev. A}\ }\textbf {\bibinfo {volume} {88}},\ \bibinfo {pages} {032115} (\bibinfo {year} {2013})}\BibitemShut {NoStop}%
\bibitem [{\citenamefont {Campbell}\ \emph {et~al.}(2018)\citenamefont {Campbell}, \citenamefont {Ciccarello}, \citenamefont {Palma},\ and\ \citenamefont {Vacchini}}]{campbell2018system}%
  \BibitemOpen
  \bibfield  {author} {\bibinfo {author} {\bibfnamefont {S.}~\bibnamefont {Campbell}}, \bibinfo {author} {\bibfnamefont {F.}~\bibnamefont {Ciccarello}}, \bibinfo {author} {\bibfnamefont {G.~M.}\ \bibnamefont {Palma}},\ and\ \bibinfo {author} {\bibfnamefont {B.}~\bibnamefont {Vacchini}},\ }\href {https://doi.org/10.1103/PhysRevA.98.012142} {\bibfield  {journal} {\bibinfo  {journal} {Phys. Rev. A}\ }\textbf {\bibinfo {volume} {98}},\ \bibinfo {pages} {012142} (\bibinfo {year} {2018})}\BibitemShut {NoStop}%
\bibitem [{\citenamefont {vom Ende}(2023)}]{ende2023finitedimensional}%
  \BibitemOpen
  \bibfield  {author} {\bibinfo {author} {\bibfnamefont {F.}~\bibnamefont {vom Ende}},\ }\href {https://arxiv.org/abs/2306.03667} {\bibfield  {journal} {\bibinfo  {journal} {arXiv:2306.03667}\ } (\bibinfo {year} {2023})}\BibitemShut {NoStop}%
\bibitem [{\citenamefont {Owen}\ \emph {et~al.}(2019)\citenamefont {Owen}, \citenamefont {Kolchinsky},\ and\ \citenamefont {Wolpert}}]{owen2019number}%
  \BibitemOpen
  \bibfield  {author} {\bibinfo {author} {\bibfnamefont {J.~A.}\ \bibnamefont {Owen}}, \bibinfo {author} {\bibfnamefont {A.}~\bibnamefont {Kolchinsky}},\ and\ \bibinfo {author} {\bibfnamefont {D.~H.}\ \bibnamefont {Wolpert}},\ }\href {https://arxiv.org/abs/1709.00765v4} {\bibfield  {journal} {\bibinfo  {journal} {New Journal of Physics}\ }\textbf {\bibinfo {volume} {21}},\ \bibinfo {pages} {013022} (\bibinfo {year} {2019})}\BibitemShut {NoStop}%
\bibitem [{\citenamefont {Ray}\ \emph {et~al.}(2021)\citenamefont {Ray}, \citenamefont {Boyd}, \citenamefont {Wimsatt},\ and\ \citenamefont {Crutchfield}}]{PhysRevResearch.3.023164}%
  \BibitemOpen
  \bibfield  {author} {\bibinfo {author} {\bibfnamefont {K.~J.}\ \bibnamefont {Ray}}, \bibinfo {author} {\bibfnamefont {A.~B.}\ \bibnamefont {Boyd}}, \bibinfo {author} {\bibfnamefont {G.~W.}\ \bibnamefont {Wimsatt}},\ and\ \bibinfo {author} {\bibfnamefont {J.~P.}\ \bibnamefont {Crutchfield}},\ }\href {https://doi.org/10.1103/PhysRevResearch.3.023164} {\bibfield  {journal} {\bibinfo  {journal} {Phys. Rev. Res.}\ }\textbf {\bibinfo {volume} {3}},\ \bibinfo {pages} {023164} (\bibinfo {year} {2021})}\BibitemShut {NoStop}%
\bibitem [{\citenamefont {Korzekwa}\ and\ \citenamefont {Lostaglio}(2021)}]{PhysRevX.11.021019}%
  \BibitemOpen
  \bibfield  {author} {\bibinfo {author} {\bibfnamefont {K.}~\bibnamefont {Korzekwa}}\ and\ \bibinfo {author} {\bibfnamefont {M.}~\bibnamefont {Lostaglio}},\ }\href {https://doi.org/10.1103/PhysRevX.11.021019} {\bibfield  {journal} {\bibinfo  {journal} {Phys. Rev. X}\ }\textbf {\bibinfo {volume} {11}},\ \bibinfo {pages} {021019} (\bibinfo {year} {2021})}\BibitemShut {NoStop}%
\bibitem [{\citenamefont {Korzekwa}\ \emph {et~al.}(2022)\citenamefont {Korzekwa}, \citenamefont {Pucha{\l}a}, \citenamefont {Tomamichel},\ and\ \citenamefont {{\.Z}yczkowski}}]{korzekwa2022encoding}%
  \BibitemOpen
  \bibfield  {author} {\bibinfo {author} {\bibfnamefont {K.}~\bibnamefont {Korzekwa}}, \bibinfo {author} {\bibfnamefont {Z.}~\bibnamefont {Pucha{\l}a}}, \bibinfo {author} {\bibfnamefont {M.}~\bibnamefont {Tomamichel}},\ and\ \bibinfo {author} {\bibfnamefont {K.}~\bibnamefont {{\.Z}yczkowski}},\ }\href {https://doi.org/10.1109/TIT.2022.3157440} {\bibfield  {journal} {\bibinfo  {journal} {IEEE Trans. Inf. Theory}\ }\textbf {\bibinfo {volume} {68}},\ \bibinfo {pages} {4518} (\bibinfo {year} {2022})}\BibitemShut {NoStop}%
\bibitem [{\citenamefont {Wolpert}\ \emph {et~al.}(2019)\citenamefont {Wolpert}, \citenamefont {Kolchinsky},\ and\ \citenamefont {Owen}}]{wolpert2019space}%
  \BibitemOpen
  \bibfield  {author} {\bibinfo {author} {\bibfnamefont {D.~H.}\ \bibnamefont {Wolpert}}, \bibinfo {author} {\bibfnamefont {A.}~\bibnamefont {Kolchinsky}},\ and\ \bibinfo {author} {\bibfnamefont {J.~A.}\ \bibnamefont {Owen}},\ }\href {https://doi.org/10.1038/s41467-019-09542-x} {\bibfield  {journal} {\bibinfo  {journal} {Nat. Commun}\ }\textbf {\bibinfo {volume} {10}},\ \bibinfo {pages} {1} (\bibinfo {year} {2019})}\BibitemShut {NoStop}%
\bibitem [{\citenamefont {Taranto}\ \emph {et~al.}(2023)\citenamefont {Taranto}, \citenamefont {Bakhshinezhad}, \citenamefont {Bluhm}, \citenamefont {Silva}, \citenamefont {Friis}, \citenamefont {Lock}, \citenamefont {Vitagliano}, \citenamefont {Binder}, \citenamefont {Debarba}, \citenamefont {Schwarzhans}, \citenamefont {Clivaz},\ and\ \citenamefont {Huber}}]{PRXQuantum.4.010332}%
  \BibitemOpen
  \bibfield  {author} {\bibinfo {author} {\bibfnamefont {P.}~\bibnamefont {Taranto}}, \bibinfo {author} {\bibfnamefont {F.}~\bibnamefont {Bakhshinezhad}}, \bibinfo {author} {\bibfnamefont {A.}~\bibnamefont {Bluhm}}, \bibinfo {author} {\bibfnamefont {R.}~\bibnamefont {Silva}}, \bibinfo {author} {\bibfnamefont {N.}~\bibnamefont {Friis}}, \bibinfo {author} {\bibfnamefont {M.~P.}\ \bibnamefont {Lock}}, \bibinfo {author} {\bibfnamefont {G.}~\bibnamefont {Vitagliano}}, \bibinfo {author} {\bibfnamefont {F.~C.}\ \bibnamefont {Binder}}, \bibinfo {author} {\bibfnamefont {T.}~\bibnamefont {Debarba}}, \bibinfo {author} {\bibfnamefont {E.}~\bibnamefont {Schwarzhans}}, \bibinfo {author} {\bibfnamefont {F.}~\bibnamefont {Clivaz}},\ and\ \bibinfo {author} {\bibfnamefont {M.}~\bibnamefont {Huber}},\ }\href {https://doi.org/10.1103/PRXQuantum.4.010332} {\bibfield  {journal} {\bibinfo  {journal} {PRX Quantum}\ }\textbf {\bibinfo {volume} {4}},\ \bibinfo {pages} {010332} (\bibinfo {year} {2023})}\BibitemShut {NoStop}%
\bibitem [{\citenamefont {Son}\ and\ \citenamefont {Ng}(2023)}]{jeongraknelly}%
  \BibitemOpen
  \bibfield  {author} {\bibinfo {author} {\bibfnamefont {J.}~\bibnamefont {Son}}\ and\ \bibinfo {author} {\bibfnamefont {N.~H.~Y.}\ \bibnamefont {Ng}},\ }\href {https://arxiv.org/abs/2303.13020} {\bibfield  {journal} {\bibinfo  {journal} {arXiv:2303.13020}\ } (\bibinfo {year} {2023})}\BibitemShut {NoStop}%
\bibitem [{\citenamefont {Marshall}\ \emph {et~al.}(1979)\citenamefont {Marshall}, \citenamefont {Olkin},\ and\ \citenamefont {Arnold}}]{marshall1979inequalities}%
  \BibitemOpen
  \bibfield  {author} {\bibinfo {author} {\bibfnamefont {A.~W.}\ \bibnamefont {Marshall}}, \bibinfo {author} {\bibfnamefont {I.}~\bibnamefont {Olkin}},\ and\ \bibinfo {author} {\bibfnamefont {B.~C.}\ \bibnamefont {Arnold}},\ }\href {https://link.springer.com/book/10.1007/978-0-387-68276-1} {\emph {\bibinfo {title} {Inequalities: theory of majorization and its applications}}},\ Vol.\ \bibinfo {volume} {143}\ (\bibinfo  {publisher} {Springer},\ \bibinfo {year} {1979})\BibitemShut {NoStop}%
\bibitem [{\citenamefont {Ruch}\ \emph {et~al.}(1978)\citenamefont {Ruch}, \citenamefont {Schranner},\ and\ \citenamefont {Seligman}}]{Rusch1978}%
  \BibitemOpen
  \bibfield  {author} {\bibinfo {author} {\bibfnamefont {E.}~\bibnamefont {Ruch}}, \bibinfo {author} {\bibfnamefont {R.}~\bibnamefont {Schranner}},\ and\ \bibinfo {author} {\bibfnamefont {T.~H.}\ \bibnamefont {Seligman}},\ }\href {https://doi.org/10.1063/1.436364} {\bibfield  {journal} {\bibinfo  {journal} {J. Chem. Phys}\ }\textbf {\bibinfo {volume} {69}},\ \bibinfo {pages} {386} (\bibinfo {year} {1978})}\BibitemShut {NoStop}%
\bibitem [{\citenamefont {Horodecki}\ and\ \citenamefont {Oppenheim}(2013)}]{horodecki2013quantumness}%
  \BibitemOpen
  \bibfield  {author} {\bibinfo {author} {\bibfnamefont {M.}~\bibnamefont {Horodecki}}\ and\ \bibinfo {author} {\bibfnamefont {J.}~\bibnamefont {Oppenheim}},\ }\href {https://doi.org/10.1142/S0217979213450197} {\bibfield  {journal} {\bibinfo  {journal} {Int. J. Mod. Phys. B}\ }\textbf {\bibinfo {volume} {27}},\ \bibinfo {pages} {1345019} (\bibinfo {year} {2013})}\BibitemShut {NoStop}%
\bibitem [{\citenamefont {Artin}(2015)}]{artin2015gamma}%
  \BibitemOpen
  \bibfield  {author} {\bibinfo {author} {\bibfnamefont {E.}~\bibnamefont {Artin}},\ }\href@noop {} {\emph {\bibinfo {title} {The gamma function}}}\ (\bibinfo  {publisher} {Courier Dover Publications},\ \bibinfo {year} {2015})\BibitemShut {NoStop}%
\bibitem [{{\relax DLMF}()}]{NIST:DLMF}%
  \BibitemOpen
  {\relax DLMF},\ \href {https://dlmf.nist.gov/} {\bibinfo {title} {{\it NIST Digital Library of Mathematical Functions}}},\ \bibinfo {howpublished} {\url{https://dlmf.nist.gov/}, Release 1.1.9 of 2023-03-15} (\bibinfo {year} {2023}),\ \bibinfo {note} {f.~W.~J. Olver, A.~B. {Olde Daalhuis}, D.~W. Lozier, B.~I. Schneider, R.~F. Boisvert, C.~W. Clark, B.~R. Miller, B.~V. Saunders, H.~S. Cohl, and M.~A. McClain, eds.}\BibitemShut {Stop}%
\end{thebibliography}%
	
\end{document}